\documentclass[a4paper,11pt]{article}
\pdfoutput=1 

\usepackage{jheppub} 

\usepackage[T1]{fontenc} 
\usepackage{slashed}

\usepackage{indentfirst}
\usepackage{amsmath,amssymb,amsthm}
\usepackage{calligra,mathrsfs}
\usepackage{graphicx,color}
\usepackage[all,cmtip]{xy}
\usepackage{bbm}
\usepackage{hyperref}
\usepackage{mathtools}
\usepackage{verbatim}
\usepackage{stmaryrd}
\usepackage{tikz}
\usetikzlibrary{cd}
\usetikzlibrary{matrix}

\usepackage{bm}
\usepackage[margin=0.8cm,font=footnotesize]{caption}

\definecolor{lcolor}{rgb}{0.6,0.3,0.3}
\hypersetup{colorlinks=true,linkcolor=lcolor,urlcolor=lcolor,citecolor=lcolor}

\definecolor{qcolor}{rgb}{0.19,0.55,0.91}
\definecolor{hcolor}{rgb}{0.9,0.2,0.5}
\definecolor{scolor}{rgb}{0.9,0.5,0.2}

\newcommand{\tc}[2]{\textcolor{#1}{#2}}

\newcommand{\p}{\partial}
\newcommand{\Z}{\mathbbm{Z}}
\newcommand{\R}{\mathbbm{R}}
\newcommand{\C}{\mathbbm{C}}
\newcommand{\one}{\mathbbm{1}}

\newcommand{\eps}{\epsilon}

\newcommand{\step}{\vskip 3mm}

\newcommand{\ax}{{\mathfrak a}}

\newcommand{\gx}{{\mathfrak g}}

\newcommand{\mx}{{\mathfrak m}}

\newcommand{\px}{{\mathfrak p}}
\newcommand{\qx}{{\mathfrak q}}

\DeclareMathOperator{\sheafHom}{\mathscr{Hom}\text{\kern -3pt {\calligra\large om}}\,}

\DeclareMathOperator{\image}{image}

\DeclareMathOperator{\End}{End}
\DeclareMathOperator{\Hom}{Hom}

\DeclareMathOperator{\Res}{Res}

\renewcommand{\subset}{\subseteq}

\newtheoremstyle{mytheoremstyle}
{14pt}
{14pt}
{\itshape}
{20pt}
{\bfseries}
{.}
{.5em}
{}

\newtheoremstyle{myremarkstyle}
{10pt}
{10pt}
{}
{20pt}
{\itshape}
{.}
{.5em}
{}

\newtheoremstyle{myproofstyle}
{12pt}
{12pt}
{}
{20pt}
{\itshape}
{.}
{.5em}
{}

\theoremstyle{mytheoremstyle}
\newtheorem{theorem}{Theorem}
\newtheorem{definition}[theorem]{Definition}
\newtheorem{lemma}[theorem]{Lemma}
\newtheorem{corollary}[theorem]{Corollary}

\newcounter{assp}

\theoremstyle{myremarkstyle}
\newtheorem{remark}{Remark}

\theoremstyle{myproofstyle}
\newtheorem*{PROOF}{Proof}

\usetikzlibrary{positioning}
\usetikzlibrary{arrows}
\usetikzlibrary{backgrounds}
\usetikzlibrary{matrix}

\usepackage{multirow}

\DeclareMathOperator{\Tor}{Tor}
\DeclareMathOperator{\Der}{Der}
\DeclareMathOperator{\Coder}{Coder}

\newcommand{\pd}{\nrightarrow}
\newcommand{\pdd}{\mathbin{{\let\scriptstyle\textstyle\substack{%
                    \nleftarrow \\[-1em] \nrightarrow}}}}
\newcommand{\ndd}{\mathbin{{\let\scriptstyle\textstyle\substack{%
                    \leftarrow \\[-1em] \rightarrow}}}}
\newcommand{\xxx}[3]{\rule{#1 mm}{0pt}$\mathllap{#2}#3$}

\newcommand{\yma}{\tc{scolor}{M}}

\definecolor{cobalt}{rgb}{0.0, 0.28, 0.67}
\definecolor{burgundy}{rgb}{0.5, 0.0, 0.13}
\definecolor{chromeyellow}{rgb}{1.0, 0.65, 0.0}
\definecolor{cinnabar}{rgb}{0.89, 0.26, 0.2}
\definecolor{crimson}{rgb}{0.86, 0.08, 0.24}
\definecolor{jade}{rgb}{0.0, 0.66, 0.42}
\definecolor{lightseagreen}{rgb}{0.13, 0.7, 0.67}
\usepackage{fancyvrb} 

\newcommand{\cx}{\tc{lightseagreen}{\mathfrak{c}}}

\usepackage{moresize}
\DefineVerbatimEnvironment{verbatim}{Verbatim}{xleftmargin=.3in,fontsize=\ssmall,formatcom=\color{cobalt}}
\DefineVerbatimEnvironment{vv}{BVerbatim}{fontsize=\ssmall,formatcom=\color{cobalt}}

\title{\boldmath A homotopy BV algebra for Yang-Mills\\
                 and color-kinematics}

\author{Michael Reiterer}

\affiliation{The Hebrew University of Jerusalem, Israel}

\emailAdd{michael.reiterer@protonmail.com}

\abstract{%
Yang-Mills gauge theory on Minkowski space
supports a Batalin-Vilkovisky-infinity algebra structure,
all whose operations are local.
To make this work,
the axioms for a BV-infinity algebra
are deformed by a quadratic element,
here the Minkowski wave operator.
This homotopy structure implies BCJ/color-kinematics duality;
a cobar construction yields a strict algebraic structure
whose Feynman expansion for Yang-Mills tree amplitudes
complies with the duality.
It comes with a `syntactic kinematic algebra'.
}
\begin{document} 
\maketitle
\flushbottom

\newcommand{\operad}{\tc{hcolor}{\textnormal{P}}}
\newcommand{\operadinf}{\tc{hcolor}{\textnormal{P}_\infty}}
\newcommand{\gerst}{\tc{hcolor}{\textnormal{G}}}

\newcommand{\ainf}{\tc{hcolor}{\textnormal{A}_\infty}}
\newcommand{\linf}{\tc{hcolor}{\textnormal{L}_\infty}}
\newcommand{\cinf}{\tc{hcolor}{\textnormal{C}_\infty}}
\newcommand{\ginf}{\tc{hcolor}{\textnormal{G}_\infty}}
\newcommand{\ginfq}[1]{\tc{hcolor}{\textnormal{G}_\infty^{#1}}}
\newcommand{\bvinf}{\tc{hcolor}{\textnormal{BV}_\infty}}
\newcommand{\bv}[1]{\tc{hcolor}{\textnormal{BV}_\infty^{#1}}}
\newcommand{\naivecurv}{\tc{scolor}{N}}
\newcommand{\curv}{\tc{hcolor}{K}}
\newcommand{\pcurv}{\tc{qcolor}{\mathbf{P}}}
\newcommand{\N}{\mathbbm{N}}
\newcommand{\comp}[1]{\tc{scolor}{\overline{#1}}}

\renewcommand{\tt}{\tc{scolor}{\alpha}}
\newcommand{\qq}{\tc{qcolor}{\beta}}
\newcommand{\hh}{\tc{hcolor}{\gamma}}
\newcommand{\uu}{\tc{hcolor}{\kappa}}

\newcommand{\Qx}{\tc{hcolor}{Q}}
\newcommand{\Hopf}{\tc{qcolor}{H}}
\newcommand{\fil}[1]{\tc{qcolor}{\mathcal{F}_{#1}}}
\newcommand{\ddr}{d_{\textnormal{dR}}}
\newcommand{\deldr}{\delta_{\textnormal{dR}}}
\newcommand{\ldr}{\Delta_{\textnormal{dR}}}
\newcommand{\apw}{\ax_{\textnormal{w}}}
\newcommand{\Apw}{A_{\textnormal{w}}}
\newcommand{\apx}[1]{\ax_{\textnormal{w},#1}}
\newcommand{\Apx}[1]{A_{\textnormal{w},#1}}

\newcommand{\prop}{h_{\sharp}}
\newcommand{\Aprop}{h_{A\sharp}}

\newcommand{\imap}{\tc{qcolor}{I}}
\newcommand{\pmap}{\tc{qcolor}{P}}
\newcommand{\zmap}{\tc{qcolor}{Z}}

\newcommand{\apr}{\tc{scolor}{\Theta}}

\section{Introduction}\label{sec:intro}

We construct a Batalin-Vilkovisky-infinity
algebra \cite{tt,gtv}
for Yang-Mills gauge theory on Minkowski spacetime.
This implies the Bern-Carrasco-Johansson relations \cite{bcj,tz}
for tree amplitudes.
Furthermore, a cobar construction yields a quasi-isomorphic strict algebra
whose Feynman expansion
complies with BCJ/color-kinematics duality \cite{bcj,bdhk,bccjr}\footnote{%
This posits that YM tree amplitudes may be written as a sum of cubic trees
where kinematic numerators satisfy `identities in one-to-one correspondence
to the Jacobi identities obeyed by color factors' \cite{bdhk}.}.
Similar structures have been considered before in gauge theory,
see Zeitlin's \cite{zeitlin,zeitlin2,zeitlin3}\footnote{%
Zeitlin considers, among others,
homotopy BV/Gerstenhaber algebras
and generalized Yang-Mills equations associated 
to certain vertex operator algebras,
via a quasi-classical limit of Lian-Zuckerman algebras.},
but not a structure suitable for BCJ, for which
one needs to deform the BV-infinity axioms.
\step
One can formulate Yang-Mills in terms of a
differential graded commutative algebra.
To define this dgca \cite{zeitlin,costello,nr},
let $\Omega = \Omega(\R^4)$ be the smooth complex-valued differential forms
on Minkowski space;
$\ddr$ the de Rham differential; $\Omega^2 = \Omega^2_+ \oplus \Omega^2_-$
the invariant decomposition\footnote{%
Decomposition into eigenspaces of the Hodge star operator
for the Minkowski metric.
},
and $\pi_+ : \Omega^2 \to \Omega^2_+$ the associated projection.
Consider
\begin{equation}\label{eq:cpx}
\vcenter{\hbox{%
\begin{tikzpicture}[node distance = 11mm and 24mm, auto]
  \node (a0) {$\Omega^0$};
  \node (a1) [right=of a0] {$\Omega^1$};
  \node (a2) [right=of a1] {$\Omega^2_+$};
  \node (b2) [below=of a1] {$\Omega^2_+$};
  \node (b3) [right=of b2] {$\Omega^3$};
  \node (b4) [right=of b3] {$\Omega^4$};
  \draw[->] (a0) to node {$\ddr$} (a1);
  \draw[->] (a1) to node {$\pi_+\ddr$} (a2);
  \draw[->] (b2) to node {$\one$} (a2);
  \draw[->] (b2) to node {$\ddr$} (b3);
  \draw[->] (b3) to node {$\ddr$} (b4);
\end{tikzpicture}}}
\end{equation}
Let $\ax = \ax^0 \oplus \ax^1 \oplus \ax^2 \oplus \ax^3$
be the four columns in \eqref{eq:cpx} respectively,
and $d : \ax^i \to \ax^{i+1}$ the differential given by all arrows in \eqref{eq:cpx}.
The gca product $\ax^i \otimes \ax^j \to \ax^{i+j}$
is given by wedging forms:
the product of an element in the upper row and one in the upper (lower) row is placed
into the upper (lower) row; the product of two elements in the lower row is zero;
see Section \ref{sec:defs}.
Then $\ax$ is a dgca: $d$ is a differential;
the gca product is associative and graded commutative;
$d$ is a derivation\footnote{%
Note that such a dgca $\ax$ can be defined on
any 4-manifold with a metric or conformal metric.}.
Given a finite-dimensional
`color' Lie algebra $\cx$,
the classical Yang-Mills equation of motion
is the Maurer-Cartan equation \cite{mm} in the 
differential graded Lie algebra $\ax \otimes \cx$\footnote{%
The unknown in the Maurer-Cartan equation is an element of $\ax^1 \otimes \cx
= (\Omega^1 \otimes \cx) \oplus (\Omega^2_+ \otimes \cx)$,
consisting of a $\cx$-valued 1-form, the vector potential;
and the $\Omega^2_+$-half of a $\cx$-valued 2-form, the field strength.}.
The focus in this paper is on $\ax$ which is the `kinematic' part.
\step
Denote by $h: \ax^i \to \ax^{i-1}$
any map\footnote{%
All maps are understood to be $\C$-linear.} of the form
\begin{equation}\label{eq:cpxhhh}
\vcenter{\hbox{%
\begin{tikzpicture}[node distance = 11mm and 24mm, auto]
  \node (a0) {$\Omega^0$};
  \node (a1) [right=of a0] {$\Omega^1$};
  \node (a2) [right=of a1] {$\Omega^2_+$};
  \node (b2) [below=of a1] {$\Omega^2_+$};
  \node (b3) [right=of b2] {$\Omega^3$};
  \node (b4) [right=of b3] {$\Omega^4$};
  \draw[->] (a1) to node {} (a0);
  \draw[->] (a2) to node {} (a1);
  \draw[->] (b2) to node {} (a0);
  \draw[->] (b3) to node {} (a1);
  \draw[->] (b4) to node {} (a2);
  \draw[->] (b3) to node {} (b2);
  \draw[->] (b4) to node {} (b3);
\end{tikzpicture}}}
\end{equation}
subject to the following requirements:
all arrows are invariant under translations on $\R^4$;
the four horizontal arrows are homogeneous
first order partial differential operators;
the three diagonal arrows are $C^\infty(\R^4)$-linear
meaning they are without derivatives; and
\begin{equation}\label{eq:hxs}
\begin{aligned}
h^2 & \;=\; 0\\
dh + hd & \;=\; \Box \one
\end{aligned}
\end{equation}
where $\Box = \eta^{\mu\nu}\p_\mu \p_\nu$
is the wave operator for the Minkowski metric $\eta$ on $\R^4$.
Such an $h$ exists but it is not unique; see Section \ref{sec:defs}.
It need not be Lorentz invariant.
\step
By definition, the failure of $h$ to be second order
on $\ax$ is measured by\footnote{%
The order is in the graded sense,
as in \cite[Section 1]{kos},
not the order as a partial differential operator.
}\textsuperscript{,}\footnote{%
The $-h(1)xyz$ term is absent because $h(1)=0$, since $\ax$ is zero in degree $-1$.}
\begin{multline}\label{eq:sdef}
S(x,y,z) \;=\; 
h(xyz) - h(xy)z - (-1)^x xh(yz) - (-1)^{(x+1)y} yh(xz)\\
+ h(x)yz + (-1)^x xh(y)z + (-1)^{x+y} xyh(z)
\end{multline}
a graded symmetric map $S : \ax^{\otimes 3} \to \ax$. Here and in similar instances below,
$x,y,z,\ldots$ are homogeneous
elements of $\ax$, and juxtaposition denotes the gca product in $\ax$.

\begin{theorem}[Second order up to homotopy] \label{theorem:2ax}
Every such $h$ is second order on $\ax$ up to homotopy.
Namely, there exists a map
$\theta_3 : \ax^i \otimes \ax^j \otimes \ax^k \to \ax^{i+j+k-2}$ such that
\begin{equation}\label{eq:exact}
S(x,y,z) \;=\;
 d \theta_3(x,y,z) - \theta_3(dx,y,z) - (-1)^x \theta_3(x,dy,z) - (-1)^{x+y} \theta_3(x,y,dz)
\end{equation}
The homotopy $\theta_3$ is graded symmetric;
it is $C^\infty(\R^4)$-linear in each argument; it is invariant under translations on $\R^4$;
and it satisfies the identities\footnote{%
The first identity in \eqref{eq:identity4}
implies the second using graded symmetry \eqref{eq:gras}.
}
\begin{equation}\label{eq:identity4}
\begin{aligned}
x\theta_3(y,u,v) - (-1)^{xy} y \theta_3(x,u,v)
+ \theta_3(x,yu,v) - (-1)^{xy} \theta_3(y,xu,v) & = 0\\
\theta_3(xy,u,v)-\theta_3(x,yu,v)+\theta_3(x,y,uv)-(-1)^{x(y+u+v)} \theta_3(y,u,vx) & = 0
\end{aligned}
\end{equation}
\end{theorem}
\begin{corollary} \label{corollary:bcjr}
The BCJ relations in the form given in \cite{tz} hold.
\end{corollary}
Theorem \ref{theorem:2ax} suggests that the dgca operations
together with $h$ and $\theta_3$ are
but a fragment of a more comprehensive algebraic structure
that subsumes \eqref{eq:hxs}, \eqref{eq:exact}, \eqref{eq:identity4} into its axioms.
Specifically, $\bvinf$- and $\ginf$-algebras \cite{tt,gtv}
are used in deformation theory \cite{t0,t,tt,w},
and Zeitlin \cite{zeitlin,zeitlin2,zeitlin3}
considered them in the context of generalized Yang-Mills.
The $\bvinf$ axioms require that $dh+hd$ be zero,
which here fails due to $\Box$ in \eqref{eq:hxs}.
This prompts one to define a
deformed notion of a
$\bvinf$-algebra, see Section \ref{subsec:mod},
to be called a $\bv{\Box}$-algebra\footnote{%
Given a graded vector space $V$
that is also a module over a cocommutative Hopf algebra $\Hopf$,
and a central quadratic element $\Qx \in \Hopf$,
the dgLa controlling $\bvinf$-algebras is restricted to
$\Hopf$-equivariant elements, then deformed by $\Qx$
to a curved dgLa that, by definition, controls $\bv{\Qx}$-algebras.
Every $\bv{\Qx}$-algebra has a differential and contains a $\cinf$-algebra,
but not necessarily an $\linf$-algebra.
This notion is not specific to Yang-Mills.
For Yang-Mills on Minkowski spacetime,
set $V = \ax[2]$ and $\Hopf = \C[\p_0,\ldots,\p_3]$ and $\Qx = \Box$.}.
\begin{theorem}[Main] \label{theorem:mainh}
For every such $h$, there exists a $\bv{\Box}$-algebra structure on $\ax$:
\begin{itemize}
\item Its $\cinf$ part is the dgca structure described above.
\item Its unary degree $-1$ operator is $h$.
\item All operations are
invariant under translations and local\footnote{%
All operations belong to the classes of local operators
in Definition \ref{def:locop};
they are sums of translation-invariant $C^\infty(\R^4)$-multilinear operations
composed with at most one partial derivative.}.
\end{itemize}
In the notation of \cite{gtv},
all structure maps except $m^0_{1,\ldots,1}$, $m^0_{1,\ldots,1,2}$, $m_1^1$ vanish.
The differential is $m^0_1$; the gca product is $m^0_2$;
the homotopy $\theta_3$ is $m^0_{1,2}$; and $h$ is $m_1^1$; see \eqref{eq:munux}.
\end{theorem}
The proof in Section \ref{sec:proofmain}
gives a definite choice and construction of all operations of this $\bv{\Box}$-algebra;
natural conditions based on the grading and the structure of $\R^4$,
cf.~\eqref{eq:ztu2},
fix all operations uniquely, given $h$\footnote{%
The construction does not reveal whether all but finitely many operations are identically zero.
Computer calculations, using the code in Appendix \ref{app:code},
show that there are nonzero operations of arity at least $6$.
}.
This structure 
is a superstructure of the Yang-Mills dgca
that also encompasses the `propagator numerator' $h$.
\step
This has an application to BCJ/color-kinematics duality \cite{bcj,bdhk,bccjr}.
Let $\apw \hookrightarrow \ax$ be the subspace
of finite linear combinations of plane waves\footnote{%
By definition, an element $y \in \ax$ is a plane wave with momentum $k \in \C^4$
if and only if $\tau_a y = e^{c\,k_\mu a^\mu} y$ for all $a \in \R^4$,
where $\tau_a$ denotes translation by $a$.
The constant $c \neq 0$ is a matter of convention.} with momenta in $\C^4$.
Then the $\linf$ minimal model Feynman expansion for $\apw \otimes \cx$
using $h$ as the propagator numerator gives the Yang-Mills amplitudes
as a sum of cubic trees \cite{nr}, but this does not yet comply with the duality.
However,
the upgrade to a
$\bv{\Box}$-algebra structure in Theorem \ref{theorem:mainh}
does naturally lead to a Feynman expansion
that complies with the duality\footnote{%
With `local kinematic numerators'.}.
\step
Namely,
once Yang-Mills is cast as a $\bv{\Box}$-algebra,
a generic cobar construction
yields a quasi-isomorphic strict structure\footnote{%
Strict means only unary and binary operations;
$\ax$ is not strict as a $\bv{\Box}$-algebra because $\theta_3 \neq 0$.
}, where the duality is more immediate.
This construction in Section \ref{sec:rect}
is an adaptation of a standard one \cite{dtt,gtv,lv}.
The next theorem states this result
only as applied to $\ax$ specifically.
Denote by $\Hopf = \C[\p_0,\ldots,\p_3]$
the Hopf algebra of constant coefficient differential operators on $\R^4$,
and note that $\ax$ is an $\Hopf$-module.

\begin{theorem}[Strict structure via a cobar construction]\label{theorem:rect}
Let $A$ be the free Gerstenhaber algebra
of the cofree Gerstenhaber coalgebra of $\ax$\footnote{%
See Lemma \ref{lemma:rcc} for details, including degree shifts;
set $V = \ax[2]$ there, equivalently $V[-2] = \ax$.
}.
Then a cobar construction
applied to the $\bv{\Box}$-algebra $\ax$
yields
$\Hopf$-equivariant $d_A \in \End^1(A)$ and $h_A \in \End^{-1}(A)$ 
such that:
\begin{itemize}
\item
$(A,d_A)$ is a dgca; $(A,h_A)$ is a BV algebra;
so $h_A$ is genuinely second order; and
\begin{equation}\label{eq:strdaha}
d_A^2 = 0
\qquad\quad
d_Ah_A + h_Ad_A = \Box \one
\qquad\quad
h_A^2 = 0
\end{equation}
where $\Box$ is in the sense of the induced $\Hopf$-module structure on $A$.
\item
The complexes $(\ax,d)$ and $(A,d_A)$ are quasi-isomorphic.
More concretely, the canonical degree zero map $\imap: \ax \hookrightarrow A$
is part of an $\Hopf$-equivariant contraction\footnote{%
See Lemma \ref{lemma:rcc} for details, including the equations defining a contraction.}:
\begin{equation}\label{eq:uu1}
\vcenter{\hbox{%
\begin{tikzpicture}
  \draw[anchor=east] (0,0) node {$\ax$};
  \draw[anchor=west] (2,0) node {$A$};
  \draw[->] (0.1,0.2) to[bend left=15] node[midway,anchor=south] {$\imap$} (1.9,0.2);
  \draw[->] (1.9,-0.2) to[bend left=15] node[midway,anchor=north] {$\pmap$} (0.1,-0.2);
  \draw[->] (2.7,0.2) to[out=40,in=-40,looseness=8] node[midway,anchor=west] {$\zmap$} (2.7,-0.2);
\end{tikzpicture}}}
\end{equation}
The maps $\imap$ and $\pmap$ are actually chain maps in two ways:
\begin{equation}\label{eq:uu1x}
d_A \imap = \imap d
\qquad
d \pmap = \pmap d_A
\qquad
h_A \imap = \imap h
\qquad
h \pmap = \pmap h_A
\end{equation}
\item
Denoting by $\apr : A \otimes A \to A$
and $\theta_2 : \ax \otimes \ax \to \ax$ the respective gca products:
\begin{equation}\label{eq:uu2}
\pmap \circ \apr \circ \imap^{\otimes 2} = \theta_2
\end{equation}
\end{itemize}
\end{theorem}
The cobar construction is functorial.
Applied to $\apw \hookrightarrow \ax$
it yields a structure $\Apw \hookrightarrow A$;
both $\apw$ and $\Apw$ are $\C^4$-graded by momentum $k \simeq \p$;
$\Hopf$-equivariance implies momentum conservation.
Abusing notation, denote by $\imap$ also its restriction to $\Apw$ etc.
\begin{corollary}[BCJ duality] \label{corollary:bcjd}
Use $\pd$ to indicate partial maps defined for a dense subset of momenta.
Consider the differential graded Lie algebra $\Apw \otimes \cx$,
with homology $H(\apw \otimes \cx)$.
The standard $\linf$ minimal model Feynman expansion gives for every $n \geq 2$ a map
\begin{equation}\label{eq:mnmn}
\yma_n : H^1(\apw \otimes \cx)^{\otimes n} \pd H^2(\apw \otimes \cx)
\end{equation}
as a sum over all non-planar cubic trees with one output and inputs $1\ldots n$
using:
\textnormal{\[
\begin{tabular}{c | c l}
tree element & \multicolumn{2}{c}{operation\rule{2mm}{0pt}}\\
\hline
vertex & gca product in $\Apw$, Lie bracket in $\cx$ & \xxx{21}{(\Apw^1 \otimes \cx)^{\otimes 2}}{\to \Apw^2 \otimes \cx}\\
internal line & \xxx{8}{h_A/k^2}{\otimes \one} & \xxx{21}{\Apw^2 \otimes \cx}{\pd \Apw^1 \otimes \cx}\\
input & \xxx{8}{\imap i}{\otimes \one} & \xxx{21}{H^1(\apw \otimes \cx)}{\pd \Apw^1 \otimes \cx}\\
the output & \xxx{8}{p \pmap}{\otimes \one} & \xxx{21}{\Apw^2 \otimes \cx}{\pd H^2(\apw \otimes \cx)}
\end{tabular}
\]}%
where $k^2 \simeq \Box$ is the momentum squared;
$i: H(\apw) \pdd \apw : p$ are any momentum-conserving
quasi-isomorphisms that are nonzero only along $k^2 = 0$, undefined at $k=0$.
This expansion complies with 
BCJ/color-kinematics duality\footnote{%
With `local kinematic numerators'.};
and the $\yma_n$ are the Yang-Mills amplitudes.
\end{corollary}
The definition of the duality and the proof are in Section \ref{sec:fbcj}.
\step
It has been suggested that a rationale for the duality
would be provided by a hypothetical `kinematic Lie algebra',
see \cite{moc,cs,fk,cjtw} for results in this direction.
Theorem \ref{theorem:rect} 
provides such a Lie algebra, namely the Lie bracket
\begin{equation}\label{eq:lie}
   [-,-]\;:\;\Apw^1 \otimes \Apw^1 \to \Apw^1
\end{equation}
Indeed, the Yang-Mills amplitudes are also obtained using these Feynman rules:
\[
\begin{tabular}{c | c l}
tree element & \multicolumn{2}{c}{operation\rule{2mm}{0pt}}\\
\hline
vertex & \eqref{eq:lie} in $\Apw^1$, Lie bracket in $\cx$ & \xxx{21}{(\Apw^1 \otimes \cx)^{\otimes 2}}{\to \Apw^1 \otimes \cx}\\
internal line & \xxx{8}{1/k^2}{\otimes \one} & \xxx{21}{\Apw^1 \otimes \cx}{\pd \Apw^1 \otimes \cx}\\
input & \xxx{8}{\imap i}{\otimes \one} & \xxx{21}{H^1(\apw \otimes \cx)}{\pd \Apw^1 \otimes \cx}\\
the output & \xxx{8}{p \pmap}{\otimes \one} & \xxx{21}{\Apw^1 \otimes \cx}{\pd H^1(\apw \otimes \cx)}
\end{tabular}
\]
There is a paradox here.
The bracket \eqref{eq:lie} in the free Gerstenhaber algebra
is a syntactic construction,
given by sticking together words modulo some relations, see \eqref{eq:gax}.
How can this give the Yang-Mills amplitudes?
The resolution lies in $\pmap$ which
translates the syntactic expression in $\Apw^1$ back to $\apw^1$,
in a nontrivial way.
The maps $H^1(\apw \otimes \cx)^{\otimes n} \pd H^1(\apw \otimes \cx)$
so obtained are equivalent to the $\yma_n$, as discussed in Section \ref{sec:ska}.
The Jacobi identity for the `syntactic kinematic algebra'
\eqref{eq:lie} implies BCJ duality.
\step
The perspective on BCJ duality provided
by Theorems \ref{theorem:mainh} and \ref{theorem:rect}
is that once $\ax$ is upgraded to a $\bv{\Box}$-algebra,
then the duality arises by rectifying this homotopy structure.
\step
Beware that 
Theorem \ref{theorem:rect} does not actually say that $A$ is a rectification
of $\ax$ in the sense of an $\infty$-quasi-isomorphism,
though this is very likely true,
but \eqref{eq:uu1}, \eqref{eq:uu1x}, \eqref{eq:uu2} are partial statements
that suffice to show that one gets the Yang-Mills amplitudes,
using the recursive characterization and tools from \cite{nr},
see Section \ref{sec:fbcj}.
\step
Some questions:
\begin{itemize}
\item Instead of the $\cinf$ minimal model for $\apw$,
one could consider its $\bv{\Box}$ minimal model.
Assuming this can be defined, is it an interesting object?\footnote{%
At the very least, this
conjectural $\bv{\Box}$ minimal model ought to subsume 
the BCJ relations into its axioms --
rather than implying them by an
ad-hoc calculation as in Section \ref{sec:bcjrelsNEW} --
analogously to how
the $\cinf$ minimal model
subsumes the Kleiss-Kuijf relations
into its axioms; see \cite{cg} for the $\cinf$ axioms.}
\item 
A direct approach to the `double copy'
could be to consider $\ax \otimes \ax$ understood as
a suitably defined tensor product of $\bv{\Box}$-algebras.\footnote{%
Direct approach means as opposed to going through BCJ duality.
The (completed) tensor product $\ax \otimes \ax$
naively gives objects on $\R^4 \oplus \R^4 \ni x \oplus y$,
but can be restricted to things depending only on $x+y$,
or analogously, products of pairs of plane waves
with equal momentum.}\textsuperscript{,}\footnote{%
See also \cite[Section 3]{zeitlin3} for considerations in this ballpark.}
\item
Are there invariant bilinear pairings that give a `cyclic' $\bv{\Box}$-algebra, for loops?
\item What
are other applications and ramifications of this $\bv{\Box}$-algebra structure?
\end{itemize}


\section{Application to BCJ} \label{sec:bcjx}

This section is intended as motivation
for readers familiar with BCJ \cite{bcj,bdhk,bccjr}.
It includes proof sketches of Corollaries \ref{corollary:bcjr} and \ref{corollary:bcjd}.
This section uses results shown in detail later.

\step
\emph{Plane waves and propagators.}
This is the only section where we work with plane waves.
Assume an $h$ as in Section \ref{sec:intro} is fixed.
\begin{itemize}
\item
Let $\apw \hookrightarrow \ax$ be
the subalgebra given by finite linear combinations of plane waves
with complex momentum denoted $k \in \C^4$.
This defines a $\C^4$-grading on $\apw$, conserved by all operations.
Then $h$ is, relative to standard trivializations,
a matrix with entries that are affine linear functions of $k \simeq \p$.
Abbreviate $k^2 = \eta^{\mu\nu}k_\mu k_\nu \simeq \Box$.
Set
\begin{equation}\label{eq:prp}
\prop = h/k^2
\end{equation}
This `propagator'
is ill-defined when acting on plane waves with $k^2 = 0$;
so we write $\prop : \apw \pd \apw$ to indicate that this is a partial map.
The arrow $\pd$ is contagious, but all partial maps
we encounter are defined on dense subsets of momenta.
In every occurrence of $\prop$ it is implicit that
one must have $k^2 \neq 0$.
With this proviso,
\begin{equation}\label{eq:Hx}
d\prop+\prop d=\one
\qquad\qquad
\prop^2=\prop h=h\prop=0
\end{equation}
by \eqref{eq:hxs}. The homology\footnote{%
If $\prop$ was a true element of $\End^{-1}(\apw)$
then \eqref{eq:Hx} would erroneously imply that $H(\apw)$ vanishes.} $H(\apw)$
is a direct sum of the fibers of a non-trivial vector bundle over the complex cone $k^2 = 0$
in $\C^4$, assuming $k \neq 0$, see \cite{nr} and cf.~Lemma \ref{lemma:bdy}.
The homology at $k=0$ is never used.
\item Analogous notation is used for the cobar construction
$\Apw \hookrightarrow A$ of $\apw \hookrightarrow \ax$;
it is spanned by
Gerstenhaber words of co-Gerstenhaber words (cf.~\cite{w})
in plane waves in $\apw$.
The momentum $k \in \C^4$ of such a word-of-words is the total
momentum of all plane waves in it; this defines a $\C^4$-grading on $\Apw$.
Define $\Aprop : \Apw \pd \Apw$ by
\begin{equation}\label{eq:Aprp}
\Aprop = h_A/k^2
\end{equation}
It follows from \eqref{eq:strdaha} that
\begin{equation}\label{eq:AHx}
d_A \Aprop + \Aprop d_A=\one
\qquad\qquad
\Aprop^2=\Aprop h_A=h_A\Aprop=0
\end{equation}
\end{itemize}
\step
\indent\indent\emph{Feynman expansion.}
Given a contraction, aka deformation retract,
from $\apw$ (or $\Apw$) to its homology $H(\apw) \simeq H(\Apw)$,
there are standard formulas for\ldots
\begin{itemize}
\item \ldots the $\cinf$ minimal model of the dgca $\apw$ (or $\Apw$)
using planar cubic trees,
defining minimal model operations $H(\apw)^{\otimes n} \pd H(\apw)$,
see \cite{cg} or \cite[Section 13.1.9]{lv}.
\item \ldots the $\linf$ minimal model of the dgLa $\apw \otimes \cx$ (or $\Apw \otimes \cx$)
using cubic trees, defining minimal model operations
$H(\apw \otimes \cx)^{\otimes n} \pd H(\apw \otimes \cx)$,
see \cite[Section 10.3.4]{lv} or \cite{nr}.
\end{itemize}
Let us call them Feynman-Kontsevich-Soibelman trees.
These trees have one output and $n \geq 2$ inputs labeled $1\ldots n$.
The rules for $\cinf$ trees are summarized in the next table:
\[
{\renewcommand{\arraystretch}{1.1}
\begin{tabular}{c | c l | c l}
tree element & \multicolumn{2}{c|}{operation, using $\apw$}
   & \multicolumn{2}{c}{operation, using $\Apw$}\\
\hline
vertex & gca product & \xxx{13}{\apw^i \otimes \apw^j}{\to \apw^{i+j}} &
gca product & \xxx{13}{\Apw^i \otimes \Apw^j}{\to \Apw^{i+j}}\\
internal line & $\prop$ & \xxx{13}{\apw^i}{\pd \apw^{i-1}} &
$\Aprop$ & \xxx{13}{\Apw^i}{\pd \Apw^{i-1}}\\
input & $i$ & \xxx{13}{H^i(\apw)}{\pd \apw^i} &
$\imap i$ & \xxx{13}{H^i(\apw)}{\pd \Apw^i}\\
the output & $p$ & \xxx{13}{\apw^i}{\pd H^i(\apw)} &
$p\pmap$ & \xxx{13}{\Apw^i}{\pd H^i(\apw)}
\end{tabular}}
\]
and they involve signs. For $\linf$ one must suitably tensor with $\cx$,
the Lie bracket $\cx \otimes \cx \to \cx$, and $\one: \cx \to \cx$.
Canonically, $H(\apw \otimes \cx) \simeq H(\apw) \otimes \cx$.
The quasi-isomorphisms
\begin{equation}\label{eq:ipip}
i: H(\apw) \pdd \apw : p
\end{equation}
can be and are chosen to conserve the $\C^4$-grading;
they are nonzero only along $k^2 = 0$; undefined only at $k=0$;
otherwise arbitrary.
By $\Hopf$-equivariance,
$\imap : \apw \ndd \Apw : \pmap$
in Theorem \ref{theorem:rect} conserve the $\C^4$-grading.
The minimal model operations do not depend on the contraction
-- this includes $i$ and $p$ and the propagator --
for a proof in the $\linf$ case see \cite{nr}.
\step
Working with the operations on $\apw$,
the Yang-Mills amplitudes correspond to the $\cinf$ operation $H^1(\apw)^{\otimes n} \pd H^2(\apw)$,
or, including the color Lie algebra $\cx$,
to the $\linf$ operation $H^1(\apw \otimes \cx)^{\otimes n} \pd H^2(\apw \otimes \cx)$ \cite{nr}.
Using $\Apw$ gives the same result, see Section \ref{sec:fbcj}.

\subsection{BCJ relations} \label{sec:bcjrelsNEW}

The following lemma is a direct consequence
of the definition of $S$ in \eqref{eq:sdef},
together with the dgca axioms and \eqref{eq:hxs}.
One does not need Theorem \ref{theorem:2ax} to prove this.
\begin{lemma}\label{lemma:x1}
The map $S$ has degree $-1$, is graded symmetric\footnote{%
Explicitly, $S(x,y,z) = (-1)^{xy}S(y,x,z) = (-1)^{yz}S(x,z,y)$.},
and
\begin{equation}\label{eq:dSSd}
d S(x,y,z) = - S(dx,y,z) - (-1)^x S(x,dy,z) - (-1)^{x+y} S(x,y,dz)
\end{equation}
and
\begin{equation}\label{eq:4S}
\begin{aligned}
(-1)^x xS(y,u,v) - (-1)^{xy+y} y S(x,u,v)
+ S(x,yu,v) - (-1)^{xy} S(y,xu,v) & = 0\\
S(xy,u,v)-S(x,yu,v)+S(x,y,uv)-(-1)^{x(y+u+v)} S(y,u,vx) & = 0
\end{aligned}
\end{equation}
\end{lemma}
\begin{proof}
For \eqref{eq:dSSd} also use
$\Box(f_1f_2f_3) - \Box(f_1f_2)f_3 - \Box(f_2f_3)f_1 - \Box(f_3f_1)f_2 + \Box(f_1)f_2f_3
+ \Box(f_2)f_3f_1 + \Box(f_3)f_1f_2 = 0$
for all $f_1,f_2,f_3 \in C^\infty(\R^4)$.
\qed\end{proof}

In this section
we show the BCJ relations given by Tye and Zhang \cite[equation (4.32)]{tz},
that is, Corollary \ref{corollary:bcjr}.
For every $n \geq 3$, the argument is in two steps:
\begin{itemize}
\item \emph{Step 1.}
First we define a map $S_n : \apw^{\otimes n} \pd \apw$,
which is a certain signed sum of cubic trees with $n$ inputs and one output,
where each of the $n-1$ vertices is a gca product;
and internal/input/output lines are decorated by a total
of $n-3$ propagators $\prop$ and a single $h = k^2 \prop$.
Lemma \ref{lemma:x1} is then used to show that $S_n$ is a chain map:
\[
d S_n(x_1,\ldots,x_n) = 
\textstyle\sum_{i=1}^n (-1)^{n+x_1 + \ldots + x_{i-1}}
S_n(x_1,\ldots,dx_i,\ldots,x_n)
\]
In particular, $S_n$ induces a map on homology, $H(\apw)^{\otimes n} \pd H(\apw)$.
\item \emph{Step 2.}
Using Theorem \ref{theorem:2ax} we then show that the induced
$H(\apw)^{\otimes n} \pd H(\apw)$ is zero.
Namely, $\theta_3$ is used to construct a
$T_n : \apw^{\otimes n} \pd \apw$ such that
\[
S_n(x_1,\ldots,x_n) = 
dT_n(x_1,\ldots,x_n)
+ \textstyle\sum_{i=1}^n (-1)^{n+x_1 + \ldots + x_{i-1}}
T_n(x_1,\ldots,dx_i,\ldots,x_n)
\]
\end{itemize}
The vanishing of the map $H^1(\apw)^{\otimes n} \pd H^2(\apw)$
induced by $S_n$ is \cite[equation (4.32)]{tz}.
\step
It is useful to first discuss $n=3$, then $n=4$, then sketch the general $n$ case.
\begin{itemize}
\item \emph{Case $n=3$.}
Set $S_3 = S$.
Step 1 follows from \eqref{eq:dSSd}.
Step 2 follows from \eqref{eq:exact} using $T_3 = \theta_3$.
This case goes back to an observation of Zhu \cite{zhu}.
\item \emph{Case $n=4$.}
Set $E(x,y) = \prop(xy)$ and set
\begin{align*}
S_4(x,y,u,v) & = S(E(x,y),u,v) - (-1)^x S(x,E(y,u),v)\\
& \qquad + (-1)^{x+y} S(x,y,E(u,v)) - (-1)^{y+u+x(y+u+v)} S(y,u,E(v,x))
\end{align*}
Note $dE(x,y) = xy - E(dx,y) - (-1)^x E(x,dy)$ by \eqref{eq:Hx},
with an `error term' $xy$.
Step 1 follows from this, from \eqref{eq:dSSd},
and from the second of \eqref{eq:4S} to cancel `error terms'.
Step 2 follows from \eqref{eq:exact} and the second of \eqref{eq:identity4} using
the witness
\begin{align*}
T_4(x,y,u,v) & = \theta_3(E(x,y),u,v) - (-1)^x \theta_3(x,E(y,u),v)\\
& \qquad + (-1)^{x+y} \theta_3(x,y,E(u,v)) - (-1)^{y+u+x(y+u+v)} \theta_3(y,u,E(v,x))
\end{align*}
\item \emph{General $n \geq 3$, sketch.}
Set $E(x_1) = x_1$.
For $n \geq 2$
let $E(x_1,\ldots,x_n)$ be the sum, with suitable signs, of all planar cubic 
trees with $n$ inputs and one output\footnote{%
There are $C_{n-1}$ such trees, with $C$ the Catalan numbers.},
where each vertex is a gca product; each internal 
line and the output line is decorated by $\prop$; the input lines
are not decorated.
Schematically set
\begin{multline}\label{eq:sdefs}
S_n(x_1,\ldots,x_n) =\\
\textstyle\sum \pm S(
E(x_i,x_{i+1},\ldots,x_{j-1}),
E(x_j,x_{j+1},\ldots,x_{k-1}),
E(x_k,x_{k+1},\ldots,x_{i-1}))
\end{multline}
where the subscripts are in $\Z/n\Z$
and the sum is over all decompositions
of $\Z/n\Z$ into three nonempty contiguous subsets\footnote{%
The number of ways of decomposing $\Z/n\Z$
into $m$ nonempty contiguous subsets is ${n \choose m}$.}.
Step 1 and Step 2 follow like before\footnote{%
In this entire section, only the second identity in \eqref{eq:4S}
and \eqref{eq:identity4} respectively is used;
there should be an analogous discussion using the first in 
\eqref{eq:4S}
and \eqref{eq:identity4} respectively.}.
\end{itemize}
The translation to \cite{tz} uses the following remarks:
\begin{itemize}
\item Take the definition of $S_n$, \eqref{eq:sdefs},
and substitute the definitions of $S$, \eqref{eq:sdef}, and of $E$.
One obtains cubic trees as described in Step 1.
Since $h\prop=\prop h=0$, \eqref{eq:Hx},
a tree is zero unless all $n-3$ propagators $\prop$ and the single $h$
appear on \emph{distinct} lines.
\item There are some trees where $h$ decorates an input or output line,
but they do not contribute to $H^1(\apw)^{\otimes n} \pd H^2(\apw)$
by \eqref{eq:kih};
these degrees give the YM amplitudes.
So only trees where all $\prop$ and $h$ appear on \emph{distinct internal} lines
contribute to this.
\item Replace the single $h$ by $h = k^2 \prop$,
and regard $k^2$ (the square of the momentum flowing through an internal line)
as a weight multiplying a
cubic tree that uses $\prop$ on all internal lines,
that is, a standard $\cinf$ Feynman-Kontsevich-Soibelman tree.
\end{itemize}

The discussion in this section is ad-hoc.
Presumably, if one can define the $\bv{\Box}$ minimal model of $\apw$,
the BCJ relations will be subsumed into the $\bv{\Box}$ axioms.

\newcommand{\hopt}{\tc{hcolor}{h^{\textnormal{opt}}}}
\newcommand{\Ahopt}{\tc{hcolor}{h^{\textnormal{opt}}_A}}
\subsection{BCJ/color-kinematics duality} \label{sec:fbcj}

\begin{lemma}
If a functional $\ell \in A^\ast$
satisfies $\ell h_A = 0$, then for all $x,y,z \in \ker h_A$:
\begin{equation}\label{eq:jxyz}
\ell(h_A(xy) z) + (-1)^{x(y+z)} \ell(h_A(yz)x) + (-1)^{z(x+y)} \ell(h_A(zx)y) = 0
\end{equation}
where juxtaposition denotes the gca product in $A$.
Analogous for $\Apw \hookrightarrow A$.
\end{lemma}
\begin{proof}
Since $(A,h_A)$ is a BV algebra
by Theorem \ref{theorem:rect} (cf.~Lemma \ref{lemma:rcc}), $h_A$ is genuinely second order,
so the right hand side of
\eqref{eq:sdef},
with $h_A$ in the role of $h$
and all free variables in $A$, vanishes.
Using graded symmetry of the gca product in $A$, the claim follows.
\qed\end{proof}

To prove Corollary \ref{corollary:bcjd},
consider for every $n \geq 2$ the $n$-ary operation
of the $\linf$ minimal model,
but specialized to the degrees of interest, see \eqref{eq:mnmn}.
Concretely, $\yma_n$ is a sum of Feynman-Kontsevich-Soibelman trees
with propagator \eqref{eq:Aprp}:
\begin{equation}\label{eq:treeexp}
\yma_n = \textstyle\sum_{T \in \mathcal{T}_n} \yma_{n,T'} \otimes c_{T'}
\end{equation}
Here $\mathcal{T}_n$ is the set of all (non-planar) cubic trees with one output
and $n$ inputs labeled $1\ldots n$, so that $|\mathcal{T}_n| = (2n-3)!!$,
and $T'$ denotes any planar embedding of $T$.
Here
$c_{T'} : \cx^{\otimes n} \to \cx$ 
nests $n-1$ instances of the Lie bracket as dictated by the planar embedding $T'$,
and
\[
\yma_{n,T'} : H^1(\apw)^{\otimes n} \pd H^2(\apw)
\] 
is defined by decorating $T'$ as follows:
\[
\begin{tabular}{c | c l}
tree element & \multicolumn{2}{c}{operation}\\
\hline
vertex & gca product & \xxx{15}{\Apw^1 \otimes \Apw^1}{\to \Apw^2}\\
internal line & $\Aprop = h_A/k^2$ & \xxx{15}{\Apw^2}{\pd \Apw^1}\\
input & $\imap i$ & \xxx{15}{H^1(\apw)}{\pd \Apw^1}\\
the output & $p \pmap$ & \xxx{15}{\Apw^2}{\pd H^2(\apw)}
\end{tabular}
\]
e.g.~$\yma_{3,(1(23))}(x,y,z) = p \pmap ((\imap i x) \Aprop ((\imap i y)(\imap i z)))$.
So \eqref{eq:treeexp} does not depend on the choice of the $T'$
since the Lie bracket on $\cx$
and the product $\Apw^1 \otimes \Apw^1 \to \Apw^2$ are antisymmetric.
\step
One can view $\mathcal{T}_n$ as an undirected graph
where an edge corresponds to\footnote{%
Cf.~the associahedron, aka Stasheff polytope,
which is for planar cubic trees.}
`one associativity move', see Figure \ref{fig:t4}.
By definition, \eqref{eq:treeexp}
satisfies BCJ/color-kinematics duality \cite{bcj,bdhk,bccjr} if and only if
for every $3$-clique\footnote{%
Given a graph, an $n$-clique is a subset of $n$ vertices
that are all pairwise connected by an edge.}
$U,V,W \in \mathcal{T}_n$ one has
\begin{equation}\label{eq:bcjkcd}
k_U^2 \yma_{n,U'} + k_V^2 \yma_{n,V'} + k_W^2 \yma_{n,W'} = 0
\end{equation}
where $k_U$ is the momentum flowing
through the distinguished line in $U$\footnote{%
The distinguished line in $U$ is the one not present in $V$ and $W$.},
and the planar embeddings $U'$, $V'$, $W'$
are chosen compatibly in the sense that the four subtrees
connecting to the distinguished lines,
which appear jointly in $U$ and $V$ and $W$,
are embedded in the same way as part of $U'$ and $V'$ and $W'$,
and the relative planar embedding of the four subtrees is cyclic,
in particular so that $c_{U'} + c_{V'} + c_{W'} = 0$ by the Jacobi identity in $\cx$.
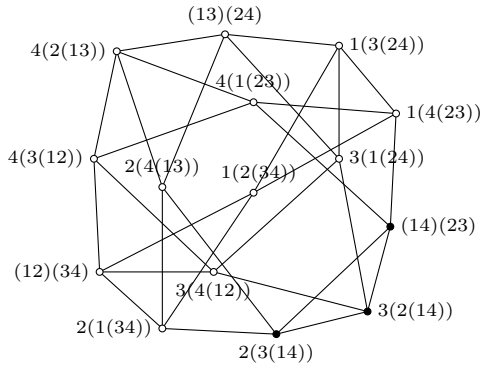
\begin{figure}
\centering
\begin{tikzpicture}[scale=0.75,every node/.style={font=\ssmall}]
\coordinate (a1) at (0.0,3.1);
\coordinate (a2) at (5.3,3.9);
\coordinate (a3) at (2.8,4.1);
\coordinate (a4) at (2.8,2.5);
\coordinate (a5) at (0.1,1.1);
\coordinate (a6) at (2.1,1.1);
\coordinate (a7) at (4.3,5.1);
\coordinate (a8) at (4.3,3.1);
\coordinate (a9) at (0.4,5.0);
\coordinate (a10) at (2.3,5.3);
\coordinate (a11) at (1.2,2.6);
\coordinate (a12) at (1.2,0.1);
\coordinate (a13) at (4.8,0.4);
\coordinate (a14) at (5.2,1.9);
\coordinate (a15) at (3.2,0.0);
\draw (a1) to (a3); \draw (a1) to (a5); \draw (a1) to (a6); \draw (a1) to (a9); \draw (a2) to (a3); \draw (a2) to (a4); \draw (a2) to (a7); \draw (a2) to (a14); \draw (a3) to (a9); \draw (a3) to (a14); \draw (a4) to (a5); \draw (a4) to (a7); \draw (a4) to (a12); \draw (a5) to (a6); \draw (a5) to (a12); \draw (a6) to (a8); \draw (a6) to (a13); \draw (a7) to (a8); \draw (a7) to (a10); \draw (a8) to (a10); \draw (a8) to (a13); \draw (a9) to (a10); \draw (a9) to (a11); \draw (a10) to (a11); \draw (a11) to (a12); \draw (a11) to (a15); \draw (a12) to (a15); \draw (a13) to (a14); \draw (a13) to (a15); \draw (a14) to (a15);
\draw (a1) node [anchor=east] {$4(3(12))$};
\draw (a2) node [anchor=west] {$1(4(23))$};
\draw (a3) node [anchor=south] {$4(1(23))$};
\draw (a4) node [anchor=south,xshift=2] {$1(2(34))$};
\draw (a5) node [anchor=east] {$(12)(34)$};
\draw (a6) node [anchor=north] {$3(4(12))$};
\draw (a7) node [anchor=west] {$1(3(24))$};
\draw (a8) node [anchor=west] {$3(1(24))$};
\draw (a9) node [anchor=east] {$4(2(13))$};
\draw (a10) node [anchor=south] {$(13)(24)$};
\draw (a11) node [anchor=south] {$2(4(13))$};
\draw (a12) node [anchor=east] {$2(1(34))$};
\draw (a13) node [anchor=west] {$3(2(14))$};
\draw (a14) node [anchor=west] {$(14)(23)$};
\draw (a15) node [anchor=north] {$2(3(14))$};
\draw[fill=white] (a1) circle (0.06);
\draw[fill=white] (a2) circle (0.06);
\draw[fill=white] (a3) circle (0.06);
\draw[fill=white] (a4) circle (0.06);
\draw[fill=white] (a5) circle (0.06);
\draw[fill=white] (a6) circle (0.06);
\draw[fill=white] (a7) circle (0.06);
\draw[fill=white] (a8) circle (0.06);
\draw[fill=white] (a9) circle (0.06);
\draw[fill=white] (a10) circle (0.06);
\draw[fill=white] (a11) circle (0.06);
\draw[fill=white] (a12) circle (0.06);
\draw[fill=black] (a13) circle (0.06);
\draw[fill=black] (a14) circle (0.06);
\draw[fill=black] (a15) circle (0.06);
\end{tikzpicture}
\caption{The graph $\mathcal{T}_4$. Each vertex is a cubic tree
with one output and four inputs $1,\ldots,4$.
There are many $3$-cliques; an example is given by the three black vertices.
In $\mathcal{T}_n$, every vertex belongs to exactly $n-2$ many $3$-cliques,
and there are no larger cliques.
}\label{fig:t4}
\end{figure}
\step
BCJ duality \eqref{eq:bcjkcd} for the expansion \eqref{eq:treeexp}
follows from \eqref{eq:jxyz},
by letting $x$, $y$, $z$ be the outputs of three of the four subtrees,
and letting $\ell$ be the input to the fourth.
The degrees work out to be $x,y,z \in \Apw^1$ and $\ell \in (\Apw^2)^\ast$
so that the signs in \eqref{eq:jxyz} disappear.
\begin{itemize}
\item
Proof of $x \in \ker h_A$:
If $x$ is not an overall input, then
$x = h_A \xi$ for some $\xi \in \Apw^2$,
so $h_A x = (h_A)^2 \xi = 0$ by \eqref{eq:strdaha}.
If $x$ is an overall input, then $x = \imap i \xi$
for some $\xi \in H^1(\apw)$,
so $h_A x = h_A \imap i \xi = \imap h i \xi = 0$
by \eqref{eq:uu1x} and \eqref{eq:kih}.
Analogously $y,z \in \ker h_A$.
\item
Proof of $\ell h_A=0$:
If $\ell$ is not the overall output,
then $\ell = \nu h_A$ for some $\nu \in (\Apw^1)^\ast$,
so $\ell h_A = \nu (h_A)^2 = 0$ by \eqref{eq:strdaha}.
If $\ell$ is the overall output,
then $\ell = \nu p \pmap$ for some $\nu \in H^2(\apw)^\ast$,
so $\ell h_A = \nu p \pmap h_A = \nu p h \pmap = 0$
by \eqref{eq:uu1x} and \eqref{eq:kih}.
\end{itemize}
\step
\indent\indent
To prove Corollary \ref{corollary:bcjd},
we must also check that $\yma_n$ are the ordinary YM amplitudes\footnote{%
E.g., why does the Feynman expansion
for $\apw \otimes \cx$, which fails BCJ duality, give the same result?}.
We exploit the fact that the YM amplitudes
can be recursively characterized in terms of the `factorization of their residues'.
This is known from Britto-Cachazo-Feng-Witten recursion.
We use the characterization of A.~N\"utzi and the author in \cite{nr} that
is simpler in that it does not involve any BCFW shift
nor limiting conditions `at infinity'
under such shifts\footnote{%
These points `at infinity'
are irrelevant for the recursive characterization,
by a Hartogs extension theorem,
because they are contained in (projectively) a codimension two subset.}.
\step
Let $k_1,\ldots,k_n$ be the input momenta viewed as standard coordinates on $(\C^4)^n \simeq \C^{4n}$;
and abbreviate $k_{n+1} = k_1 + \ldots + k_n$ which is the output momentum.
The argument that the $\yma_n$ are the ordinary Yang-Mills amplitudes goes as follows (sketch):
\begin{itemize}
\item Recall \eqref{eq:mnmn};
 recall that $\yma_n$ is undefined if
 \eqref{eq:Aprp}
 is undefined meaning $k_J^2 = 0$ for some $J \subset \{1,\ldots,n\}$
 with $1 < |J| < n$,
where $k_J = \sum_{j \in J} k_j$;
 recall that the one-particle homology
 is a sheaf over the cone in $\C^4$ minus the origin,
 see \cite{nr} and cf.~Lemma \ref{lemma:bdy}.
 So, geometrically,
 $\yma_n$ is a section of an algebraic sheaf along $V_n - P_n$ where
\[
V_n = \{k_1^2 = \ldots = k_{n+1}^2 = 0\}
\qquad\quad
P_n = \textstyle
\bigcup_{j=1\ldots n+1} \{k_j = 0\} \cup
\bigcup_{1 < |J| < n} \{k_J^2 = 0\}
\]
Here $P_n$ is a closed subset of the variety $V_n$,
of codimension one if $n>2$.
\item $\yma_n$ is homogeneous jointly in $k_1,\ldots,k_n$
with the natural degree of homogeneity,
because all constructions in this paper, including $\Apw$, scale naturally\footnote{%
Since $h$ scales naturally
by the requirements in Section \ref{sec:intro}, all operations do.}.
Thus $\yma_n$ is, equivalently, a section of a sheaf on
the projective variety associated to $V_n - P_n$.
\item $\yma_2$ coincides with the map
$H^1(\apw \otimes \cx)^{\otimes 2} \pd H^2(\apw \otimes \cx)$
that is induced by the gca product $\apw^1 \otimes \apw^1 \to \apw^2$
times the Lie bracket $\cx \otimes \cx \to \cx$,
because of \eqref{eq:uu2},
hence it is Lorentz invariant and nonzero on both irreducible components of $V_2 - P_2$.
\end{itemize}
Hereafter, $n>2$. Then $V_n$ is irreducible \cite{nr}.
Decompose $\bigcup_{1 < |J| < n} \{k_J^2 = 0\} \subset V_n$
into its irreducible components
denoted $V(\px) \subset V_n$, where
$\px$ is the corresponding prime divisor in the coordinate ring of $V_n$.
A classification and properties of the $\px$ is in \cite{nr};
beware that one $V(\px)$ can be a component of several $\{k_J^2 = 0\}$;
this happens but only for $n=3$.
\begin{itemize}
\item $\yma_n$ has only simple poles along each $V(\px)$.
More precisely, $\yma_n$ extends,
as a section of a sheaf that allows only these specific simple poles \cite{nr},
to the bigger open set $V_n - Z_n$,
where the bad set $Z_n = \bigcup_{\px \neq \qx} V(\px) \cap V(\qx)$ has codimension two.
\item
The recursive characterization requires showing that
$\Res_{\px} \yma_n$, the residue along $V(\px)$,
equals some bilinear expression in $\yma_2,\ldots,\yma_{n-1}$ for every $\px$,
see \cite{nr}\footnote{This is known as `factorization of residues'.}.
This is not immediate when using the propagator \eqref{eq:Aprp}.
But one can use gauge independence in \cite{nr}
to replace the relevant instances of \eqref{eq:Aprp} by another propagator --
a composition of an optimal homotopy as in \cite{nr}
with the contraction \eqref{eq:uu1} --
suitable to evaluate $\Res_{\px}\yma_n$
and to check that it is of the required form. Here are more details:
\begin{itemize}
\item Let
$\apw = \oplus_{k \in \C^4} \apx{k}$
and
$\Apw = \oplus_{k \in \C^4} \Apx{k}$
be the $\C^4$-gradings by momentum.
Then $\apx{k}$ is finite-dimensional, but $\Apx{k}$ is infinite-dimensional
since it is the span of all words-of-words
in plane waves in $\apw$ with total momentum $k$.
\item
By $\Hopf$-equivariance,
$d(\apx{k}) \subset \apx{k}$
and $d_A(\Apx{k}) \subset \Apx{k}$
and the binary operations
on $\apw$ and $\Apw$ conserve the $\C^4$-grading.
Hence the gauge independence theorem as stated in \cite{nr} applies 
not only to $\apw \otimes \cx$ but also to $\Apw \otimes \cx$.
\item Suppose we want to evaluate $\Res_\px \yma_n$ at a point
\begin{equation}\label{eq:kstar}
(k_{1\star},\ldots,k_{n\star}) \in V(\px) - Z_n
\end{equation}
For every $J$ with $1 < |J| < n$ such that $V(\px) \subset \{k_J^2 = 0\}$
(if $n > 3$ then $J$ is unique) set
$q_\star = k_{J\star} = \textstyle\sum_{j \in J} k_{j\star}$,
so $q_\star \in \C^4-0$ and $q_\star^2 = 0$.
Denote by
\[
i_\star : H(\apx{q_\star}) \ndd \apx{q_\star} : p_\star
\qquad\qquad
\imap_\star : \apx{q_\star} \ndd \Apx{q_\star} : \pmap_\star
\]
the restrictions of \eqref{eq:ipip} respectively \eqref{eq:uu1} to $q_\star$;
and by
$\chi_\star: H(\apx{q_\star}) \to H(\apx{q_\star})$
the canonical degree $-1$ map whose only component is \eqref{eq:h2121}.
\item 
An optimal homotopy \cite{nr} for $\apw$ based at $q_\star$ is
a degree $-1$ map
$\hopt : \apw \pd \apw$
with these properties:
it is given by a matrix whose entries are rational functions of $k$
so that $k^2 \hopt$ is regular at $q_\star$;
it satisfies
$d\hopt+\hopt d = \one$
and
$(\hopt)^2 = 0$
and $\lim_{k \to q_\star} k^2 \hopt = i_\star \chi_\star p_\star$.
This is constructed in \cite{nr} by nesting
two instances of the homological perturbation lemma.
Set\footnote{%
The maps $\imap$, $\pmap$, $\zmap$
from \eqref{eq:uu1} are unconditionally, as opposed to partially, defined.}
\[
\Ahopt = \zmap + \imap \hopt \pmap \;:\; \Apw \pd \Apw
\]
It has these properties:
$k^2 \Ahopt$ is regular at $q_\ast$ and
(to check this one also uses the equations satisfied by the contraction
\eqref{eq:uu1}, spelled out in Lemma \ref{lemma:rcc})\footnote{%
A precise definition of the limit is omitted.
One possibility is this:
Take a word-of-words of $m \geq 1$ plane waves viewed as depending parametrically on the $m$ momenta;
then replace them by $m$ sequences of momenta so that the
$m$ sequences converge separately and so that the total momentum converges to $q_\star$.
Evaluating $k^2 \Ahopt$ along such a sequence gives the claimed result.}:
\begin{equation}\label{eq:heqx}
d_A \Ahopt + \Ahopt d_A = \one
\qquad
(\Ahopt)^2 = 0
\qquad
\textstyle\lim_{k \to q_\star} k^2 \Ahopt
   = \imap_\star i_\star \chi_\star p_\star \pmap_\star
\end{equation}
\item
For all $T \in \mathcal{T}_n$,
consider the line -- if there is one -- with momentum $k_J$ flowing through,
and replace the corresponding $\Aprop$ by $\Ahopt$ in the definition of $\yma_{n,T'}$.
The first two in \eqref{eq:heqx} mean that $\Ahopt$ is an admissible propagator,
hence this replacement does not change $\yma_n$, by gauge independence \cite{nr}.
The last in \eqref{eq:heqx}
implies that $\Res_{\px} \yma_n$ is, at the point \eqref{eq:kstar},
of the required form.
\end{itemize}
\end{itemize}
Hence the $\yma_n$ are the Yang-Mills amplitudes,
by the recursive characterization in \cite{nr}.
\subsection{Syntactic kinematic algebra}\label{sec:ska}

A slightly different Feynman expansion,
using \eqref{eq:lie} at each vertex, also produces the Yang-Mills amplitudes $\yma_n$.
Namely, for every $n \geq 2$ set
\[
\widetilde{\yma}_n = \textstyle\sum_{T \in \mathcal{T}_n} \widetilde{\yma}_{n,T'} \otimes c_{T'}
\;\;:\;\;
 H^1(\apw \otimes \cx)^{\otimes n} \pd H^1(\apw \otimes \cx)
\]
where
$\widetilde{\yma}_{n,T'} : H^1(\apw)^{\otimes n} \pd H^1(\apw)$
is defined by decorating $T'$ as follows:
\[
\begin{tabular}{c | c l}
tree element & \multicolumn{2}{c}{operation} \\
\hline
vertex & gLa product & \xxx{15}{\Apw^1 \otimes \Apw^1}{\to \Apw^1}\\
internal line & $1/k^2$ & \xxx{15}{\Apw^1}{\pd \Apw^1}\\
input & $\imap i$ & \xxx{15}{H^1(\apw)}{\pd \Apw^1}\\
the output & $p \pmap$ & \xxx{15}{\Apw^1}{\pd H^1(\apw)}
\end{tabular}
\]
e.g.~$\widetilde{\yma}_{3,(1(23))}(x,y,z)
      = p \pmap [\imap i x, (1/k^2) [\imap i y,\imap i z]]$.
Note that $\Apw^1$ is an ordinary Lie algebra.
\begin{lemma}[Equivalence] \label{lemma:ska} For every $n \geq 2$,
\[
\widetilde{\yma}_n \;=\; (-1)^{n-1} (\chi \otimes \one) \yma_n
\]
where $\chi : H^2(\apw) \pd H^1(\apw)$
is the canonical isomorphism, see \eqref{eq:h2121}.
\end{lemma}
\begin{proof}
Recall that $(\Apw,h_A)$ is a BV algebra, see \eqref{eq:unf}.
Hence for all $x,y \in \Apw^1$:
\begin{equation}\label{eq:rewr}
[x,y] = - h_A(xy) + h_A(x)y - xh_A(y)
\end{equation}
Rewrite each vertex in $\widetilde{\yma}_{n,T'}$ using \eqref{eq:rewr}
and multiply out, which produces $3^{n-1}$ terms.
All terms but one vanish because
$h_A \imap i = \imap h i = 0$ on $H^1(\apw)$
by \eqref{eq:uu1x} and \eqref{eq:kih};
and because $h_A^2 = 0$ by \eqref{eq:strdaha}.
The only term that survives is the one that uses $(x,y) \mapsto -h_A(xy)$ at all vertices. 
Note that $p \pmap h_A = p h \pmap$
by \eqref{eq:uu1x}. Summing over all trees $T \in \mathcal{T}_n$ gives
\begin{equation}\label{eq:hwkshs}
\widetilde{\yma}_n = (-1)^{n-1} (ph \otimes \one)\yma'_n
\end{equation}
where $\yma'_n : H^1(\apw \otimes \cx)^{\otimes n} \pd \apw^2 \otimes \cx$
is defined like $\yma_n$ but without the final $p$ in all trees;
in particular $(p\otimes \one)\yma'_n = \yma_n$.
But the image of $\yma'_n$ is in the kernel of the differential,
meaning $(d \otimes \one) \yma'_n = 0$\footnote{%
First use $d\pmap = \pmap d_A$, see \eqref{eq:uu1x}.
Then this is a property of the construction of $\yma'_n$,
not of contributions from individual trees but only of the sum over all trees,
see the cancellation lemma in \cite{nr}.};
and $p$ applied to a plane wave with momentum $k$ is nonzero only if $k^2=0$;
therefore \eqref{eq:h2121}
implies that one can replace $ph$ by $\chi p$ in \eqref{eq:hwkshs}.
\qed\end{proof}

\newcommand{\hstar}{\tc{scolor}{{\star}}}
\newcommand{\axa}{\ax_{\#}}
\newcommand{\da}{d_{\#}} 
\newcommand{\ha}{h_{\#}} 
\section{Definition of $\ax$ and existence of $h$}\label{sec:defs}
\emph{Definition of $\ax$.}
Let $\Lambda$
be the exterior $\C$-algebra generated by $e^0,e^1,e^2,e^3$
in degree one.
The de Rham dgca on $\R^4$ is
$\Omega = C^\infty(\R^4) \otimes \Lambda$
with differential $\ddr = \p_\mu \otimes e^\mu$
where by a small abuse of notation,
$e^\mu$ is also used as a multiplication operator from the left.
Let $\C \oplus \C \eps$ be the dgca where $\eps$
is an element of degree $-1$;
the product is given by $\eps^2 = 0$;
and the differential is given by
$1 \mapsto 0$ and $\eps \mapsto 1$.
Consider the tensor product of dgca
\[
\axa\;=\;\Omega \otimes (\C \oplus \C \eps) \;\simeq\; \Omega \oplus \Omega\eps
\]
This means that differential and product on $\axa$ are respectively given by
\begin{align*}
\da:\qquad x \oplus u \eps
& \mapsto (\ddr x + (-1)^u u) \oplus (\ddr u) \eps\\
(x \oplus u \eps,\; y \oplus v \eps) & \mapsto
xy \oplus (xv + (-1)^y uy) \eps
\end{align*}
for all $x,y,u,v\in \Omega$. 
In the diagram for $\axa$ below,
$\Omega$ occupies the upper row,
$\Omega \eps$ occupies the lower row with $\eps$ suppressed.
With this understanding, the differential $\da$ is as follows,
where the six columns correspond to degrees $-1\ldots 4$ of $\axa$ respectively:
\newcommand{\xx}{2.4}
\newcommand{\yy}{-1.6}
\[
\vcenter{\hbox{%
\begin{tikzpicture}
  \node (a0) at (0,0) {$\Omega^0$};
  \node (a1) at (\xx,0) {$\Omega^1$};
  \node (a2) at (2*\xx,0) {$\Omega^2$};
  \node (a3) at (3*\xx,0) {$\Omega^3$};
  \node (a4) at (4*\xx,0) {$\Omega^4$};
  \node (b0) at (-\xx,\yy) {$\Omega^0$};
  \node (b1) at (0,\yy) {$\Omega^1$};
  \node (b2) at (\xx,\yy) {$\Omega^2$};
  \node (b3) at (2*\xx,\yy) {$\Omega^3$};
  \node (b4) at (3*\xx,\yy) {$\Omega^4$};
  \draw[->] (a0) to node [above] {$\ddr$} (a1);
  \draw[->] (a1) to node [above] {$\ddr$} (a2);
  \draw[->] (a2) to node [above] {$\ddr$} (a3);
  \draw[->] (a3) to node [above] {$\ddr$} (a4);
  \draw[->] (b0) to node [above] {$\ddr$} (b1);
  \draw[->] (b1) to node [above] {$\ddr$} (b2);
  \draw[->] (b2) to node [above] {$\ddr$} (b3);
  \draw[->] (b3) to node [above] {$\ddr$} (b4);
  \draw[->] (b0) to node [above,xshift=-4] {$\one$} (a0);
  \draw[->] (b1) to node [above,xshift=-6] {$-\one$} (a1);
  \draw[->] (b2) to node [above,xshift=-4] {$\one$} (a2);
  \draw[->] (b3) to node [above,xshift=-6] {$-\one$} (a3);
  \draw[->] (b4) to node [above,xshift=-4] {$\one$} (a4);
\end{tikzpicture}}}
\]
Set $\Omega^2_{\pm} = C^\infty(\R^4) \otimes \Lambda^2_{\pm}$
where, by definition,
\begin{align}
\Lambda^2_- & = \C(e^0e^1-ie^2e^3) 
\oplus \C (e^0e^2-ie^3e^1) \oplus \C(e^0e^3-ie^1e^2)\\
\label{eq:lambda2plus}
\Lambda^2_+ & = \C(e^0e^1+ie^2e^3) 
\oplus \C (e^0e^2+ie^3e^1) \oplus \C(e^0e^3+ie^1e^2)
\end{align}
Using in particular
$\Omega_-^2 \Omega_+^2 = 0$ one finds,
see \cite{zeitlin,costello,nr} and Appendix \ref{app:bv} for more details:
\begin{lemma}[Definition of $\ax$]
The following is a dgca subquotient of $\axa$:
\begin{equation}\label{eq:defa}
\ax = \frac{\rule{69pt}{0pt}\Omega \oplus (\Omega^2_+ \oplus \Omega^3 \oplus \Omega^4)\eps}{%
             (\Omega^2_- \oplus \Omega^3 \oplus \Omega^4) \oplus 0 \eps \rule{69pt}{0pt}}
\end{equation}
This means that the numerator is a sub-dgca of $\axa$,
and the denominator is a dgca ideal in the numerator.
As a module,
$\ax \simeq
(\Omega^0 \oplus \Omega^1 \oplus \Omega^2_+)
\oplus (\Omega^2_+ \oplus \Omega^3 \oplus \Omega^4)\eps$;
as a complex,
$\ax$ coincides with \eqref{eq:cpx};
and as a dgca, $\ax$ has the structure announced in Section \ref{sec:intro}.
\end{lemma}
\emph{Existence of $h$.}
Note that
$\Lambda^2_- = \{ x \in \Lambda^2 \mid \hstar x = -ix \}$
and 
$\Lambda^2_+ = \{ x \in \Lambda^2 \mid \hstar x = ix \}$,
provided one picks the orientation appropriately,
where $\hstar : \Lambda^i \to \Lambda^{4-i}$
is the Hodge star operator for the bilinear form 
$\eta = \text{diag}(-1,1,1,1)$
on $\C e^0 \oplus \ldots \oplus \C e^3$.
Note that $\hstar^2 = 1$ on odd degree elements,
$\hstar^2 = -1$ on even degree elements.
Also write $\hstar: \Omega^i \to \Omega^{4-i}$.
Denote by
$\deldr = \hstar \ddr \hstar : \Omega^i \to \Omega^{i-1}$
the `adjoint' of $\ddr$.
Let $i_\mu : \Lambda^i \to \Lambda^{i-1}$
be the interior multiplications using a dual basis of the $e^\mu$,
so $e^\mu e^\nu + e^\nu e^\mu = 0$
and $e^\mu i_\nu + i_\nu e^\mu = \delta^\mu_\nu$
and $i_\mu i_\nu + i_\nu i_\mu = 0$.
Then equivalently
\[
\deldr = -\eta^{\mu\nu} \p_\mu \otimes i_\nu
\]
and therefore $\ddr \deldr + \deldr \ddr = -\Box$
where $\Box = \eta^{\mu\nu}\p_\mu\p_\nu$.
Define
\[
\ha = -\deldr \otimes \one
\;\;:\;\; \axa^i \to \axa^{i-1}
\]
Then 
$\ha^2 = 0$ and $\da \ha + \ha \da = \Box$. The map $\ha$ is explicitly given as follows:
\[
\vcenter{\hbox{%
\begin{tikzpicture}
  \node (a0) at (0,0) {$\Omega^0$};
  \node (a1) at (\xx,0) {$\Omega^1$};
  \node (a2) at (2*\xx,0) {$\Omega^2$};
  \node (a3) at (3*\xx,0) {$\Omega^3$};
  \node (a4) at (4*\xx,0) {$\Omega^4$};
  \node (b0) at (-\xx,\yy) {$\Omega^0$};
  \node (b1) at (0,\yy) {$\Omega^1$};
  \node (b2) at (\xx,\yy) {$\Omega^2$};
  \node (b3) at (2*\xx,\yy) {$\Omega^3$};
  \node (b4) at (3*\xx,\yy) {$\Omega^4$};
  \draw[->] (a1) to node [above] {$-\deldr$} (a0);
  \draw[->] (a2) to node [above] {$-\deldr$} (a1);
  \draw[->] (a3) to node [above] {$-\deldr$} (a2);
  \draw[->] (a4) to node [above] {$-\deldr$} (a3);
  \draw[->] (b1) to node [above] {$-\deldr$} (b0);
  \draw[->] (b2) to node [above] {$-\deldr$} (b1);
  \draw[->] (b3) to node [above] {$-\deldr$} (b2);
  \draw[->] (b4) to node [above] {$-\deldr$} (b3);
\end{tikzpicture}}}
\]
While $\ha$ does not descend to the subquotient $\ax$, a modification of it does:
\begin{lemma}[Existence of $h$] \label{lemma:exh}
Given $\C$-linear maps
$z: \Lambda^2 \to \Lambda^0$
and $w: \Lambda^4 \to \Lambda^2$
that satisfy\footnote{%
These conditions are equivalent to $\Lambda^2_- \subset  \ker z$
and $\image w \subset \Lambda^2_+$ respectively.}
$z\hstar = iz$ and $\hstar w = iw$,
define
$X: \axa^i \to \axa^{i-2}$ by
\[
\vcenter{\hbox{%
\begin{tikzpicture}
  \node (a0) at (0,0) {$\Omega^0$};
  \node (a1) at (\xx,0) {$\Omega^1$};
  \node (a2) at (2*\xx,0) {$\Omega^2$};
  \node (a3) at (3*\xx,0) {$\Omega^3$};
  \node (a4) at (4*\xx,0) {$\Omega^4$};
  \node (b0) at (-\xx,\yy) {$\Omega^0$};
  \node (b1) at (0,\yy) {$\Omega^1$};
  \node (b2) at (\xx,\yy) {$\Omega^2$};
  \node (b3) at (2*\xx,\yy) {$\Omega^3$};
  \node (b4) at (3*\xx,\yy) {$\Omega^4$};
  \draw[->,bend right=17] (a2) to node [above] {$z$} (a0);
  \draw[->,bend right=17] (a3) to node [above] {$i\hstar$} (a1);
  \draw[->,bend right=17] (b3) to node [above] {$-i\hstar$} (b1);
  \draw[->,bend right=17] (b4) to node [above] {$w$} (b2);
\end{tikzpicture}}}
\]
Then $\ha + [\da,X] : \axa^i \to \axa^{i-1}$
descends to the subquotient $\ax$, meaning it leaves numerator
and denominator in \eqref{eq:defa} invariant.
The induced map $h: \ax^i \to \ax^{i-1}$
satisfies $dh+hd = \Box$.
If in addition\footnote{%
Clearly such $z$ and $w$ exist.
Necessarily $z,w \neq 0$,
so this $h$ is not Lorentz invariant.}
$zw = 2i\hstar$ 
then 
$h^2 = 0$
and $h$ satisfies all requirements in Section \ref{sec:intro}.
\end{lemma}
\begin{proof}
Recall here that $\Omega\eps$
are things in the lower row in the diagrams.
Descent amounts
to checking that $\ha + [\da,X]$
maps
$\Omega^2_+\eps \to \Omega^0 \oplus 0\eps$,
 $\Omega^3 \eps \to \Omega^1 \oplus \Omega^2_+ \eps$,
 $\Omega^2_- \to  0 \oplus 0\eps$,
and $\Omega^3 \to \Omega^2_- \oplus 0\eps$.
For $h^2 = 0$ it suffices, since $h^2$ has degree $-2$, to check that
$h^2$ annihilates $\Omega^2_+$, $\Omega^3\eps$, $\Omega^4\eps$.
Both these things are by direct calculation.
\qed\end{proof}

One can actually check that all $h$ that satisfy the requirements of 
Section \ref{sec:intro} are as in Lemma \ref{lemma:exh}.
Nevertheless, in this paper the specific form of $h$ in Lemma \ref{lemma:exh} is never used.

\newcommand{\onepw}{\apx{k}}
\begin{lemma}[Plane wave homology] \label{lemma:bdy}
Let $\onepw \subset \ax$ be the subspace of 
plane waves with momentum $k \in \C^4-0$.
The homology of $d|_{\onepw}$ is as follows.
If $k^2 \neq 0$ then $H(\onepw) = 0$. If $k^2 = 0$ then
$H^0(\onepw) = H^3(\onepw) = 0$
and $H^1(\onepw) \simeq H^2(\onepw) \simeq \C^2$ and
\begin{equation}\label{eq:kih}
\begin{aligned}
   \ker(d: \onepw^1 \to \onepw^2) \; &\subset\; \ker(h: \onepw^1 \to \onepw^0)\\
   \image(h: \onepw^3 \to \onepw^2) \; &\subset\; \image(d: \onepw^1 \to \onepw^2)
\end{aligned}
\end{equation}
Further, for $k^2 = 0$,
the map $h : \onepw^2 \to \onepw^1$ induces an isomorphism
\begin{equation} \label{eq:h2121}
H^2(\onepw) \to H^1(\onepw)
\end{equation}
and this isomorphism is canonical, independent of $h$.
\end{lemma}
\begin{proof}
The homology is calculated in \cite{nr}.
Let $k^2=0$.
Then (i) $dh+hd=0$ on $\onepw$;
(ii) $\ker(d: \onepw^0 \to \onepw^1)=0$ because $H^0(\onepw)=0$ and $\onepw^{-1}=0$;
(iii) $\onepw^3 = \image(d: \onepw^2 \to \onepw^3)$ because $H^3(\onepw)=0$ and $\onepw^4=0$.
They imply \eqref{eq:kih}:
if $x \in \onepw^1$ and $dx=0$
then $0 = hdx = -dhx$ by (i) hence $hx=0$ by (ii);
if $x \in \onepw^3$ then $x = dy$ for $y \in \onepw^2$ by (iii)
hence $hx = hdy = -dhy$ by (i) so $hx$ is in the image of $d$.
By (i), $h$ induces a map \eqref{eq:h2121};
that this is an isomorphism follows by expanding
$d$ and $h$ to first order around the given momentum $k$ and using
\eqref{eq:hxs}, details omitted.
The proof does not use $h^2 = 0$.
\qed\end{proof}


\section{Definition of $\bv{\Qx}$-algebras}\label{sec:ginfMAIN}

Section \ref{sec:recapnew}
contains an account of $\ginf$- and $\bvinf$-algebras following \cite{tt},
but restricted to the special case in \cite[Remark 2.8]{tt},
\cite[Proposition 24]{gtv}.
In Section \ref{subsec:mod}
this definition is deformed
using a quadratic element in a Hopf algebra.
Section \ref{subsec:reduction}
is about a special case where all but certain operations are zero,
to be applied to Yang-Mills in Section \ref{sec:proofmain}.

\begin{remark}[Dualization]\label{remark:dualspaces}
Following \cite{tt}, the account below avoids the cofree Gerstenhaber coalgebra.
This requires the following clarification.
Given a graded vector space $V$,
the algebraic structures we define are collections of maps
$V^\ast \to (V^\ast)^{\otimes n}$
subject to certain quadratic axioms; here $V^\ast$ is the dual.
In the unlikely case that $\dim_\C V < \infty$ this is,
by taking duals, equivalent to collections of maps
$V^{\otimes n} \to V$ subject to the corresponding dual axioms.
Whether $V$ is finite-dimensional or not,
we agree that the proper way to interpret the definitions is
that \emph{one must dualize the axioms, rather than the maps,
           to get axioms for the collection of maps $V^{\otimes n} \to V$}.
 It is still sometimes convenient to say that we dualize the maps,
 but this is always meant formally and with the reservations in this remark.
\end{remark}


\subsection{Recap of ordinary $\ginf$- and $\bvinf$-algebras}\label{sec:recapnew}

\emph{Graded vector space.} The ground field is $\C$.
The constructions in this section require a graded vector space
 $V = \oplus_{i \in \Z} V^i$ as input data.
Associated objects such as $V \otimes V$
or the dual $V^\ast = \Hom_\C(V,\C)$
carry an induced grading;
a $\C$-linear map from one graded vector space
to another has degree $i$ iff it raises the degree by $i$,
thus $(V^\ast)^i \simeq (V^{-i})^\ast$.
\step
\emph{Gerstenhaber algebra.}
Briefly, a Gerstenhaber algebra
is a graded vector space $A$ that has two compatible structures:
\begin{itemize}
\item A graded commutative associative algebra or \emph{gca} structure $A^i \otimes A^j \to A^{i+j}$.
\item A degree $-1$ graded Lie algebra or \emph{gLa} structure $[-,-] : A^i \otimes A^j \to A^{i+j-1}$.
\end{itemize}
The gca multiplication is denoted by juxtaposition.
In full detail, the gca axioms, the gLa axioms and the compatibility axiom
say that the following must be zero for all $a,b,c \in A$:
\begin{subequations}\label{eq:gax}
\begin{align}
& ab - (-1)^{ab} ba\\
& (ab)c - a(bc) \displaybreak[0] \\
& [a,b] +(-1)^{(a+1)(b+1)} [b,a] \displaybreak[0] \\
& [a,[b,c]] + (-1)^{(a+1)(b+c)} [b,[c,a]] + (-1)^{(c+1)(a+b)} [c,[a,b]]\\
& [ab,c] - a[b,c] - (-1)^{ab}b[a,c]
\end{align}
\end{subequations}
A Gerstenhaber algebra may or may not have a unit.
\step
Equivalently,
a Gerstenhaber algebra is an algebra of the homology of the little disks operad.
The little disks operad is a sequence of topological spaces
indexed by an integer $n$,
closely related to configuration spaces of $n$ points in $\R^2$,
together with an operad structure,
whose singular homology neatly captures the Gerstenhaber algebra axioms, see \cite{ds}.
\step
\emph{Free Gerstenhaber algebra\footnote{%
The free Gerstenhaber algebra of $V^\ast$ comes in for the following reason.
In general, see \cite{dtt},
to every quadratic operad $\operad$ one can associate 
its Koszul dual cooperad $\operad^{\vee}$ \cite{gk}.
Some quadratic operads are Koszul \cite{gk}, an adjective.
For conceptual reasons \cite{dtt},
if $\operad$ is quadratic and Koszul,
a $\operadinf$-algebra on $V$ is defined to be a codifferential of degree $1$
on the cofree $\operad^{\vee}$-coalgebra cogenerated by $V$.
The Gerstenhaber operad $\gerst$
(axioms \eqref{eq:gax}) is quadratic and Koszul \cite{gj} (like Ass, Lie, Com).
It is Koszul self-dual (like Ass, unlike Lie, Com),
so, modulo degree shifts,
a $\gerst^{\vee}$-coalgebra is the same thing as a Gerstenhaber coalgebra
(operations and axioms linearly dual to \eqref{eq:gax}).
Hence a $\ginf$-algebra on $V$ is a codifferential of degree $1$
on the cofree Gerstenhaber coalgebra cogenerated by $V$;
this is formally the same thing as a differential of degree $1$
on the free Gerstenhaber algebra generated by $V^\ast$,
cf.~\cite[Proposition 4.2.14]{gk} and \cite{tt} and Remark \ref{remark:dualspaces}.}.}
Let $V$ be a graded vector space.
The free Gerstenhaber algebra of $V^\ast$ \emph{without} unit will be denoted
$GV^\ast$,
with a degree zero map $V^\ast \hookrightarrow GV^\ast$.
The free Gerstenhaber algebra $GV^\ast$ is given by the following construction (sketch):
\begin{itemize}
\item Terms:
Linear combinations of `Gerstenhaber words'
that in this context are build from
elements of $V^\ast$ combined using two formal operations:
the gca product to-be denoted by juxtaposition;
the gLa product to-be denoted by $[-,-]$.
The products have no properties yet,
except that all $\C$-linearities are already understood.
\item Relations: Quotient by the relations in the Gerstenhaber algebra axioms \eqref{eq:gax}.
\end{itemize}
The empty term is not allowed (no unit).
An example of a term is $[v_1[v_2,v_3],v_4]v_5$
with $v_i \in V^\ast$; its degree is $\sum_i |v_i|-2$,
and it is equal to\footnote{%
Additional brackets are omitted, using the associativity axiom in \eqref{eq:gax}.}
\[
   v_1 [[v_2,v_3],v_4]v_5 + (-1)^{v_1(v_2+v_3-1)} [v_2,v_3] [v_1,v_4]v_5
\]
in $GV^\ast$. Like in this example,
one can always move all gca products to the outside.
\step
Producing a $\C$-basis for $GV^\ast$ is an interesting topic.
This acquires a geometric flavor
using the isomorphic\footnote{%
This way of associating to a vector space a
free algebra is known as a `Schur functor', see e.g.~\cite{lv}.
The action of the permutation group $S_n$
is the obvious one but involves degree-dependent signs.
}
$GV^\ast = \oplus_{n>0} H(D_n) \otimes_{\C[S_n]} (V^\ast)^{\otimes n}$,
with $D_n$ the space of $n$ little disks, $S_n$ the symmetric group.
The size of $GV^\ast$ is therefore related to the size of the singular homology $H(D_n)$
 which may be explicitly calculated, together with the cohomology ring,
 using the de Rham pairing, that is, integration of forms \cite{ds}.
\step
\emph{The $n$-degree and the decreasing filtration.}
The $n$-degree is given by the word length on $GV^\ast$,
so $n|_{V^\ast} = 1$ and
$n(ab) = n([a,b]) = n(a) + n(b)$ for all $n$-homogeneous $a,b \in GV^\ast$.
Abbreviate $n(a) = n_a$.
The $n$-degree
is different from the primary grading on $GV^\ast$ coming from the grading on $V$.
Define a decreasing filtration by letting $\fil{n} = \fil{n} GV^\ast$
be the $\C$-span of terms of $n$-degree $\geq n$.
Then $\fil{1} = GV^\ast$
and $\fil{n}\fil{m} \subset \fil{n+m}$ and $[\fil{n},\fil{m}] \subset \fil{n+m}$;
and $\bigcap_{n \geq 1} \fil{n} = 0$.
An endomorphism is in $\fil{n}\End_\C(GV^\ast)$
iff it maps $\fil{m} \to \fil{m+n}$ for all $m$.
\begin{remark}[Completion]\label{remark:completion}
The completion relative to the filtration $\fil{\bullet} GV^\ast$
allows countably infinite linear combinations analogous to formal power series.
In the remainder of this section,
one should replace $GV^\ast$ by this completion,
however we will not introduce separate notation for this.
The completion is also filtered.
\end{remark}
\emph{Derivations.}
For every Gerstenhaber algebra $A$,
a homogeneous element\footnote{%
Here and below, $\End_\C$ are simply the $\C$-linear maps,
not endomorphisms of algebras.} $\delta \in \End_\C(A)$
is called a Gerstenhaber derivation, or simply derivation, if
\begin{align*}
\delta(ab) & = \delta(a)b + (-1)^{\delta a} a\delta(b)\\
\delta([a,b]) & = [\delta(a),b] + (-1)^{\delta(a-1)} [a,\delta(b)]
\end{align*}
The space of all Gerstenhaber derivations is denoted 
$\Der(A)$ and it is a gLa using the graded commutator,
$[\delta_1,\delta_2] = \delta_1\delta_2 - (-1)^{\delta_1\delta_2} \delta_2\delta_1$
for all $\delta_1,\delta_2 \in \Der(A)$,
so it is a sub-gLa of $\End_\C(A)$;
one checks that $[\delta_1,\delta_2]$ is itself a Gerstenhaber derivation.
\step
\emph{Definition of $\ginf$-algebra.}
For every graded vector space $V$ define \cite{tt}:
\[
\text{$\ginf$-algebra structure on $V$}
\qquad
\Longleftrightarrow
\qquad
\text{$\delta \in \Der^1(GV^\ast)$ with $[\delta,\delta] = 0$}\\
\]
A few comments are in order:
\begin{itemize}
\item $\delta$ being of degree one, the Maurer-Cartan equation $[\delta,\delta] = 0$
is equivalent to $\delta^2 = 0$. Therefore, a $\ginf$ structure
is a derivation $\delta$ of degree one
that is also a differential.
\item There are analogous definitions of $\ainf$, $\cinf$, $\linf$ using
respectively the free graded associative or tensor algebra;
the free gLa;
the free gca,
without unit. Note \cite{tt} that $GV^\ast$ contains
a copy of the free gLa as a quotient gLa and
a copy of the free dgca as a quotient gca,
and every Gerstenhaber derivation descends to these quotients.
Therefore every $\ginf$ structure contains both a $\cinf$ and an $\linf$ structure.
\item $\delta$ is determined by its restriction to $V^\ast$ and this data 
is unconstrained\footnote{%
By an `infinitesimal' version of the universal property of $GV^\ast$.},
\begin{equation}\label{eq:dparx}
\Der(GV^\ast) \simeq \Hom_\C(V^\ast,GV^\ast)
\end{equation}
By decomposing $GV^\ast$ one obtains a decomposition of $\delta$.
Recall Remarks \ref{remark:dualspaces}, \ref{remark:completion} here.
\end{itemize}
A $\ginf$-algebra with $\delta(V^\ast) \subset V^\ast V^\ast + [V^\ast,V^\ast]$
is an ordinary Gerstenhaber algebra\footnote{%
Namely, after dualization,
it is a Gerstenhaber algebra on $V[-2]$ where $V[-2]^i = V^{i-2}$.
The part mapping into $V^\ast V^\ast$ encodes the gLa product,
the part mapping into $[V^\ast,V^\ast]$ encodes the gca product.}.
\step
\emph{The canonical operator $\tt$.}
The free Gerstenhaber algebra has an interesting canonical structure \cite{tt}.
Namely, define $\tt \in \End^{-1}_\C(GV^\ast)$ by requiring\footnote{%
This operator is denoted $\Delta_{\text{Lie}}$ in \cite{tt}.
In this paper, $\Delta$ is reserved
for a coproduct in Section \ref{subsec:mod}.}
\begin{equation}\label{eq:bvtt}
\begin{aligned}
\tt|_{V^\ast} & = 0\\
\tt(ab) & =
(-1)^a [a,b] + \tt(a)b + (-1)^a a \tt(b)\\
\tt([a,b]) & = [\tt(a),b] + (-1)^{a-1} [a,\tt(b)]
\end{aligned}
\end{equation}
for all $a,b \in GV^\ast$. This recursively fixes how $\tt$ acts on all terms,
and one checks that it is well-defined
meaning compatible with the relations \eqref{eq:gax}\footnote{%
Direct calculation. See also \cite{tt}.}. Further:
\begin{itemize}
\item $\tt$ has degree $-1$ because $|[a,b]| = |ab|-1$.
\item $\tt^2 = 0$\footnote{%
Check
$\tt^2|_{V^\ast}=0$
and $\tt^2(ab) = \tt^2(a)b + a\tt^2(b)$
and $\tt^2([a,b]) = [\tt^2(a),b] + [a,\tt^2(b)]$.}
and therefore\footnote{%
Explicitly, set
$d_{\tt}(e) = \tt e - (-1)^e e \tt$ for all 
$e \in \End_\C(GV^\ast)$.}
$d_{\tt} = [\tt,-] \in\End_\C^{-1}(\End_\C(GV^\ast))$
is a differential, $d_{\tt}^2 = 0$,
and it satisfies the Leibniz rule
for the gLa structure on $\End_\C(GV^\ast)$\footnote{%
Explicitly, 
the Leibniz rule is
$d_{\tt}([e_1,e_2]) = [d_{\tt}(e_1),e_2]
+ (-1)^{e_1}[e_1,d_{\tt}(e_2)]$ for all $e_1,e_2 \in \End_\C(GV^\ast)$.
It is trivially true
since $d_{\tt}$ is given by a graded commutator.
}.
\item If $\delta \in \Der(GV^\ast)$
then $d_{\tt}(\delta) \in \Der(GV^\ast)$\footnote{%
By direct calculation, one finds
$d_{\tt}(\delta)|_{V^\ast}=0$ and
$d_{\tt}(\delta)(ab) = (d_{\tt}(\delta)(a))b
  + (-1)^{(\delta-1)a} a (d_{\tt}(\delta)(b))$
and
$d_{\tt}(\delta)([a,b]) = [d_{\tt}(\delta)(a),b] +
   (-1)^{(\delta-1)(a-1)} [a, d_{\tt}(\delta)(b)]$,
for all $a,b \in GV^\ast$.}.
So, $d_{\tt} \in \End^{-1}_\C(\Der(GV^\ast))$.
\end{itemize}
Hence $\Der(GV^\ast)$
is a differential graded Lie algebra
with a differential $d_{\tt}$ of degree $-1$.
\step
\emph{Definition of $\bvinf$-algebra.}
With \cite[Remark 2.8]{tt}, \cite[Proposition 24]{gtv} define:
\[
\text{$\bvinf$-algebra structure on $V$}
\quad
\Longleftrightarrow
\quad
\left\{
\begin{aligned}
&\text{$\delta^i \in \Der^i(GV^\ast)$ for $i=1,-1,-3,-5,\ldots$ with}\\
& \qquad\qquad\qquad d_{\tt}\delta + \tfrac12 [\delta,\delta] = 0\\
&\text{where $\delta = \delta^1 + \delta^{-1} + \delta^{-3} + \ldots$}
\end{aligned}
\right.
\]
For degree reasons it is necessary that $[\delta^1,\delta^1] = 0$,
hence every $\bvinf$ structure contains a $\ginf$ structure.
The next equation is $d_{\tt}\delta^1 + [\delta^1,\delta^{-1}] = 0$, and so forth.

\subsection{$\bv{\Qx}$-algebras,
with axioms deformed by a quadratic element $\Qx$}\label{subsec:mod}

Here we
modify the definition of a $\bvinf$-algebra in the case when $V$ is also a
module over a cocommutative Hopf algebra $\Hopf$,
and a central `quadratic' element in $\Hopf$ is fixed\footnote{%
Operads $\operad$ in the category of $\Hopf$-modules,
where $\Hopf$ is a graded cocommutative Hopf algebra,
have been discussed in the literature \cite{sawa,be}.
One can then form a semidirect product of operads $\operad \rtimes \Hopf$,
see \cite{sawa}.
The action of $\Hopf$ on $\operad$ can encode interesting information,
for instance the semidirect product of the operad of Gerstenhaber algebras
with the Hopf algebra of dual numbers, acting in a certain way, is isomorphic
to the operad of BV algebras \cite{be}.
It is not excluded that some of this theory can be applied here.}.
\step
\emph{Hopf module and quadratic element.}
In this section, $V$ is both a graded vector space
and an $\Hopf$-module.
Here $\Hopf$ is a cocommutative Hopf algebra with coproduct denoted $\Delta$.
The Hopf algebra is not graded and
leaves the grading on $V$ invariant.
Further, the constructions in this section require the choice of
an element $\Qx \in \Hopf$ with these properties:
\begin{itemize}
\item $\Qx$ satisfies the 7-term identity
\begin{equation}\label{eq:7term}
   \Delta^2 \Qx
   - (\one + \sigma + \sigma^2) (1 \otimes \Delta \Qx)
   + (\one + \sigma + \sigma^2) (1 \otimes 1 \otimes \Qx)
\;=\; 0
\end{equation}
where $1 \in \Hopf$ is the unit and where
$\sigma: \Hopf^{\otimes 3} \to \Hopf^{\otimes 3}$,
$h_1 \otimes h_2 \otimes h_3 \mapsto h_2 \otimes h_3 \otimes h_1$.
\item $\Qx$ is in the kernel of the counit $\Hopf \to \C$ of the Hopf algebra\footnote{%
Actually this follows from \eqref{eq:7term}
by applying $(\one \otimes \eps)(\one \otimes \one \otimes \eps)$
where $\eps: \Hopf \to \C$ is the counit,
and using the Hopf algebra axioms.
The counit plays the role of the trivial representation of a Hopf algebra.}.
\item $\Qx$ is in the center of the Hopf algebra.
\end{itemize}
The motivating example is
when $\Hopf$ is a polynomial algebra\footnote{%
On $\C[x_1,\ldots,x_n]$
the coalgebra structure is determined by
$\Delta x_i = x_i \otimes 1 + 1 \otimes x_i$.}
and $\Qx$ is a polynomial of degree at most two, without constant term.
Another example,
not in general commutative but cocommutative,
is the universal enveloping algebra $\Hopf = U\gx$\footnote{%
On $U\gx$
the coalgebra structure is determined by
$\Delta g = g \otimes 1 + 1 \otimes g$ for all $g \in \gx$.}
of a Lie algebra $\gx$,
and $\Qx$ an element of degree at most two, without constant term, in the center of $U\gx$.
\step
\emph{Hopf module structure on $GV^\ast$.}
Since $V$ is an $\Hopf$-module,
so is $(V^\ast)^{\otimes n}$\footnote{%
For $n=0$ the module structure is the counit, but we do not need $n=0$ here.}.
By cocommutativity of $\Hopf$,
on $(V^\ast)^{\otimes n}$
the $S_n$-module and $\Hopf$-module structures are compatible\footnote{%
That is, one has a bimodule.}.
Hence $GV^\ast$ is an $\Hopf$-module.
The following maps are $\Hopf$-equivariant:
the inclusion $V^\ast \hookrightarrow GV^\ast$;
the gca and gLa products as maps $GV^\ast \otimes GV^\ast \to GV^\ast$;
the map $\tt : GV^\ast \to GV^\ast$.
\step
As an example, if $h \in \Hopf$ and $v_i \in V^\ast$ then
\[
h ([v_1[v_2,v_3],v_4]v_5)
= [(h_1v_1)[(h_2v_2),(h_3v_3)],(h_4v_4)](h_5v_5)
\]
where $\Delta^4h = h_1 \otimes \cdots \otimes h_5$ is Sweedler notation
for the iterated coproduct $\Delta^4: \Hopf \to \Hopf^{\otimes 5}$.
\step
\emph{Equivariant derivations.}
Denote by
$\Der(GV^\ast)^{\Hopf} \subset \Der(GV^\ast)$
the sub-graded Lie algebra of $\Hopf$-equivariant elements.
Note that $d_{\tt}$ maps this sub-gLa to itself.
\step
\emph{The operator $\qq$.}
Here the fixed element $\Qx \in \Hopf$ is used
to make a definition inspired by and superficially similar to \eqref{eq:bvtt}.
Namely define
$\qq \in \End^1_\C(GV^\ast)$
by 
\begin{equation}\label{eq:qq81}
\begin{aligned}
\qq|_{V^\ast} & = 0\\
\qq(ab) & =
\qq(a)b + (-1)^a a \qq(b)\\
\qq([a,b]) & = 
(-1)^a \big(
\Qx(ab) - (\Qx a)b - a(\Qx b)
\big)
+ [\qq(a),b] + (-1)^{a-1} [a,\qq(b)]
\end{aligned}
\end{equation}
for all $a,b \in GV^\ast$.
This recursively fixes how $\qq$ acts on all terms,
and one checks that it is well-defined meaning compatible with the relations
\eqref{eq:gax};
this calculation requires the 7-term identity \eqref{eq:7term}
and is in Appendix \ref{app:proofs}. Further:
\begin{itemize}
\item $\qq$ has degree $+1$ because $|ab| = |[a,b]| + 1$.
\item $\qq: GV^\ast \to GV^\ast$ is $\Hopf$-equivariant,
by the assumption that $\Qx$ is in the center.
\item $\qq^2 = 0$\footnote{%
Check
$\qq^2|_{V^\ast}=0$
and $\qq^2(ab) = \qq^2(a)b + a\qq^2(b)$
and $\qq^2([a,b]) = [\qq^2(a),b] + [a,\qq^2(b)]$,
which implies $\qq^2 = 0$.
The 7-term identity \eqref{eq:7term} is not needed for this calculation.
}
and therefore
$d_{\qq} = [\qq,-] \in \End^1_\C(\End_\C(GV^\ast))$
is a differential, $d_{\qq}^2 = 0$,
and it satisfies the Leibniz rule for the gLa structure on $\End_\C(GV^\ast)$.
\item
If $\delta \in \Der(GV^\ast)^{\Hopf}$
then $d_{\qq}(\delta) \in \Der(GV^\ast)^{\Hopf}$.
So, $d_{\qq} \in \End^1_\C(\Der(GV^\ast)^{\Hopf})$.
\end{itemize}
The last
claim, which is only for $\Hopf$-equivariant $\delta$,
is also proved in Appendix \ref{app:proofs}.
\step
\emph{A curvature operator.}
The sum $d_{\tt} + d_{\qq}$ is not a differential
because
$d_{\tt}d_{\qq} + d_{\qq}d_{\tt} \neq 0$.
Its restriction to $\Hopf$-equivariant
derivations is explicitly given by
\begin{equation}\label{eq:dddd}
d_{\tt}d_{\qq} + d_{\qq}d_{\tt} = -[\curv,-]
\qquad \text{on $\Der(GV^\ast)^{\Hopf}$}
\end{equation}
where, by definition, $\curv \in \Der^0(GV^\ast)^{\Hopf}$ is given by
\begin{align*}
\curv|_{V^\ast} & = \Qx \one \\
\curv(ab) & = \curv(a)b + a\curv(b)\\
\curv([a,b]) & = [\curv(a),b] + [a,\curv(b)]
\end{align*}
Note that $\curv$ is $\Hopf$-equivariant since $\Qx$ is in the center.
One finds
$d_{\tt}(\curv) = d_{\qq}(\curv) = 0$.
See Appendix \ref{app:proofs} for proofs.
In summary, on $\Der(GV^\ast)^{\Hopf}$,
the map $d_{\tt}+d_{\qq}$ is not a differential,
but together with $-\curv$ one has a
curved differential $\Z/2\Z$-graded Lie algebra.
The Maurer-Cartan equation in this curved dgLa
motivates the following definition.
\step
\emph{Definition of $\bv{\Qx}$-algebra.} Set
\begin{equation}\label{eq:defheq}
\text{$\bv{\Qx}$-algebra structure on $V$}
\quad
\Longleftrightarrow
\quad
\left\{
\begin{aligned}
&\text{$\delta^i \in \Der^i(GV^\ast)^{\Hopf}$ for $i=1,-1,-3,\ldots$ with}\\
& \qquad\quad
-\curv + (d_{\tt} + d_{\qq}) \delta + \tfrac12 [\delta,\delta] = 0\\
&\text{where $\delta = \delta^1 + \delta^{-1} + \delta^{-3} + \ldots$}
\end{aligned}
\right.
\end{equation}
The first few degree levels of these equations are\footnote{%
In the special case $\delta^{-3} = \delta^{-5} = \ldots = 0$,
only the three displayed equations are non-vacuous.}
\begin{equation}\label{eq:defheq2}
\begin{aligned}
\tfrac12 [\delta^1,\delta^1] & = -d_{\qq} \delta^1\\
d_{\tt} \delta^1 + [\delta^1,\delta^{-1}] & = \curv - d_{\qq} \delta^{-1}\\
d_{\tt} \delta^{-1}
+ [\delta^1,\delta^{-3}] + \tfrac12[\delta^{-1},\delta^{-1}] & = -d_{\qq}\delta^{-3}\\
& \;\;\vdots
\end{aligned}
\end{equation}
The terms on the right correspond to the deformation by $\Qx$.
Already the first equation is deformed,
so that $\delta^1$ is not necessarily a $\ginf$-algebra\footnote{%
Since it only involves $\delta^1$, we could say that it defines
a $\ginfq{\Qx}$-algebra.}.
But by the definition of $d_{\qq}$,
note that $\delta^1$ must still contain a $\cinf$-algebra, including a differential.


\subsection{The auxiliary operator $\Gamma$} \label{sec:gammax}
Here a $\Gamma \in \End_\C^1(\Der(GV^\ast))$ is defined
such that
$\Gamma^2 = 0$
and 
\begin{equation}\label{eq:dGGd}
d_{\tt} \Gamma + \Gamma d_{\tt} = \one_{>0}
\end{equation}
where, for every derivation $\delta$, set $\one_{>0}(\delta)=0$ if $n_{\delta}=0$
and $\one_{>0}(\delta)=\delta$ if $n_{\delta}>0$.
The $n$-degree was defined in Section \ref{sec:recapnew};
a derivation $\delta$ has $n$-degree equal to $n_{\delta}$
iff it increases the $n$-degree by $n_{\delta}$.
Further, $\Gamma$
maps $\Hopf$-equivariant derivations to $\Hopf$-equivariant derivations.
\step
\emph{Construction of $\Gamma$.}
Define $\hh \in \End^1_\C(GV^\ast)$ by
\begin{align*}
\hh|_{V^\ast} & = 0 \\
\hh(ab) & =
\hh(a)b + (-1)^a a \hh(b)\\
\hh([a,b]) & = 
(-1)^a
n_a n_b ab  + [\hh(a),b] + (-1)^{a-1} [a,\hh(b)]
\end{align*}
for all $n$-homogeneous $a,b \in GV^\ast$.
It is well-defined on $GV^\ast$.
Both $\tt$ and $\hh$ preserve the $n$-degree.
By direct calculation,
$\hh^2 = 0$,
and $\uu = \tt\hh+\hh\tt \in \End^0_\C(GV^\ast)$ is given by\footnote{%
Check
$\uu|_{V^\ast} = 0$
and $\uu(ab) =
n_an_b ab
+ \uu(a)b + a \uu(b)$
and $\uu([a,b]) = n_a n_b [a,b]  + [\uu(a),b] + [a,\uu(b)]$.}
\begin{equation}\label{eq:unn}
\uu(a) = \tfrac12 n_a(n_a-1) a
\end{equation}
Now, if $\delta \in \Der(GV^\ast)$ then
 $[\hh,\delta]$ need not be a derivation.
Hence, for every derivation $\delta$ of $n$-degree equal to $n_{\delta}$,
let $\Gamma(\delta) \in \Der(GV^\ast)$ be the unique derivation with
\[
\Gamma(\delta)|_{V^\ast}
= \begin{cases}
0 & \text{if $n_{\delta}=0$}\\
\tfrac{2}{(n_{\delta}+1)n_{\delta}} [\hh, \delta]|_{V^\ast} & \text{if $n_{\delta}>0$}
\end{cases}
\]
Then $\Gamma^2 = 0$ by $\hh^2 = 0$.
To see \eqref{eq:dGGd}, suppose $n_{\delta}>0$; the case $n_{\delta}=0$ is trivial.
Then $\delta(V^\ast)$ has $n$-degree equal to $n_{\delta}+1$.
Set $\delta' = (d_{\tt}\Gamma + \Gamma d_{\tt})\delta$
which is a derivation, $n_{\delta'} = n_{\delta}$, and
$\delta'|_{V^\ast} = \frac{2}{(n+1)n}(\uu \delta)|_{V^\ast}$,
and therefore \eqref{eq:unn} implies $\delta' = \delta$ as required.
\step
\emph{Another property of $\Gamma$.} If $\delta \in \Der(GV^\ast)$
and $n_{\delta}=0$ then $[\delta,-]$ commutes with $\Gamma$ up to sign.
More explicitly, using the definition of $\Gamma$, this means that
\begin{equation}\label{eq:dvv}
\delta(V^\ast) \subset V^\ast
\qquad
\Longrightarrow
\qquad
\Gamma[\delta,-] = (-1)^\delta [\delta,\Gamma-]
\end{equation}

\subsection{A special class}\label{subsec:reduction}

\emph{The special case.}
Throughout this section,
\[
	\delta = \textstyle\sum_{k \geq 1} \mu_k + \textstyle\sum_{k \geq 2} \nu_k
\]
where $\mu_k,\nu_k \in \Der^{\text{odd}}(GV^\ast)^{\Hopf}$ 
satisfy
\begin{equation}\label{eq:munux}
\begin{aligned}
\mu_k(V^\ast) & \subset (V^\ast)^{\otimes k}\\
\nu_k(V^\ast) & \subset [V^\ast,V^\ast] (V^\ast)^{\otimes (k-2)}
\end{aligned}
\end{equation}
It is always understood
that $\mu_k = 0$ for $k \leq 0$
respectively $\nu_k = 0$ for $k \leq 1$.
Further assume that almost all operations have degree one:
$\mu_k, \nu_k \in \Der^1$ for all $k \geq 2$ and
\[
\mu_1 = d + h
\]
where $d \in \Der^{1}$ and $h \in \Der^{-1}$.
(In the notation of \cite{gtv}, this means that
only the structure maps $m^0_{1,\ldots,1}$, $m^0_{1,\ldots,1,2}$ and $m_1^1$ can be nonzero.)
The notation $d$, $h$ is merely suggestive for the synonymous operations
in the Yang-Mills dgca $\ax$;
beware that in this section $d$ and $h$ are
not primarily maps $V^\ast \to V^\ast$
but as derivations that map $V^\ast \to V^\ast$.
Using \eqref{eq:munux},
\begin{equation}\label{eq:dx23}
d_{\tt} \mu_1 = 0
\qquad
d_{\tt} \nu_2 = 0
\qquad
d_{\qq} \mu_k = 0
\qquad
\Gamma \mu_k = 0
\end{equation}
\step
\emph{Decomposition of the axioms.}
Define $A_k, B_k, C_k \in \Der^{\text{even}}(GV^\ast)^{\Hopf}$ by
\begin{alignat*}{5}
A_k & = -\curv \delta_{k1} &\;& + d_{\qq} \nu_k \;&&+\tfrac12 && \textstyle\sum_{m+n-1=k} [\mu_m,\mu_n]\\
B_k & =  &\;& + d_{\tt} \mu_k \;&&\hskip2.5pt +&& \textstyle\sum_{m+n-1=k} [\mu_m,\nu_n]\\
C_k & = &\;& + d_{\tt} \nu_k \;&&+\tfrac12 && \textstyle\sum_{m+n-1=k} [\nu_m,\nu_n]
\end{alignat*}
with $\delta_{k1}$ the Kronecker delta.
The sums are finite, and
note that $A_k$, $B_k$, $C_k$ are zero except when $k \geq 1$, $k \geq 2$, $k \geq 3$
respectively using \eqref{eq:dx23}. By \eqref{eq:dx23},
\[
-\curv + (d_{\tt}+d_{\qq}) \delta 
+ \tfrac12 [\delta,\delta]
\;=\;
\textstyle\sum_{k \geq 1} A_k + \sum_{k \geq 2} B_k + \sum_{k \geq 3} C_k
\]
Hence if all $A_k$, $B_k$, $C_k$ vanish then $\delta$ is a $\bv{\Qx}$-algebra structure,
see \eqref{eq:defheq}.
The converse is also true, because $A_k$, $B_k$, $C_k$
map $V^\ast$ to linearly independent sectors of $GV^\ast$:
\begin{align*}
A_k(V^\ast) & \subset (V^\ast)^{\otimes k}\\
B_k(V^\ast) & \subset [V^\ast,V^\ast] (V^\ast)^{\otimes (k-2)}\\
C_k(V^\ast) & \subset [[V^\ast,V^\ast],V^\ast] (V^\ast)^{\otimes (k-3)} 
+ \underbrace{[V^\ast,V^\ast][V^\ast,V^\ast] (V^\ast)^{\otimes(k-4)}}_{%
\text{only for $k \geq 4$}} 
\end{align*}
The goal of this section is to spell out
the linear dependencies among the $A_k$, $B_k$, $C_k$
and to reduce the axioms
to a form that one can check for Yang-Mills, in Section \ref{sec:proofmain}.
\step
\emph{Linear dependencies.} One finds
\begin{align}
\label{eq:claim1}
\Gamma A_k & = 0\\
\label{eq:claim2}
d_{\qq} A_k & = 0\\
\label{eq:claim3}
d_{\qq} B_k & = \textstyle\sum_{m+n-1=k} \rule{50pt}{0pt}[A_m, \mu_n] \displaybreak[0] \\
\label{eq:claim4}
d_{\tt} A_k + d_{\qq} C_k
& = \textstyle\sum_{m+n-1=k} [A_m,\nu_n] + [B_m,\mu_n]\\
\label{eq:claim5}
d_{\tt} B_k \rule{38pt}{0pt} & = \textstyle\sum_{m+n-1=k} [B_m,\nu_n] + [C_m,\mu_n]\\
\label{eq:claim6}
d_{\tt} C_k \rule{38pt}{0pt} & = \textstyle\sum_{m+n-1=k} [C_m,\nu_n]
\end{align}
To see this, use \eqref{eq:dddd}; \eqref{eq:dx23};
$d_{\tt}(\curv) = d_{\qq}(\curv) = 0$;
$d_{\tt}^2 = d_{\qq}^2 = 0$;
the Leibniz rule for $d_{\tt}$ and $d_{\qq}$;
and the Jacobi identity by which $[y,[y,y]] = 0$ for all $y \in \Der^{\text{odd}}(GV^\ast)$.

\begin{lemma}[Reduced axioms]\label{lemma:redx}
$\delta$ is a $\bv{\Qx}$-algebra iff
\begin{align}\label{eq:abc}
A_1 & = 0 &
& \text{$\Gamma B_k = 0$\; for all $k \geq 2$} &
& \text{$\Gamma C_k = 0$\; for all $k \geq 3$}
\end{align}
More precisely,
for every integer $K$,
if \eqref{eq:abc} hold for $k < K$ then $A_k=B_k=C_k=0$ for $k < K$,
and if in addition $\Gamma C_K = 0$ then $C_K=0$. 
\end{lemma}
\begin{proof}
We prove $\Leftarrow$ by induction on $k$.
At each step we show that $A_k$, $B_k$, $C_k$ vanish.
Case $k=1$: Vacuous since $B_1=C_1=0$.
Case $k=2$: Use \eqref{eq:claim5}
to see that $d_{\tt} B_2 = 0$ (note that $\nu_1 = 0$ and $C_2=0$)
and by assumption $\Gamma B_2 = 0$,
which, since $B_2$ has $n$-degree $+1$ and using \eqref{eq:dGGd}, implies $B_2 = 0$.
Now \eqref{eq:claim4} implies $d_{\tt} A_2 = 0$
which with \eqref{eq:claim1} and \eqref{eq:dGGd} implies $A_2 = 0$.
Case $k=3$:
$d_{\tt} C_3 = 0$ by \eqref{eq:claim6}
which with $\Gamma C_3 =0$ and \eqref{eq:dGGd}
implies $C_3 = 0$.
Now \eqref{eq:claim5} implies $d_{\tt} B_3 = 0$
which with $\Gamma B_3 = 0$
implies $B_3 = 0$.
Now $d_{\tt} A_3 = 0$ by \eqref{eq:claim4},
and \eqref{eq:claim1}, imply $A_3 = 0$.
For $k \geq 4$, the argument is just like for $k=3$.
\qed\end{proof}
\begin{lemma}[A special case of the special case]\label{lemma:specialspecial}
Restrict to the case when $\delta$
is determined by $\mu_1 = d + h$
as before and by a sequence $\theta_k \in \Der^2(GV^\ast)^{\Hopf}$, $k \geq 2$, with
\begin{equation}\label{eq:thk}
 \theta_k(V^\ast) \subset (V^\ast)^{\otimes k}
\end{equation}
by setting
\begin{equation}\label{eqzqkj}
\begin{aligned}
\nu_k & = d_{\tt} \theta_k\\
\mu_k & = -\Gamma [h,\nu_k]
\end{aligned}
\end{equation}
for all $k \geq 2$;
they have degree one and are consistent with \eqref{eq:munux}. Then:
\begin{itemize}
\item Part A: $\delta$ is a $\bv{\Qx}$-algebra iff
\begin{itemize}
\item $[h,h] = 0$
and $[d,h] = \curv$
and $[d,d] = 0$.
\item $b_k = 0$ where, by definition,
$b_k = \Gamma[d,\nu_k] + \Gamma \sum_{m+n-1=k,\;m\neq 1} [\mu_m,\nu_n]$,
for $k \geq 2$.
\item $c_k = 0$ where, by definition,
$c_k = \tfrac12 \Gamma \sum_{m+n-1=k} [\nu_m,\nu_n]$,
for $k \geq 3$.
\end{itemize}
\item Part B: Assuming only the first bullet in Part A,
the following implications hold:
\begin{itemize}
\item If $b_K=c_K=0$ for all $K < k$ then $[d,c_k] = 0$.
\item If $b_K=c_K=0$ for all $K < k$ and if $c_k = 0$ then $[d,b_k] = 0$.
\end{itemize}
\end{itemize}
\end{lemma}
\begin{proof}
In Part A, the first bullet means $A_1 = 0$,
the three equations corresponding to degrees $-2$, $0$, $2$ respectively.
Let
$B_k = B_k^0 \oplus B_k^2$
and $C_k = C_k^0 \oplus C_k^2$
be the decompositions by degree.
Then $\Gamma B_k^0 = \Gamma C_k^0 = 0$ by construction of $\mu_k$ and $\nu_k$,
in fact, using \eqref{eq:dGGd}:
\begin{align*}
\Gamma B_k^0 & = \Gamma (d_{\tt}\mu_k + [h,\nu_k])
= \Gamma (-d_{\tt}\Gamma + \one_{>0}) [h,\nu_k]
= \Gamma^2 d_{\tt} [h,\nu_k] = 0\\
\Gamma C_k^0 & = \Gamma d_{\tt} \nu_k = \Gamma d_{\tt}^2 \theta_k = 0
\end{align*}
Since $b_k = \Gamma B_k^2$, $c_k = \Gamma C_k^2$,
Part A follows from Lemma \ref{lemma:redx}.
Part B:
By Lemma \ref{lemma:redx},
$B_K=C_K=0$ for $K<k$, and if $c_k=0$ then also $C_k=0$.
The degree $3$ part of \eqref{eq:claim5}
(recall $\nu_1=0$)
implies $[C_k^2,d] = 0$,
so, by \eqref{eq:dvv},
 $[d, c_k] = [d, \Gamma C_k^2]
= - \Gamma [d,C_k^2] = 0$.
If $c_k=0$ then $C_k=0$,
so the degree $3$ part of \eqref{eq:claim4}
implies $[B_k^2,d] = 0$,
so
$[d, b_k] = [d, \Gamma B_k^2]
= - \Gamma [d, B_k^2] = 0$.
\qed\end{proof}
Note that $b_k,c_k \in \Der^3(GV^\ast)^{\Hopf}$ satisfy
\begin{equation}\label{eq:bkck}
\begin{aligned}
b_k (V^\ast) & \subset (V^\ast)^{\otimes k}\\
c_k(V^\ast) & \subset [V^\ast,V^\ast] (V^\ast)^{\otimes (k-2)}
\end{aligned}
\end{equation}
Note that \eqref{eq:thk} implies $\Gamma \theta_k=0$
and therefore $\Gamma \nu_k = \theta_k$
by \eqref{eq:dGGd}.
Hence \eqref{eq:dvv} implies $\Gamma[d,\nu_k] = -[d,\Gamma \nu_k] = -[d,\theta_k]$.
Therefore
$b_k = -[d,\theta_k] + q_k$
where, by definition,
\begin{equation}\label{eq:defq}
q_k = \textstyle\Gamma\sum_{m+n-1=k,\,m\neq 1}[\mu_m,\nu_n]
\end{equation}
Part B and $[d,d]=0$ give the implication
\begin{equation}\label{eq:eqq}
\text{If $b_K = c_K = 0$ for all $K < k$ and if $c_k=0$
then $[d,q_k] = 0$.}
\end{equation}

\section{A vanishing theorem}\label{sec:vanish}
\newcommand{\talg}{{\tc{qcolor}{T}}}
\newcommand{\lalg}{{\tc{scolor}{\Lambda}}}
\newcommand{\elzero}{{\tc{hcolor}{\phi}}}
\newcommand{\elone}{{\tc{qcolor}{\psi}}}
\newcommand{\casedist}{\qquad\quad}
\newcommand{\crx}[2]{{\tc{qcolor}{e^{#1}_{#2}}}} 
\newcommand{\VVV}{{\tc{hcolor}{U}}}

This section is stand-alone and uses its own notation. The goal is
a homology vanishing theorem,
Theorem \ref{theorem:vanish}, which is used as a lemma in Section \ref{sec:proofmain}.
\step
This commutative diagram contains the objects used in this section, with $n \geq 1$:
\begin{equation}\label{eq:cdg43}
\vcenter{\hbox{%
\begin{tikzpicture}[node distance = 6mm and 10mm, auto]
  \node (Vn) {$\VVV_n$};
  \node (Fn) [right=16mm of Vn] {$\talg_n$};
  \node (Ln) [right=of Fn] {$\lalg_n$};
  \node (An) [right=of Ln] {$A_n$};
  \node (V) [below=of Vn] {$\VVV$};
  \node (F) [below=of Fn] {$\talg$};
  \node (L) [below=of Ln] {$\lalg$};
  \node (A) [below=of An] {$A$};
  \draw[right hook->] (Vn) to node {} (Fn);
  \draw[->>] (Fn) to node {} (Ln);
  \draw[->>] (Ln) to node {} (An);
  \draw[right hook->] (V) to node {} (F);
  \draw[->>] (F) to node {} (L);
  \draw[->>] (L) to node {} (A);
  \draw[->>] (Vn) to node [swap,xshift=-4] {$f=f_n$} (V);
  \draw[->>] (Fn) to node {} (F);
  \draw[->>] (Ln) to node {} (L);
  \draw[->>] (An) to node {} (A);
\end{tikzpicture}}}
\end{equation}
Here
$\VVV_n = \oplus_{\mu=0}^3 \oplus_{j=1}^n \C \crx{\mu}{j}$
is the $4n$-dimensional vector space with basis elements $\crx{\mu}{j}$;
$\talg_n$ is the tensor algebra generated by $\VVV_n$, with unit,
graded so that the $\crx{\mu}{j}$ are in degree $1$;
$\lalg_n = \talg_n / (\text{$v^2$ with $v \in \VVV_n$})$
is the exterior algebra; and $A_n$ is the quotient algebra
\[
A_n = \lalg_n / (\text{
$\crx{0}{j} \crx{1}{j} - i\crx{2}{j} \crx{3}{j}$,\;\;
$\crx{0}{j} \crx{2}{j} - i\crx{3}{j} \crx{1}{j}$,\;\;
$\crx{0}{j} \crx{3}{j} - i\crx{1}{j} \crx{2}{j}$
with $j=1\ldots n$} )
\]
All these algebras are graded.
Set $\VVV = \VVV_1$ and for its basis we always
use the shorthand $\crx{\mu}{} = \crx{\mu}{1}$.
Set $\talg = \talg_1$, $\lalg = \lalg_1$, $A = A_1$.
Let $f=f_n$ be the $\C$-linear map
\begin{equation}\label{eq:deff}
f : \VVV_n \to \VVV,\quad \crx{\mu}{j} \mapsto \crx{\mu}{}
\end{equation}
The other vertical arrows in \eqref{eq:cdg43}
are algebra maps induced by $f$.
The Hilbert series\footnote{%
Formal power series $\sum_{k \geq 0} d_k t^k$
where $d_k$ is the $\C$-dimension of the subspace of degree $k$.} are:
\begin{align*}
      \talg_n : &\; (1-4nt)^{-1} &
\lalg_n : &\; (1+t)^{4n} = (1+4t+6t^2+4t^3+t^4)^n &
      A_n : &\; (1+4t + 3t^2)^n
\end{align*}
\begin{remark}[Algebra isomorphisms] \label{remark:prod}
For $n\geq 2$ one has\footnote{%
If $A$, $B$ are unital graded associative $\C$-algebras,
commutative or not, then the tensor product of unital associative $\C$-algebras
is the vector space $A \otimes_\C B$
with product $(a \otimes b) (a' \otimes b') = (-1)^{ba'} (aa') \otimes (bb')$.}
$\talg_n \not\simeq \talg^{\otimes n}$,
in fact $\talg_n$ is much bigger. By contrast
$\lalg_n \simeq \lalg^{\otimes n}$,
and this algebra isomorphism can be written down explicitly
using Koszul signs in a standard way.
It induces an isomorphism $A_n \simeq A^{\otimes n}$.
\end{remark}
Given a graded left $A$-module $M$
  and a graded left $A_n$-module $M_n$, define
\begin{equation}\label{eq:xkl}
X^{k,\ell}
\;=\; S^k \VVV_n^\ast \otimes \Hom_\C^\ell(M_n,M)
\end{equation}
where $S^k \VVV_n^\ast$
are the homogeneous polynomials of degree $k$ on $\VVV_n$;
and $\Hom^\ell$ are the maps of degree $\ell$.
For $\elzero \in X^{k,\ell}$,
denote its evaluation at $v \in \VVV_n$ by $\elzero_v \in \Hom_\C^{\ell}(M_n,M)$.
Define
\[
   d^{k,\ell}\;:\; X^{k,\ell} \to X^{k+1,\ell+1},\;
   \elzero \mapsto \elone
\qquad \text{where}\qquad \elone_v(m) = f(v) \elzero_v(m) - (-1)^k \elzero_v(vm)
\]
for all $m \in M_n$;
here $vm$ is $A_n$-module multiplication where
$v \in \VVV_n$ also stands for the corresponding element of $A_n$
via \eqref{eq:cdg43};
similarly $f(v)\elzero_v(m)$ is $A$-module multiplication.
This is well-defined:
$v\mapsto \elone_v$ is a polynomial homogeneous of degree $k+1$;
and $\elone_v \in \Hom^{\ell+1}$.
\begin{lemma} \label{lemma:cell}
For every integer $\ell$, the following sequence is a complex:
\[
C_\ell\quad:\quad 0 \to X^{0,\ell} \xrightarrow{\;\;d^{0,\ell}\;\;}
                        X^{1,\ell+1} \xrightarrow{\;\;d^{1,\ell+1}\;\;}
                        X^{2,\ell+2} \xrightarrow{\;\;d^{2,\ell+2}\;\;} \ldots
\]
Its homologies will be denoted $H^0(C_\ell)$, $H^1(C_\ell)$, and so forth.
\end{lemma}
\begin{proof} This calculation uses $f(v)^2 = 0 \in A$ and $v^2 = 0 \in A_n$:
\begin{multline*}
(d^{k+1,\ell+1}d^{k,\ell} \elzero)_v(m)
= f(v) (d^{k,\ell} \elzero)_v(m) - (-1)^{k+1} (d^{k,\ell}\elzero)_v(vm)\\
= f(v) \big( f(v) \elzero_v(m)
   - (-1)^k \elzero_v(vm)
     \big)
   - (-1)^{k+1} \big(
      f(v) \elzero_v(vm) - (-1)^ k \elzero_v(v^2m)
   \big)  = 0
\end{multline*}
\qed\end{proof}
\begin{theorem}[A vanishing theorem for $H^0$ and $H^1$] \label{theorem:vanish}
In both Case 1 and Case 2 below one has:
If $\ell < 0$ then $H^0(C_{\ell}) = 0$;
and if $\ell+1 < 0$ then $H^1(C_{\ell}) = 0$.
\begin{itemize}
\item Case 1: $M = A g$ and $M_n = A_n g$
where $g$ is a generator of degree $0$\footnote{%
So this is the case of two free modules of rank one.}.
\item Case 2: 
$M = (Ag_1 \oplus Ag_2 \oplus Ag_3)/S$
and
$M_n = (A_ng_1 \oplus A_ng_2 \oplus A_ng_3)/S_n$
where $g_1$, $g_2$, $g_3$ are generators of degree $0$;
$S$ is the graded left $A$-submodule generated by
\begin{equation}\label{eq:8gens}
\begin{aligned}
 & \crx{1}{} g_1 - \crx{2}{} g_2
&& \crx{1}{} g_1 - \crx{3}{} g_3
&& \crx{1}{} g_2 + \crx{2}{} g_1
&& \crx{2}{} g_3 + \crx{3}{} g_2\\
 & \crx{3}{} g_1 + \crx{1}{} g_3
&& \crx{0}{} g_1 + i\crx{2}{} g_3
&& \crx{0}{} g_2 + i\crx{3}{} g_1
&& \crx{0}{} g_3 + i\crx{1}{} g_2
\end{aligned}
\end{equation}
and $S_n$ is the graded left $A_n$-submodule generated by
the same eight relations but with each occurrence of $\crx{\mu}{}$
in \eqref{eq:8gens} replaced by $\crx{\mu}{1}$ for $\mu = 0,1,2,3$.
\end{itemize}
\end{theorem}
\begin{remark}\label{remark:eeee}
In Case 2, $M$ is isomorphic as an $\lalg$-module
to the $\lalg$-submodule of $\lalg$ generated by
$\crx{0}{}\crx{1}{} + i\crx{2}{}\crx{3}{}$,
$\crx{0}{}\crx{2}{} + i\crx{3}{}\crx{1}{}$,
$\crx{0}{}\crx{3}{} + i\crx{1}{}\crx{2}{}$
in the role of $g_1,g_2,g_3$ respectively,
which one can check is also an $A$-module.
The gradings match after a degree shift by $2$.
\end{remark}
The Hilbert series are as follows:
\begin{align*}
& \text{Case 1:} &    M : &\quad 1+4t+3t^2 &
      M_n : &\quad (1+4t+3t^2)^n\\
& \text{Case 2:} &    M : &\quad 3+4t+t^2 &
      M_n : &\quad (3+4t+t^2) (1+4t+3t^2)^{n-1}
\end{align*}
As modules, $M$ and $M_n$ are finitely generated by degree zero elements.
\begin{remark}[The multilinear perspective] \label{remark:prodcont}
This is a continuation of Remark \ref{remark:prod}.
Denoting by $\otimes$ the $\C$-tensor product as before,
there are canonical isomorphisms
\begin{alignat*}{5}
& \text{Case 1:}\qquad &
 M & \simeq E^0 &\qquad M_n & \simeq M^{\otimes n} \simeq (E^0)^{\otimes n}\\
& \text{Case 2:} &
 M & \simeq E^1 & M_n & \simeq M \otimes (Ag)^{\otimes (n-1)} \simeq E^1 \otimes (E^0)^{\otimes (n-1)}
\end{alignat*}
where $E^r = E^{r,0} \oplus E^{r,1} \oplus E^{r,2}$ with $r=0,1$ are the graded modules given,
in the notation of Section \ref{sec:defs}, by
$E^{0,0} = \Lambda^0$, $E^{0,1} = \Lambda^1$, $E^{0,2} = \Lambda^2_+$
and
$E^{1,0} = \Lambda^2_+$, $E^{1,1} = \Lambda^3$, $E^{1,2} = \Lambda^4$.
So elements of $\Hom_\C(M_n,M)$ may equivalently be viewed as multilinear maps
with $n$ slots,
and elements of $X^{k,\ell}$ as multilinear maps depending polynomially
on a parameter in $\VVV_n$. 
One can translate the definition of $C_\ell$ and its differential
to this perspective.
\end{remark}
\begin{remark}
This homology could also be computed using the
BGG correspondence between graded $\lalg_n$-modules
and linear $S \VVV_n^\ast$-complexes;
see especially \cite[Theorem 7.8]{eb}\footnote{%
Summary:
Let $U$ be a vector space, $\dim U < \infty$;
let $\Lambda = \Lambda U$ and $S = S U^\ast$;
let $\mx \subset \Lambda$ be the unique maximal ideal;
let $P$ be a graded $\Lambda$-module, $\dim P < \infty$;
let $P^\ast$ be the dual vector space which is canonically a $\Lambda$-module;
let $F \to P^\ast \to 0$ be a free resolution of $P^\ast$ as a $\Lambda$-module.
A canonical differential on $S \otimes_\C P$
is given by acting from the left with $\one \in \End(U) \simeq U^\ast \otimes U$.
Then there are canonical isomorphisms
\[
H(F/\mx F) \;\simeq\; \Tor^{\Lambda}(P^\ast,\Lambda/\mx) \;\simeq\; H(S \otimes_\C P)^\ast
\]
of bigraded vector spaces;
the bigradings are linear combinations of the homological and module degrees
detailed in \cite[Theorem 7.8]{eb}.
This is based on the two ways of computing $\Tor$:
the right $\simeq$ comes from using the minimal free resolution of $\Lambda/\mx$
in the proof of \cite[Theorem 7.8]{eb};
the left $\simeq$ is from the definition of $\Tor$
 using a free resolution of its left argument.
If $F$ is minimal then $H(F/\mx F) \simeq F/\mx F$.
So the graded Betti numbers of $P^\ast$,
which can be read off from any minimal free resolution, 
determine the dimensions of the individual parts of
the bigraded vector space $H(S \otimes_{\C} P)$.
No attempt was made to apply this.}.
To apply this here,
one would endow $\Hom_\C(M_n,M)$ with a $\lalg_n$-module structure\footnote{%
Define a $\lalg_n$-module structure on $\Hom_{\C}(M_n,M)$ by
$(v \cdot \elzero)(m) = f(v) \elzero(m) - (-1)^{\elzero} \elzero(vm)$.
Then define $D^{k,\ell}: X^{k,\ell} \to X^{k+1,\ell+1}, \elzero \mapsto \elone$
by $\elone_v(m) = f(v)\elzero_v(m) - (-1)^{\ell} \elzero_v(vm)$
or, using the module structure, $\elone_v = v \cdot \elzero_v$
The $D^{k,\ell}$ are equivalent to the $d^{k,\ell}$:
Define involutions
$i^{k,\ell} \in \End(X^{k,\ell})$
by $(i^{k,\ell} \elzero)_v(m) = \elzero_v(\sigma_{k+\ell}m)$,
where $\sigma_p \in \End(M_n)$ is given by $\sigma_p(m) = (-1)^{pm} m$,
then $i^{k+1,\ell+1} \circ d^{k,\ell} \circ i^{k,\ell} = D^{k,\ell}$.}.
\end{remark}


\begin{proof}
For $H^0(C_{\ell})=0$,
we must show that 
if $\elzero \in X^{0,\ell}$ and $d^{0,\ell}\elzero=0$ then $\elzero = 0$.
Since $v \mapsto \elzero_v$ is a constant map, write $\elzero_v = \elzero$.
Then $d\elzero=0$ is equivalent to
$\elzero(vm) = f(v) \elzero(m)$
for all $v \in \VVV_n$, $m \in M_n$.
So to prove $\elzero=0$ it suffices to prove that $\elzero$ annihilates the
module generators of $M_n$,
that is, to prove that $\elzero(g)$ respectively
$\elzero(g_1), \elzero(g_2), \elzero(g_3)$
are zero. This holds because these elements of $M$ have degree $\ell < 0$,
but $M$ is zero in negative degree.

For $H^1(C_{\ell}) = 0$,
given $\elone \in X^{1,\ell+1}$ with $d^{1,\ell+1} \elone = 0$,
we must show there exists\footnote{%
By $H^0(C_\ell)=0$ for $\ell < 0$,
hence for $\ell+1 < 0$:
If $\elzero$ exists then it is unique.}
a $\elzero \in X^{0,\ell}$ with $d^{0,\ell} \elzero = -\elone$;
the minus sign is for convenience.
Here $d\elone=0$ is equivalent to
\begin{equation}\label{eq:dv}
f(v) \elone_v(m) + \elone_v(vm) = 0
\end{equation}
for all $v \in \VVV_n$ and $m \in M_n$.
Since $v \mapsto \elone_v$ is a linear map, $d\elone=0$ is equivalent to
\begin{equation}\label{eq:dv2}
f(v) \elone_w(m) + f(w) \elone_v(m) + \elone_v(wm) + \elone_w(vm) = 0
\end{equation}
for all $v,w \in \VVV_n$ and $m \in M_n$.
In the following, $G$ stands for $g$ in Case 1 respectively
for either of $g_1$, $g_2$, $g_3$ in Case 2.
Note that $\elone_v(G) = 0$ for all $v$
 because it has degree $\ell + 1 < 0$.
We now define $\elzero$ in several steps.
First define a map $\elzero^\talg$ on free $\talg_n$-modules:
\begin{alignat*}{5}
& \text{Case 1:}\casedist
                 & \elzero^\talg &\in \Hom_\C^\ell(\talg_n g,M) \\
& \text{Case 2:} & \elzero^\talg &\in \Hom_\C^\ell(\talg_ng_1 \oplus \talg_ng_2 \oplus \talg_ng_3, M)
\end{alignat*}
where, recursively in the degree $N \geq 0$ of the input,
\begin{equation}\label{eq:recdefx}
\begin{aligned}
\elzero^\talg(v_N\cdots v_1G) & = \elone_{v_N}(v_{N-1} \cdots v_1 G)
+ f(v_N) \elzero^\talg(v_{N-1}\cdots v_1 G)\\
\elzero^\talg(G) & = 0
\end{aligned}
\end{equation}
for all $v_1,\ldots,v_N \in \VVV_n$ and $G = g$ or $G = g_1,g_2,g_3$.
Note that $\elzero^\talg$ is well-defined and has degree $\ell$ as claimed.
The recursion \eqref{eq:recdefx} means
(here the final term is actually zero):
\begin{multline}\label{eq:nonrec}
\elzero^\talg(v_N\cdots v_1G) \;=\;
\elone_{v_N}(v_{N-1}\cdots v_1 G)
\;+\; f(v_N) \elone_{v_{N-1}}(v_{N-2}\cdots v_1 G)
\;+\; \\
\;+\; \ldots
\;+\; f(v_N) \cdots f(v_3) \elone_{v_2}(v_1G)
\;+\; f(v_N) \cdots f(v_2) \elone_{v_1}(G)
\end{multline}
Let $I \subset \talg_n$ be the two-sided ideal generated by all $v^2$ with $v \in \VVV_n$.
Recall $\talg_n/I = \lalg_n$.
Claim: $\elzero^\talg(IG) = 0$. It suffices to check that
$\elzero^\talg(v_N\cdots v_1G) = 0$
when $v_{p+1} = v_p = v$ for $p$ with $0 \leq p \leq N-1$;
here $N \geq 2$.
Since $\elone_w(IG) = 0$
for all $w \in \VVV_n$ because $IG = 0 \subset M_n$,
and since $f(v)^2 = 0 \in A$, it suffices by \eqref{eq:nonrec} to check that
$\elone_v(v-) + f(v) \elone_v(-) = 0$,
which follows from \eqref{eq:dv}.
Therefore $\elzero^\talg(IG)=0$, so $\elzero^\talg$ descends to a map
\begin{alignat*}{5}
& \text{Case 1:}\casedist
                 & \elzero^{\lalg} &\in \Hom_\C^\ell(\lalg_n g,M) \\
& \text{Case 2:} & \elzero^{\lalg} &\in \Hom_\C^\ell(\lalg_ng_1 \oplus \lalg_ng_2 \oplus \lalg_ng_3, M)
\end{alignat*}
Denote by $J \subset \lalg_n$ the left ideal generated by all
\[
E^1_j = \crx{0}{j} \crx{1}{j} - i\crx{2}{j} \crx{3}{j}
\qquad
E^2_j = \crx{0}{j} \crx{2}{j} - i\crx{3}{j} \crx{1}{j}
\qquad
E^3_j = \crx{0}{j} \crx{3}{j} - i\crx{1}{j} \crx{2}{j}
\]
with $j=1\ldots n$.
Since $\lalg_n$ is graded commutative,
$J$ is a two-sided ideal, and $A_n = \lalg_n/J$.
Claim: $\elzero^\lalg(JG) = 0$.
To see this, it suffices to check that
$\elzero^\lalg(v_N\cdots v_3 E^1_j G) = 0$;
here $N \geq 2$. 
($E^2_j$, $E^3_j$ are analogous.)
Since
$\elone_w(JG)=0$ for all $w \in \VVV_n$
because $JG = 0 \subset M_n$,
and since $\elone_w(G)=0$, it suffices by \eqref{eq:nonrec} to check that\footnote{%
In each of these three claimed equations,
$G$ stands for the same element
on both sides.}
\begin{align}\label{eq:v01v23}
\elone_{\crx{0}{j}}(\crx{1}{j}G) & =  i\elone_{\crx{2}{j}}(\crx{3}{j}G) &
\elone_{\crx{0}{j}}(\crx{2}{j}G) & =  i\elone_{\crx{3}{j}}(\crx{1}{j}G) &
\elone_{\crx{0}{j}}(\crx{3}{j}G) & =  i\elone_{\crx{1}{j}}(\crx{2}{j}G)
\end{align}
If $\ell+2 < 0$ then all terms in \eqref{eq:v01v23}
are separately zero;
the proof of \eqref{eq:v01v23}
in the remaining case $\ell+2=0$ is given after this proof.
Therefore $\elzero^\lalg(JG)=0$, so $\elzero^\lalg$ descends to a map
\begin{alignat*}{5}
& \text{Case 1:}\casedist
                 & \elzero^A &\in \Hom_\C^\ell(A_n g,M) \\
& \text{Case 2:} & \elzero^A &\in \Hom_\C^\ell(A_ng_1 \oplus A_ng_2 \oplus A_ng_3, M)
\end{alignat*}
In Case 1 we have $M_n = A_n g$, so set $\elzero = \elzero^A$\footnote{%
This does not conclude Case 1, one must still prove $d\elzero = -\elone$.}.
In Case 2 recall the eight generators of the graded left $A_n$-submodule
$S_n \subset A_n g_1 \oplus A_n g_2 \oplus A_n g_3$. Denote them by
\[
P_1 = \crx{1}{1} g_1 - \crx{2}{1} g_2
\qquad \ldots \qquad
P_8 = \crx{0}{1} g_3 + i\crx{1}{1} g_2
\]
Claim: $\elzero^A(S_n) = 0$. 
It suffices to check that
$\elzero^A(v_N\cdots v_2 P_1) = 0$; here $N \geq 1$. 
($P_2,\ldots, P_8$ are analogous.)
Since
$\elone_w(S_n)=0$ for all $w \in \VVV_n$
because $S_n = 0 \subset M_n$,
it suffices by \eqref{eq:nonrec} to check
\smash{$\elone_{\crx{1}{1}}(g_1) - \elone_{\crx{2}{1}}(g_2) = 0$},
which is true because the $\elone_w$ annihilate all generators
for degree reasons. Hence $\elzero^A(S_n)=0$ and $\elzero^A$ descends to a map
\begin{alignat*}{5}
& \text{Case 2:}\casedist & \elzero &\in \Hom_\C^\ell(M_n, M) \simeq X^{0,\ell}
\end{alignat*}
We are ready to check that $d^{0,\ell}\elzero=-\elone$ which is equivalent to
$\elzero(vm) = \elone_v(m) + f(v)\elzero(m)$
for all $v \in \VVV_n$ and $m \in M_n$.
This follows from the recursion \eqref{eq:recdefx}
which is satisfied by $\elzero^\talg$ by definition,
and which is therefore also satisfied by $\elzero$.
\qed\end{proof}
\begin{proof}[of equation \eqref{eq:v01v23}]
The case of interest is $\ell + 2 = 0$,
but only assume $\ell + 1 < 0$ for the time being.
Two observations:
\begin{itemize}
\item
$\elone_v(wG) = - \elone_w(vG)$
for all $v,w \in \VVV_n$, by \eqref{eq:dv2}.
\item
$\elone_v(\crx{0}{j} \crx{1}{j}G) = i\elone_v(\crx{2}{j}\crx{3}{j}G)$
for all $v \in \VVV_n$,
because $\crx{0}{j} \crx{1}{j} = i\crx{2}{j} \crx{3}{j}$ in $A_n$.\\
Analogous identities are obtained by cyclically permuting the indices $1,2,3$.
\end{itemize}
For fixed $j$ and $G$ abbreviate the following elements of $M$
(recall \eqref{eq:deff}):
\[
r_{\mu\nu\rho} = \elone_{\crx{\mu}{j}}(\crx{\nu}{j}\crx{\rho}{j}G)
\qquad
s_{\mu\nu\rho} = \crx{\mu}{} \elone_{\crx{\nu}{j}}(\crx{\rho}{j}G)
\]
where $\mu,\nu,\rho = 0,1,2,3$.
By the two observations above and by \eqref{eq:dv2}:
\begin{align*}
   r_{\mu\nu\rho} & = -r_{\mu\rho\nu} &
        r_{\mu 0 1} & = ir_{\mu 2 3}\\
   s_{\mu\nu\rho} + s_{\nu\mu\rho} + r_{\mu\nu\rho} + r_{\nu\mu\rho} & = 0 &
        r_{\mu 0 2} & = ir_{\mu 3 1}\\
   s_{\mu\nu\rho} & = -s_{\mu\rho\nu} &
        r_{\mu 0 3} & = ir_{\mu 1 2}
\end{align*}
Eliminating all $r_{\mu\nu\rho}$
from this finite linear system, one finds 
$(s_{001} - is_{023}) + i(s_{302} - is_{331}) = 0$
and $(s_{301} - is_{323}) + i(s_{002} - is_{031}) = 0$
and $s_{101} - is_{123} = s_{202}-is_{231}$ 
and analogous identities obtained by cyclically permuting $1,2,3$.
It is convenient to write this in matrix form, $WZ = 0$.
Here $W$ is an $8 \times 3$ matrix with entries in $A$ of degree one,
viewed as a degree one map $M^3 \to M^8$, and $Z \in M^3$ is a vector:
\[
W = {\footnotesize \begin{pmatrix}
\crx{0}{} & i\crx{3}{} & 0\\
\crx{3}{} & i\crx{0}{} & 0\\
0 & \crx{0}{} & i\crx{1}{} \\
0 & \crx{1}{} & i\crx{0}{} \\
i\crx{2}{} & 0 & \crx{0}{} \\
i\crx{0}{} & 0 & \crx{2}{} \\
\crx{1}{} & -\crx{2}{} & 0\\
0 & \crx{2}{} & -\crx{3}{} \\
\end{pmatrix}}
\qquad 
Z = \begin{pmatrix}
\elone_{\crx{0}{j}}(\crx{1}{j}G) - i\elone_{\crx{2}{j}}(\crx{3}{j}G)\\
\elone_{\crx{0}{j}}(\crx{2}{j}G) - i\elone_{\crx{3}{j}}(\crx{1}{j}G)\\
\elone_{\crx{0}{j}}(\crx{3}{j}G) - i\elone_{\crx{1}{j}}(\crx{2}{j}G)
\end{pmatrix} 
\]
If $\ell+2<0$ then we have $Z=0$ for degree reasons.
If $\ell+2=0$ then $Z \in (\C g)^3$ in Case 1;
$Z \in (\C g_1 \oplus \C g_2 \oplus \C g_3)^3$ in Case 2.
But one checks
%
that $W$ is injective as a map\footnote{%
The dimensions are as follows.
Since $W$ has degree one,
let us only take into account the degree one subspace of $M^8$
which is $\simeq (\C^4)^8 \simeq \C^{32}$ in both cases.
With this understanding, the injectivity claim is for $\C$-linear maps
of dimensions
$\C^3 \to \C^{32}$ respectively $\C^9 \to \C^{32}$.}
\begin{alignat*}{5}
& \text{Case 1:}\casedist
                 & (\C g)^3 & \to M^8 \\
& \text{Case 2:} & (\C g_1 \oplus \C g_2 \oplus \C g_3)^3 & \to M^8
\end{alignat*}
Hence $WZ=0$ implies $Z=0$ which is \eqref{eq:v01v23}.
\qed\end{proof}

\section{$\bv{\Box}$-algebra for Yang-Mills}\label{sec:combineNEW}

\subsection{Proof of Theorem \ref{theorem:mainh}} \label{sec:proofmain}

Here we construct the $\bv{\Box}$-algebra
claimed in Theorem \ref{theorem:mainh}.
(Then, Theorem \ref{theorem:2ax} follows by writing out two axioms, namely,
in the notation of Lemma \ref{lemma:specialspecial}:
the 3-ary identity $b_3=0$ to get \eqref{eq:exact};
the 4-ary identity $c_4=0$ to get \eqref{eq:identity4}.)
The conventions in
Remarks \ref{remark:dualspaces} and \ref{remark:completion} are in force.
This proof is organized around Lemma \ref{lemma:specialspecial}.
Use the following data:
\[
\begin{tabular}{r|l}
data in Section \ref{sec:ginfMAIN} & Yang-Mills in this section\\
\hline
$V$ & $\ax[2]$, see Section \ref{sec:defs}\\
$\Hopf$ & $\C[\p_0,\ldots,\p_3]$\\

$\Qx$ & $\Box = \eta^{\mu\nu}\p_\mu\p_\nu$
\end{tabular}
\]
So $\ax = V[-2]$\footnote{%
For a graded vector space $V$,
one denotes by $V[k]^i = V^{i+k}$ the shifted space.},
that is $V^i = \ax^{i+2}$, therefore $V = V^{-2} \oplus \ldots \oplus V^1$.
The coproduct of $\Hopf$ is given by
$\p_\mu \mapsto \p_\mu \otimes \one + \one \otimes \p_\mu$.
The $\Hopf$-module structure on $\ax$ is given by differentiation,
the one on $\ax^{\otimes n}$ is induced by the coproduct.
Note that $\Qx = \Box$ satisfies \eqref{eq:7term}.
\step
\emph{The first three operations.}
Recall that, by assumption of Theorem \ref{theorem:mainh}, an $h$
satisfying the requirements in Section \ref{sec:intro} is given,
in particular \eqref{eq:hxs} holds.
\begin{itemize}
\item The differential in \eqref{eq:cpx} maps $d: \ax^i \to \ax^{i+1}$,
equivalently $V^i \to V^{i+1}$, so its dual maps
$(V^\ast)^i \to (V^\ast)^{i+1}$.
By \eqref{eq:dparx} it defines a derivation also denoted
$d \in \Der^1(GV^\ast)^{\Hopf}$.
Analogously, $h: \ax^i \to \ax^{i-1}$
in \eqref{eq:hxs} gives $h \in \Der^{-1}(GV^\ast)^{\Hopf}$.
Set $\mu_1 = d+h$.
\item
The gca product maps $\ax^i \otimes \ax^j \to \ax^{i+j}$,
or $V^i \otimes V^j \to V^{i+j+2}$,
so formally its dual maps
$(V^{\ast})^i \to \sum_{i_1+i_2 = i+2} (V^\ast)^{i_1} \otimes (V^\ast)^{i_2}$.
It is graded commutative,
so defines a derivation $\theta_2 \in \Der^2(GV^\ast)^{\Hopf}$
by \eqref{eq:dparx}
consistent with \eqref{eq:thk}.
Set $\nu_2 = d_{\tt}\theta_2 \in \Der^1(GV^\ast)^{\Hopf}$.
\end{itemize}
Some assumptions in Lemma \ref{lemma:specialspecial} follow:
$[d,d]=[h,h] = 0$ and $[d,h] = \curv$ using \eqref{eq:hxs};
and $b_2=c_3=0$ because
$[d,\nu_2] = [\nu_2,\nu_2]=0$ by the dgca Leibniz rule and associativity.

\begin{definition}[Local operators] \label{def:locop}
Denote by $\mathcal{C}_n$ the space of all maps $\ax^{\otimes n} \to \ax$
that are $C^\infty(\R^4)$-linear in all $n$ arguments and
invariant under translations on $\R^4$.
Denote
\begin{equation}\label{eq:dndef}
\mathcal{D}_n
  \;=\;
\big\{\;
\textstyle\sum_{j=1}^n c_j^\mu (\one^{\otimes(j-1)} \otimes \p_\mu \otimes \one^{\otimes(n-j)}) : \ax^{\otimes n} \to \ax
\;\mid\; c_j^\mu \in \mathcal{C}_n
\;\big\}
\end{equation}
a class of translation invariant, homogeneous first order partial
differential operators.
Note that $\mathcal{C}_n$ and $\mathcal{D}_n$ are finite-dimensional.
All elements of $\mathcal{C}_n$ and $\mathcal{D}_n$
are $\Hopf$-equivariant.
\end{definition}
\emph{Row decomposition.}
Denote by $\ax^i = \oplus_{r=0,1} \ax^{r,i}$
the decomposition of the complex \eqref{eq:cpx}
into the upper row, $r=0$,
and the lower row, $r=1$.
For instance, $\ax^{1,2} = \Omega^3$.
This leads to decompositions $d = d^0 + d^1$ and $h = h^0 + h^1$ where
\begin{align*}
d^s \;:\; \ax^{r,i}
& \;\to\; \ax^{r+s-1,i+1}\\
h^s \;:\; \ax^{r,i}
& \;\to\; \ax^{r+s-1,i-1}
\end{align*}
for $s=0,1$. Here
$d^0,\;h^0 \in\mathcal{C}_1$
and $d^1,\;h^1 \in \mathcal{D}_1$.
In particular, $d^1$ maps the upper (lower) row of $\ax$ to the upper (lower) row.
In the notation of Section \ref{sec:defs},
\begin{equation}\label{eq:dshk3hr3}
d^1 = \p_\mu \otimes e^\mu : \ax^{r,i} \to \ax^{r,i+1}
\end{equation}
here viewed as an element $d^1 \in \Der^1(GV^\ast)^{\Hopf}$
with $[d^1,d^1]=0$.

\step
\emph{The notation $\theta_n$.}
Below we inductively construct operations $\theta_n$ with $n \geq 2$.
By abuse of notation, the operation $\theta_n$ is simultaneously
(the correspondence is via \eqref{eq:dparx}):
\begin{itemize}
\item A graded symmetric and $\Hopf$-equivariant map
$\theta_n : \ax^{\otimes n} \to \ax$ of degree $4-2n$.\\
Equivalently, a map $V^{\otimes n} \to V$
or formally $V^\ast \to (V^\ast)^{\otimes n}$ of degree $2$.
\item A derivation $\theta_n \in \Der^2(GV^\ast)^{\Hopf}$,
as in Lemma \ref{lemma:specialspecial}.
\end{itemize}
Graded symmetry means
\begin{equation}\label{eq:gras}
\theta_n(\ldots,x,y,\ldots) = (-1)^{xy} \theta_n(\ldots,y,x,\ldots)
\end{equation}
The inductive construction will be such that
\begin{equation}\label{eq:ztu2}
\theta_n \;:\; \ax^{r_1,i_1} \otimes \cdots \otimes \ax^{r_n,i_n}
 \;\to\; \ax^{r_1+\ldots+r_n,\;i_1+\ldots+i_n+4-2n}
\qquad
\text{and}
\qquad \theta_n \in \mathcal{C}_n 
\end{equation}
The operation $\theta_2$ defined above is the gca product in $\ax$,
and is consistent with \eqref{eq:ztu2}.
The operation $\theta_3$ constructed below is the one that Theorem \ref{theorem:2ax} refers to.
\step
\emph{Outline of the induction.}
Below we inductively construct $\theta_3,\theta_4,\ldots$
so that (in the notation of Lemma \ref{lemma:specialspecial})
$b_3,c_4,b_4,c_5,b_5,\ldots$ vanish
and so that \eqref{eq:gras}, \eqref{eq:ztu2} hold.
More precisely,
using Part B of Lemma \ref{lemma:specialspecial}, one proceeds as follows,
recalling $b_n = -[d,\theta_n] + q_n$:
 Use $[d,q_3] = 0$
to construct $\theta_3$ so that $[d,\theta_3] = q_3$;
use $[d,c_4] = 0$ to show that $c_4 = 0$;
use $[d,q_4] = 0$
to construct $\theta_4$ so that $[d,\theta_4] = q_4$;
use $[d,c_5] = 0$ to show that $c_5 = 0$; etc.
At each step, the construction of $\theta_n$
uses the homology vanishing theorem, Theorem \ref{theorem:vanish}.
\step
\emph{Properties of $q_n$ and $c_n$ known inductively.}
The $q_n$ and $c_n$ are defined once $\theta_2,\ldots,\theta_{n-1}$ are defined.
Similar to the notation $\theta_n$,
here too we abuse notation, using \eqref{eq:dparx} and \eqref{eq:bkck}:
$b_n$ and $q_n$ are formally also graded symmetric maps $V^\ast \to (V^\ast)^{\otimes n}$
of degree $3$;
and $c_n$ is also a map $V^\ast \to (V^\ast)^{\otimes n}$
of degree $4$ that 
has some symmetry in the first two respectively the remaining $n-2$
factors.
Suppose we are at a stage of the induction where
$\theta_2,\ldots,\theta_{n-1}$
have been constructed and satisfy \eqref{eq:ztu2},
then one can check\footnote{%
Using \eqref{eqzqkj},
note that $q_n$ is a
sum of composition of three operations,
$h$ and $\theta_m$ and $\theta_{n+1-m}$ with $1<m<n$;
and $c_n$ is a sum of composition of two operations, $\theta_m$ and $\theta_{n+1-m}$
with $1<m<n$.} that the following is true:
\begin{subequations}\label{eq:z2w}
\begin{alignat}{5}
\label{eq:cn4}
c_n &:&\qquad \ax^{r_1,i_1} \otimes \cdots \otimes \ax^{r_n,i_n}
& \;\to\; \ax^{r_1+\ldots+r_n,\;i_1+\ldots+i_n+6-2n}\\
\label{eq:qn4}
q_n^s &:& \ax^{r_1,i_1} \otimes \cdots \otimes \ax^{r_n,i_n}
& \;\to\; \ax^{r_1+\ldots+r_n+s-1,\;i_1+\ldots+i_n+5-2n}
\end{alignat}
\end{subequations}
where $q_n = q_n^0 + q_n^1$ with
\begin{equation}\label{eq:ujn}
c_n,
\;q_n^0 \in \mathcal{C}_n
\qquad\qquad
q_n^1 \in \mathcal{D}_n
\end{equation}
\step
\emph{Induction step to prove $c_n=0$, for $n \geq 4$.}
This comes after $\theta_2,\ldots,\theta_{n-1}$
have been constructed and $b_2,c_3,b_3,\ldots,c_{n-1},b_{n-1}$ are known to vanish.
(For $b_n=0$, see below.)
\begin{itemize}
\item
By Part B of Lemma \ref{lemma:specialspecial} we have $[d,c_n]=0$. This has the form
\[
d c_n(x_1,\ldots,x_n) = \textstyle\sum_{i=1}^n \pm c_n(x_1,\ldots,x_{i-1},dx_i,x_{i+1},\ldots,x_n)
\]
for all $x_1,\ldots,x_n \in \ax$.
The signs, which may depend on degrees, are irrelevant.
By \eqref{eq:ujn},
$[d,c_n]$ is a first order partial differential operator.
The homogeneous first order part must separately vanish, $[d^1,c_n]=0$.
Using
\eqref{eq:dshk3hr3} and $c_n \in \mathcal{C}_n$ this implies
\begin{equation}\label{eq:cid4}
e^\mu c_n(x_1,\ldots,x_n)
\;=\;
\pm c_n(x_1,\ldots,x_{i-1},e^\mu x_i,x_{i+1},\ldots,x_n)
\end{equation}
separately for all $\mu=0,\ldots,3$ and $i=1\ldots n$.
\item
Denote by $c_n^{r_1\ldots r_n}$ the map \eqref{eq:cn4},
meaning the restriction of $c_n$ where the $i$-th argument
is in row $r_i$ of $\ax$.
The identities \eqref{eq:cid4} hold separately for each $c_n^{r_1\ldots r_n}$
because $e^\mu$ maps the upper (lower) row of $\ax$ to the upper (lower) row of $\ax$.
Hence $c_n^{r_1\ldots r_n}$ vanishes iff its restriction to
the generators $\oplus_{r=0,1} \ax^{r,r}$, meaning to the leftmost term in each row
in \eqref{eq:cpx}, vanishes.
By \eqref{eq:cn4}, this restriction maps
\[
\ax^{r_1,r_1} \otimes \cdots \otimes \ax^{r_n,r_n}
\;\to\; \ax^{r_1+\ldots+r_n,\;r_1+\ldots+r_n+6-2n}
\]
Since $6-2n < 0$, the space on the right is zero,
so $c_n^{r_1\ldots r_n}=0$, so $c_n=0$.
\end{itemize}
\step
\indent\indent
\emph{Induction step to construct $\theta_n$
                    and to prove $b_n = 0$, for $n \geq 3$.}
This is after $\theta_2,\ldots,\theta_{n-1}$
have been constructed and $b_2,c_3,b_3,\ldots,c_{n-1},b_{n-1},c_n$ are known to vanish.
\begin{itemize}
\item
By Part B of Lemma \ref{lemma:specialspecial}
we have $[d,q_n]=0$. This means that
\[
d q_n(x_1,\ldots,x_n) + \textstyle\sum_{i=1}^n (-1)^{x_1+\ldots+x_{i-1}}
q_n(x_1,\ldots,x_{i-1},dx_i,x_{i+1},\ldots,x_n)\;=\;0
\]
Recall that $q_n : \ax^{\otimes n} \to \ax$
is a graded symmetric map of degree $5-2n$.
By \eqref{eq:ujn},
$[d,q_n]$ is a second order partial differential operator.
The homogeneous second order part must separately vanish, 
$[d^1,q_n^1] = 0$.
Since $q_n^1 \in \mathcal{D}_n$ one has
\[
  q_n^1 = \textstyle\sum_{j=1}^n u^\mu_j (\one^{\otimes (j-1)} \otimes \p_\mu
\otimes \one^{\otimes (n-j)})
\]
for $u^\mu_j \in \mathcal{C}_n$.
The $u^\mu_j: \ax^{\otimes n} \to \ax$ have degree $5-2n$
but may not be graded symmetric.
Now $[d^1,q_n^1]=0$ implies,
separately for all $\mu,\nu=0,\ldots,3$ and $i,j=1,\ldots,n$:
\begin{equation}\label{eq:eueu}
\begin{aligned}
&e^\mu u^\nu_j(x_1,\ldots,x_n)
+ e^\nu u^\mu_i(x_1,\ldots,x_n)\\
& + (-1)^{x_1+\ldots+x_{i-1}}
  u^\nu_j(x_1,\ldots,x_{i-1},e^\mu x_i, x_{i+1},\ldots,x_n)\\
& + (-1)^{x_1+\ldots+x_{j-1}}
  u^\mu_i(x_1,\ldots,x_{j-1},e^\nu x_j, x_{j+1},\ldots,x_n)
= 0
\end{aligned}
\end{equation}
\item
Denote by $q_n^{1;r_1\ldots r_n}$ the map \eqref{eq:qn4},
meaning the restriction of $q_n^1$ where the $i$-th argument
is in row $r_i$ of $\ax$;
let $u_j^{\mu;r_1\ldots r_n}$ be the corresponding piece of $u_j^\mu$;
let $u^{r_1\ldots r_n}$ be the collection of all $4n$ maps $(u_j^{\mu;r_1\ldots r_n})$.
Explicitly,
\[
q_n^{1;r_1\ldots r_n},
\;u_j^{\mu;r_1\ldots r_n} \;:\; \ax^{r_1,i_1} \otimes \cdots \otimes \ax^{r_n,i_n}
 \;\to\; \ax^{r_1+\ldots+r_n,\;i_1+\ldots+i_n+5-2n}
\]
The \eqref{eq:eueu} hold separately for each collection $u^{r_1\ldots r_n}$.
Distinguish Case 1 where $r_i=0$ for all $i$;
Case 2 where $r_j=1$ for one $j$ and $r_i=0$ for $i\neq j$;
otherwise $u^{r_1\ldots r_n}=0$ for degree reasons.
The plan is to apply the respective Cases 1 and 2 of Theorem \ref{theorem:vanish}.
To do this, use $u_j^{\mu;r_1\ldots r_n} \in \mathcal{C}_n$
to encode the collection $u^{r_1\ldots r_n}$ in a linear map
\begin{align}
\label{eq:ztz}
\elone^{r_1\ldots r_n}\;\;:\;\;\VVV_n = \textstyle\oplus_{\mu,i} \C \crx{\mu}{j}
& \to \Hom_\C^{5-2n}(
E^{r_1} \otimes \cdots \otimes E^{r_n},\;E^{r_1+\ldots+r_n})\\
\notag
\crx{\mu}{j} & \mapsto u_j^{\mu;r_1\ldots r_n}
\end{align}
using $\ax^{r,i} = C^\infty(\R^4) \otimes E^{r,i+r}$,
where the vector spaces $E^r$
are defined in Remark \ref{remark:prodcont}.
So
$\elone^{r_1\ldots r_n}
\in X^{1,\ell+1} = X^{1,5-2n}$ (see \eqref{eq:xkl})
via the isomorphisms in
Remarks \ref{remark:prod} and \ref{remark:prodcont}.
The identities \eqref{eq:eueu}
translate to \eqref{eq:dv2} so that $d^{1,\ell+1} \elone^{r_1\ldots r_n} = 0$
with $\ell+1 = 5-2n < 0$.
\item
Hence $H^1(C_\ell) = H^1(C_{4-2n}) = 0$ in Theorem \ref{theorem:vanish}
implies that $\elone^{r_1\ldots r_n} = d^{0,\ell} \elzero^{r_1\ldots r_n}$
for some $\elzero^{r_1\ldots r_n} \in X^{0,\ell}$.
By $H^0(C_\ell) = 0$ in Theorem \ref{theorem:vanish}
and $\ell = 4-2n < 0$, these $\elzero^{r_1\ldots r_n}$ are unique.
Translating the $\elzero^{r_1\ldots r_n}$ using Remarks
\ref{remark:prod} and \ref{remark:prodcont},
there are unique
\[
\theta_n^{r_1\ldots r_n}
\;:\;
\ax^{r_1,i_1} \otimes \cdots \otimes \ax^{r_n,i_n}
\to
\ax^{r_1+\ldots+r_n,\;i_1+\ldots+i_n+4-2n}
\]
in $\mathcal{C}_n$ with $[d^1,\theta_n^{r_1\ldots r_n}] = q_n^{1;r_1\ldots r_n}$.
Denoting by $\theta_n:\ax^{\otimes n} \to \ax$ the direct sum of all these components,
by construction we have $[d^1,\theta_n] = q_n^1$ which explicitly means
\begin{multline*}
q_n^1(x_1,\ldots,x_n)\\
\;=\;
d^1 \theta_n(x_1,\ldots,x_n)
- \textstyle\sum_{i=1}^n (-1)^{x_1+\ldots+x_{i-1}}
\theta_n(x_1,\ldots,x_{i-1},d^1 x_i, x_{i+1},\ldots,x_n)
\end{multline*}
Since $q_n^1$ is graded symmetric, and since $\theta_n$ is unique as stated,
it follows that $\theta_n$ is also graded symmetric, as required.
It is also $\Hopf$-equivariant.
\item Recall that $b_n = -[d,\theta_n] + q_n$
is a graded symmetric map $\ax^{\otimes n} \to \ax$ of degree $5-2n$.
We must show $b_n=0$.
By construction of $\theta_n$,
we have $b_n \in \mathcal{C}_n$, so its pieces are
\begin{equation}\label{eq:btlt}
b_n^{r_1\ldots r_n}
\;:\;
\ax^{r_1,i_1} \otimes \cdots \otimes \ax^{r_n,i_n}
\to
\ax^{r_1+\ldots+r_n-1,\;i_1+\ldots+i_n-1+(6-2n)}
\end{equation}
Part B of Lemma \ref{lemma:specialspecial} implies $[d,b_n]=0$,
and $b_n \in \mathcal{C}_n$
implies $[d^0,b_n]=[d^1,b_n]=0$.
But $[d^1,b_n] = 0$ splits into independent
identities $[d^1,b_n^{r_1\ldots r_n}]=0$.
Hence $b_n^{r_1\ldots r_n}$ vanishes iff its restriction
to $\oplus_{r=0,1} \ax^{r,r}$ vanishes\footnote{%
By an argument like in the induction step for $c_n$.}.
If $n>3$ then $6-2n < 0$, so this vanishes for degree reasons by \eqref{eq:btlt},
so $b_n=0$. Hereafter, $n=3$.
We show that $[d^1,b_3^{r_1r_2r_3}]=0$ implies $b_3^{r_1r_2r_3}=0$
separately for every triple $r_1r_2r_3$.
Recall
\[
b_3^{r_1r_2r_3}
\;:\;
\ax^{r_1,i_1} \otimes \ax^{r_2,i_2} \otimes \ax^{r_3,i_3}
\to
\ax^{r_1+r_2+r_3-1,\;i_1+i_2+i_3-1}
\]
This can only be nonzero if exactly one or two of the $r_i$ are equal to $1$.
By graded symmetry, this leaves $u = b_3^{100}$ and $v = b_3^{110}$.
By $[d^1,b_3]=0$ they are determined by their restriction to $\oplus_{r=0,1}\ax^{r,r}$.
So $u,v \in \mathcal{C}_3$ are determined by two linear maps,
\begin{alignat*}{5}
U\;&:\quad & \Lambda^2_+ \otimes S^2 \Lambda^0 & \to \Lambda^0
\qquad&& \text{for $u$ restricted to generators}\\
V\;&:\quad & \wedge^2 \Lambda^2_+ \otimes \Lambda^0 & \to \Lambda^2_+
&& \text{for $v$ restricted to generators}
\end{alignat*}
Below we use the following bases: $1 \in \C \simeq \Lambda^0$;
and $g_1,g_2,g_3 \in \Lambda^2_+$
denote respectively the three elements in \eqref{eq:lambda2plus}. We also use
$u(e^\mu x,y,z) = -e^\mu u(x,y,z)$
and
$v(e^\mu x,y,z) = -e^\mu v(x,y,z)$
which follow from $[d^1,b_3]=0$. Consider now the two cases.
\begin{itemize}
\item 
Set $U_i = u(g_i,1,1) \in \C$.
Note $e^1 g_1 = e^2 g_2 \in \Lambda^3$, so
$u(e^1g_1,1,1) = u(e^2g_2,1,1)$,
so $e^1u(g_1,1,1) = e^2u(g_2,1,1)$,
so $e^1 U_1 = e^2 U_2 \in \Lambda^1$.
But $e^1,e^2 \in \Lambda^1$ are linearly independent, so $U_1=U_2=0$,
analogously $U_3=0$, so $U=0$.
\item 
Set $V(g_j,g_k,1) = \sum_{i=1}^3 V_{jk}^i g_i$,\,
$V_{jk}^i = -V_{kj}^i \in \C$.
Note that $g_i g_j = \delta_{ij} g_1g_1 \in \Lambda^4$.
So if $i \neq j$ then $v(g_ig_j,g_k,1) = 0$,
so $g_i v(g_j,g_k,1) = 0$, so $V^i_{jk}=0$, all when $i \neq j$.
This and antisymmetry in the two lower indices implies $V = 0$.
\end{itemize}
\end{itemize}
This concludes the proof of Theorem \ref{theorem:mainh}.


\subsection{Strict structure via a cobar construction}\label{sec:rect}

This section is a continuation of Section \ref{subsec:mod},
and it is not specific to Yang-Mills.
It contains an explicit cobar construction
that associates to a $\bv{\Qx}$-algebra a
quasi-isomorphic strict $\bv{\Qx}$-algebra,
called a $\text{dgBVa}^{\Qx}$ below.
The main task, relative to standard cobar constructions,
is to accommodate the $\Qx$-dependent terms in \eqref{eq:defheq}.
\step
The standard cobar construction \cite{dtt,lv}
associates to a dg coalgebra $C$ without counit\footnote{%
Or a conilpotent coaugmented dg coalgebra \cite{lv}.
Then the cokernel of the coaugmentation yields the version without counit.
In any event, the cobar construction involves removing the counit.
}
a dg algebra $A$ without unit\footnote{%
Or an augmented dg algebra \cite{lv}.
Then the kernel of the augmentation
yields the version without unit.}.
Here $A$ is an appropriate free algebra generated by $C$.
Schematically,
if $u_C$ and $b_C$ are the co-unary and co-binary co-operations on $C$
respectively, $u_A$ and $b_A$ the unary and binary operations on $A$, then:
\begin{itemize}
\item $u_A$ is given by an explicit formula in terms of $u_C$ and $b_C$.
\item $b_A$ is the native algebra structure on $A$.
\end{itemize}
So all operations on $C$ are packed into the unary one on $A$.
The cobar and bar constructions form an adjunction;
they work for various (co)algebraic structures;
and they can be used to rectify $\infty$-algebraic structures
to $\infty$-quasi-isomorphic strict structures \cite{dtt,gtv,lv}.
The result below is not a full rectification result,
because we only establish the quasi-isomorphism
for the lowest operations,
which suffices for our application\footnote{%
It is possible that the extension
to an $\infty$-quasi-isomorphism is immediate,
once things are set up,
also because there is an explicit candidate for
the $\infty$-quasi-isomorphism, cf.~\cite[Proposition 3]{dtt}.}.
\step
\emph{Notation for the cobar construction.} 
Let $V$ be a graded vector space and $\Hopf$-module, and an element $\Qx$ is fixed,
just as in Section \ref{subsec:mod}. Define:
\begin{itemize}
\item
$C$ is the cofree Gerstenhaber coalgebra without counit cogenerated by $V$.
It is spanned by co-Gerstenhaber words (cf.~\cite{w}) in $V$.
It is formally the linear dual of $GV^\ast$.
\item
$A$ is the free Gerstenhaber algebra without unit
generated by $C[-2]$\footnote{%
For a graded vector space $C$,
one denotes by $C[k]^i = C^{i+k}$ the shifted space.}, not completed.
It is spanned by Gerstenhaber words in $C$.
So it is spanned by words-of-words in $V$.
\end{itemize}
Both are $\Hopf$-modules and, canonically, $V[-2] \hookrightarrow C[-2] \hookrightarrow A$.
Note that $C$ is precisely the object that we avoided working with
in Section \ref{sec:ginfMAIN}, see Remark \ref{remark:dualspaces},
but here we must make the switch.
The gLa of $\Hopf$-equivariant Gerstenhaber coderivations is
denoted $\Coder(C)^{\Hopf}$,
its elements are formal duals of the elements of $\Der(GV^\ast)^{\Hopf}$.

\begin{lemma}[Cobar construction]\label{lemma:rcc}
Suppose $\delta \in \Coder(C)^{\Hopf}$ is a $\bv{\Qx}$-algebra structure on $V$,
see \eqref{eq:defheq}, with the following additional assumptions\footnote{%
These assumptions hold for the algebra in Theorem \ref{theorem:mainh}.
}:
\begin{itemize}
\item $\delta = \delta^1 + \delta^{-1}$ where $\delta^{\pm 1} \in \Coder^{\pm 1}(C)^{\Hopf}$,
so $\delta$ only has degree $1$ and $-1$ pieces\footnote{%
The notation does not distinguish
$\delta^{\pm 1} \in \Coder^{\pm 1}(C)^{\Hopf}$
and its formal dual
$\delta^{\pm 1} \in \Der^{\pm 1}(GV^\ast)^{\Hopf}$.}.
\item The only nonzero component of $\delta^{-1}$
is its $V \to V$ component\footnote{%
The $V \to V$ component is the formal
dual of the $V^\ast \to V^\ast$ component.}.
\end{itemize}
Then a cobar construction\footnote{%
The cobar construction is applied to $(C,\delta)$
here interpreted as a $\text{co-dgBVa}^{\Qx}$. See the proof.}
yields $\Hopf$-equivariant $d_A \in \End^1(A)$,
$h_A \in \End^{-1}(A)$
such that:
\begin{itemize}
\item 
Strict structure:
$A$ is a strict $\bv{\Qx}$-algebra,
or $\textnormal{dgBVa}^{\Qx}$, meaning:
\begin{itemize}
\item $A$ is a Gerstenhaber algebra, see the five axioms \eqref{eq:gax}.
\item
$d_A$ is a derivation for the gca product
and $h_A$ is a BV operator, so for all $a,b \in A$:
\begin{equation}\label{eq:unf}
\begin{aligned}
d_A(ab) & \;=\; d_A(a)b + (-1)^a a d_A(b) \\
[a,b] & \;=\; (-1)^a h_A(ab) - (-1)^a h_A(a)b - a h_A(b)
\end{aligned}
\end{equation}
\item $d_A^2 = 0$ and $d_Ah_A + h_Ad_A = \Qx \one$ and $h_A^2 = 0$.
\end{itemize}
So $(A,d_A)$ is a dgca and $(A,h_A)$ is a BV algebra,
in particular $h_A$ is second order.
\item
Quasi-isomorphism:
Denote by
$\imap : V[-2] \hookrightarrow A$
the canonical degree zero map,
and abbreviate the differentials\footnote{%
That $d$ and $h$ are differentials follows from
\eqref{eq:d2r} below
(which follows from \eqref{eq:defheq2}) using $\tt|_V=\qq|_V=0$.}
 $d = \delta^1|_V \in \End^1(V)$
and $h = \delta^{-1}|_V \in \End^{-1}(V)$.
There exist $\Hopf$-equivariant maps $\pmap$ of degree zero\footnote{%
This $\pmap$ is far from being
the canonical map $A \to V[-2]$.
},
$\zmap$ of degree $-1$, as in 
\begin{equation}\label{eq:defretract}
\vcenter{\hbox{%
\begin{tikzpicture}
  \draw[anchor=east] (0,0) node {$V[-2]$};
  \draw[anchor=west] (2,0) node {$A$};
  \draw[->] (0.1,0.2) to[bend left=15] node[midway,anchor=south] {$\imap$} (1.9,0.2);
  \draw[->] (1.9,-0.2) to[bend left=15] node[midway,anchor=north] {$\pmap$} (0.1,-0.2);
  \draw[->] (2.7,0.2) to[out=40,in=-40,looseness=8] node[midway,anchor=west] {$\zmap$} (2.7,-0.2);
\end{tikzpicture}}}
\end{equation}
such that
    $d_A \imap = \imap d$ and
    $h_A \imap = \imap h$
and $\pmap d_A = d \pmap$ and
    $\pmap h_A = h \pmap$
and $\pmap\imap = \one$
and $\imap\pmap = \one - d_A\zmap - \zmap d_A$ and $\zmap\imap=\pmap\zmap=\zmap^2=0$.
That is:
\begin{itemize}
\item $\imap$ and $\pmap$ are chain maps using the differentials
$d$ and $d_A$ respectively,
and they are quasi-isomorphisms
as witnessed by the contraction
\eqref{eq:defretract}.
\item $\imap$ and $\pmap$ are chain maps
using the differentials $h$ and $h_A$ respectively.
\end{itemize}
Furthermore,
denoting by $\apr : A \otimes A \to A$ the gca product on $A$,
by $\theta : V \otimes V \to V$
the binary product of
$\delta^1$ as a $\cinf$-structure on $V[-2]$\footnote{%
The dual of the $V^\ast \to [V^\ast,V^\ast]$ component of $\delta^1$.
This need not be an associative product.}, one has
\begin{equation}\label{eq:3vert}
   \theta \;=\; \pmap \circ \apr \circ \imap^{\otimes 2}
\end{equation}
\end{itemize}
The construction of
this $\textnormal{dgBVa}^{\Qx}$, and of the
contraction \eqref{eq:defretract}, is functorial.
\end{lemma}
\begin{proof}
Denote by $\Delta^0, \Delta^{-1} : C \to C \otimes C$ the native coproducts on $C$
of degree $0$ and $-1$ respectively,
satisfying axioms dual to \eqref{eq:gax}\footnote{%
They are formal duals of the gca and gLa products on $GV^\ast$ respectively.}.
Abusing notation, denote by
\begin{align*}
\tt & \in \End^{-1}(C) & 
\qq & \in \End^1(C) & \curv & \in \Coder^0(C)\\
d_{\tt} & \in \End^{-1}(\Coder(C)^{\Hopf}) &
d_{\qq} & \in \End^1(\Coder(C)^{\Hopf})
\end{align*}
the duals of the objects in Section \ref{sec:ginfMAIN}
of the same name. They satisfy analogous identities,
e.g.~$\tt^2=\qq^2=0$.
By \eqref{eq:defheq2}, and this 
is our implicit definition of a $\text{co-dgBVa}^{\Qx}$ on $C$:
\begin{equation}\label{eq:d2r}
\begin{aligned}
\tfrac12 [\delta^1,\delta^1] & = -d_{\qq} \delta^1 \\
d_{\tt} \delta^1 + [\delta^1,\delta^{-1}] & = \curv - d_{\qq}\delta^{-1}\\
d_{\tt}\delta^{-1} + \tfrac12 [\delta^{-1},\delta^{-1}] & = 0
\end{aligned}
\end{equation}
The assumption that $\delta^{-1}$
only has a $V \to V$ component is only used much later in this proof.
Denote the native products on $A$ by juxtaposition and $[-,-]$
like in Section \ref{sec:recapnew};
later in this proof we also denote them by
 $\nabla^0, \nabla^{-1} : A \otimes A \to A$.
All are $\Hopf$-equivariant.
Recall that, ignoring degrees, $V \hookrightarrow C \hookrightarrow A$.
First define $d_A|_C$ and $h_A|_C$ by
\begin{equation}\label{eq:dac}
\begin{aligned}
d_A(c) & = \qq(c) + \delta^1(c) + (-1)^{c_1} c_1c_2 + (-1)^{z_1-1} [z_1,z_2]\\
h_A(c) & = \tt(c) + \delta^{-1}(c)
\end{aligned}
\end{equation}
for all $c \in C$ where, using Sweedler notation,
      $\Delta^{-1} c = c_1 \otimes c_2$ and
      $\Delta^0 c = z_1 \otimes z_2$.
This is consistent with $d_A$ and $h_A$ having degrees $1$ and $-1$ respectively,
thanks to the $-2$ degree shift in the definition of $A$.
Extend $d_A$ and $h_A$ uniquely by 
\begin{equation}\label{eq:daha}
\begin{aligned}
d_A(ab) & = d_A(a)b + (-1)^a a d_A(b) \\
d_A([a,b]) & =
(-1)^a(
\Qx(ab) - (\Qx a)b - a (\Qx b)
)
+ [d_A a,b] + (-1)^{a-1} [a,d_A b] \\
\rule{0pt}{12pt}
h_A(ab) & = (-1)^a [a,b] + h_A(a)b + (-1)^a a h_A(b) \\
h_A([a,b]) & = [h_A a,b] + (-1)^{a-1} [a,h_A b]
\end{aligned}
\end{equation}
for all $a,b \in A$. 
These extensions exist\footnote{%
Just like $\tt$ and $\qq$ in Section \ref{sec:ginfMAIN}
and Appendix \ref{app:proofs}.
This uses the 7-term identity \eqref{eq:7term}.}
and they are $\Hopf$-equivariant.
To check that one has a $\text{dgBVa}^{\Qx}$
as claimed in Lemma \ref{lemma:rcc},
only $d_A^2=h_A^2=0$ and $[d_A,h_A]=\Qx\one$ remain todo.
Using \eqref{eq:daha} one checks $d_A^2, h_A^2 \in \Der(A)$,
so they vanish if their restriction to $C$ vanishes;
and $[d_A,h_A]$ fails to
be a derivation precisely in such a way\footnote{%
Analogous to the last two equations in \eqref{eq:nnonder}.}
that it suffices to show that $[d_A,h_A]|_C = \Qx \one$.
Lengthy calculations using \eqref{eq:dac} and \eqref{eq:daha} yield
\begin{align*}
h_A^2|_C &= d_{\tt}\delta^{-1} + \tfrac12 [\delta^{-1},\delta^{-1}] \\
[d_A,h_A]|_C &= d_{\tt}\delta^1 + d_{\qq}\delta^{-1} + [\delta^1,\delta^{-1}] + \tt\qq+\qq\tt \\
d_A^2|_C &= d_{\qq} \delta^1 + \tfrac12 [\delta^1,\delta^1]
\end{align*}
where the first also uses $\tt^2=0$;
the second also uses
$\delta^{-1} \in \Coder^{-1}(C)^{\Hopf}$ and the dual of \eqref{eq:bvtt};
the third also uses $\qq^2=0$
and $\delta^1 \in \Coder^1(C)^{\Hopf}$ and the duals of \eqref{eq:gax}
and \eqref{eq:qq81}.
Now \eqref{eq:d2r} imply $d_A^2|_C=h_A^2|_C=0$
and $[d_A,h_A]|_C = \curv + \tt\qq+\qq\tt = \Qx \one$,
where the very last equality uses \eqref{eq:cnq}.
This concludes the construction of the $\text{dgBVa}^{\Qx}$.

\newcommand{\nm}{n_{\text{gca}}}
\newcommand{\nl}{n_{\text{gLa}}}
\newcommand{\td}{\mathcal{D}}
To discuss \eqref{eq:defretract},
first note that $\Delta^{-1}|_V =\Delta^0|_V = 0$,
because $C$ is without counit,
and $\qq|_V = \tt|_V=0$, so \eqref{eq:dac} gives $d_A|_V = \delta^1|_V$
and $h_A|_V = \delta^{-1}|_V$,
so $\imap:V[-2] \hookrightarrow A$
so $d_A \imap = \imap d$ and $h_A \imap = \imap h$ as claimed.
To continue, note the equivalent definitions
\begin{equation}\label{eq:dahaalt}
\begin{aligned}
d_A & = \qq_A + \td(\qq) + \td(\delta^1)
+ \td(\pm \nabla^0 \circ \Delta^{-1})
+ \td(\pm \nabla^{-1} \circ \Delta^0)\\
h_A & = \tt_A + \td(\tt) + \td(\delta^{-1})
\end{aligned}
\end{equation}
where $\tt_A \in \End^{-1}(A)$, $\qq_A \in \End^1(A)$
are defined like \eqref{eq:bvtt} and \eqref{eq:qq81} respectively;
where $\mathcal{D}: \Hom(C,A) \to \Der(A)$ extends maps to derivations;
$\nabla^0, \nabla^{-1} : A \otimes A \to A$ are the two products\footnote{%
In particular, $\nabla^0 = \apr$ are two notations for the same object.};
the $\pm$ are detailed in \eqref{eq:dac};
and $C \hookrightarrow A$ is implicit in a few places.

We also introduce auxiliary gradings.
For every Gerstenhaber word, denote by $\nm$ and $\nl$
the number of gca and gLa products that it uses;
this yields a well-defined $\N^2$-grading\footnote{%
$\N = \{0,1,2,3,\ldots\}$.} on every free Gerstenhaber algebra.
Dualization and sign reversal yields an $\N^2$-grading on
every cofree Gerstenhaber coalgebra\footnote{%
Dualization produces non-positive degrees,
the sign reversal is to get an $\N^2$-grading.}.
Note that $A$ is spanned by Gerstenhaber words of
co-Gerstenhaber words in $V$. To every such
word-of-words $w \in A$ we associate:
$(\nm^A,\nl^A)$ is the $\N^2$-degree of $w$ as a Gerstenhaber word,
and $(\nm^C,\nl^C)$ is the sum of the $\N^2$-degrees of 
the co-Gerstenhaber words that constitute $w$.
This yields a well-defined $\N^4$-grading on $A$.
Also introduce the following coarser gradings:
\[
p = \nl^C + \nm^A
\qquad\quad
q = \nm^C + \nl^A
\qquad\quad
r = p+q
\]
giving decompositions $A = \oplus_{p,q} A_{p,q} = \oplus_r A_r$.
The summands in \eqref{eq:dahaalt} do not increase $r$,
as shown in Table \ref{tab:degs}.
\begin{table}
\centering
\begin{tabular}{r|cccc|cc|c}
& $\nm^A$ & $\nl^A$ & $\nm^C$ & $\nl^C$ & $p$ & $q$ & $r$\\
\hline
$\qq_A$ & $1$ & $-1$ & $0$ & $0$ & $1$ & $-1$ & $0$\\
$\td(\qq)$ & $0$ & $0$ & $-1$ & $1$ & $1$ & $-1$ & $0$\\
$\td(\delta^1)$ & $0$ & $0$ & $\leq 0$ & $\leq 0$ & $\leq 0$ & $\leq 0$ & $\leq 0$\\
$\td(\pm \nabla^0 \circ \Delta^{-1})$ & $1$ & $0$ & $0$ & $-1$ & $0$ & $0$ & $0$\\
$\td(\pm \nabla^{-1} \circ \Delta^0)$ & $0$ & $1$ & $-1$ & $0$ & $0$ & $0$ & $0$\\
\hline
$\tt_A$ & $-1$ & $1$ & $0$ & $0$ & $-1$ & $1$ & $0$\\
$\td(\tt)$ & $0$ & $0$ & $1$ & $-1$ & $-1$ & $1$ & $0$\\
$\td(\delta^{-1})$ & $0$ & $0$ & $\leq 0$ & $\leq 0$ & $\leq 0$ & $\leq 0$ & $\leq 0$
\end{tabular}
\caption{Degrees of the summands in \eqref{eq:dahaalt}.
They do not increase $r$.
The entries for $\td(\delta^{-1})$ are actually all zero if
one uses the assumption that 
$\delta^{-1}$ has only a $V \to V$ component.
}\label{tab:degs}
\end{table}
Now consider the components that neither depend on $\delta^{\pm 1}$
nor on the $\Hopf$-module structure or $\Qx$,
only on $V$ itself. Namely, set
\begin{align*}
u & = \td(\pm \nabla^0 \circ \Delta^{-1}) + \td(\pm \nabla^{-1} \circ \Delta^0)
&
v & = \tt_A + \td(\tt)
\end{align*}
A byproduct of the above proof that we have a $\text{dgBVa}^{\Qx}$
is $u^2 = v^2 = uv+vu=0$.
Note that $u(A_r) \subset A_r$ and $v(A_r) \subset A_r$,
and decomposing
$A_r = \oplus_{p+q=r} A_{p,q}$ we get:
\begin{equation}\label{eq:uvdiag}
\vcenter{\hbox{%
\begin{tikzpicture}[node distance = 11mm and 13mm, auto]
  \node (a0) {$A_{r,0}$};
  \node (a1) [right=of a0] {$A_{r-1,1}$};
  \node (a2) [right=of a1] {$\ldots$};
  \node (a3) [right=of a2] {$A_{1,r-1}$};
  \node (a4) [right=of a3] {$A_{0,r}$};
  \draw[->] (a0) to node {$v$} (a1);
  \draw[->] (a1) to node {$v$} (a2);
  \draw[->] (a2) to node {$v$} (a3);
  \draw[->] (a3) to node {$v$} (a4);
  \draw[->] (a0) to[out=120,in=60,looseness=6] node[anchor=south] {$u$} (a0);
  \draw[->] (a1) to[out=120,in=60,looseness=6] node[anchor=south] {$u$} (a1);
  \draw[->] (a3) to[out=120,in=60,looseness=6] node[anchor=south] {$u$} (a3);
  \draw[->] (a4) to[out=120,in=60,looseness=6] node[anchor=south] {$u$} (a4);
\end{tikzpicture}}}
\end{equation}
To every graded vector space $V$ one can functorially
associate \eqref{eq:uvdiag};
it is for every $r$ a universal construction coming with the Gerstenhaber operad
that involves suitable tensoring with $V^{\otimes(r+1)}$.
The homology of $u|_{A_r}$ vanishes if $r>0$\footnote{%
Because the Gerstenhaber operad is Koszul \cite{gj,mkl,dtt},
and see \cite[Theorem 4.2.5]{gk}.}.
The homology of $v|_{A_r}$ vanishes if $r>0$\footnote{%
Let $\hh \in \End^1(C)$, $\uu \in \End^0(C)$
be dual to the objects in Section \ref{sec:gammax};
let $\hh_A \in \End^1(A)$, $\uu_A \in \End^0(A)$
be like in Section \ref{sec:gammax}.
Set $w = \hh_A + \td(\hh) \in \End^1(A)$.
Then $w(A_r) \subset A_r$,
$w^2=0$ and $wv+vw = \uu_A + \td(\uu)$.
The last operator commutes with $v$ and $w$
and is invertible on $A_r$ when $r>0$, using \eqref{eq:unn}.
}.
This implies that there exists a $g \in \End^{-1}(A)$
that satisfies\footnote{%
For $r>0$, this would follow from $H(u|_{A_r})=0$,
if it was not for the requirement $vg+gv=0$.
In \eqref{eq:uvdiag}, choose $g$ from left to right,
at each step consistent with $vg+gv=0$,
which works because $H(v|_{A_r})=0$.}
\begin{equation}\label{eq:gpr}
g(A_{p,q}) \subset A_{p,q}
\qquad
g|_{A_0} = 0
\qquad
g^2 = 0
\qquad
ug + gu = \one_{>0}
\qquad
vg + gv = 0
\end{equation}
where, by definition,
$\one_{>0}|_{A_r}=\one$ if $r>0$ and $\one_{>0}|_{A_0} = 0$.
This $g$ can be chosen functorially in $V$,
specifically so that $g$ is $\Hopf$-equivariant and
commutes with every derivation on $A$ that extends a coderivation on $C$
that extends an endomorphism on $V$.
In particular, for $d = \delta^1|_V$ and $h = \delta^{-1}|_V$,
here viewed as elements of $\Coder(C)$,
\begin{equation}\label{eq:zred}
\td(d)g + g\td(d) = 0
\qquad
\td(h)g + g\td(h) = 0
\end{equation}

By \eqref{eq:dahaalt}, \eqref{eq:gpr}, \eqref{eq:zred} and Table \ref{tab:degs},
the term $\one_{>0} - d_Ag - gd_A$ decreases $(r,q) \in \N^2$
ordered lexicographically\footnote{%
So it either decreases $r$, or it preserves $r$ and decreases $q$.},
so its restriction to $A_0 \oplus \ldots \oplus A_r$ is nilpotent.
Define\footnote{%
This can be interpreted as
a formula that implements back-substitution
to `invert' triangular matrices; or as a Neumann series; or
as an application of the homological perturbation lemma.}
\begin{equation}\label{eq:zr3}
    \zmap \;=\; \textstyle\sum_{n \geq 0} g (\one_{>0} - d_A g - g d_A)^n
\end{equation}
which, when applied to any given element of $A$,
has finitely many nonzero summands.
Set $\pi = \one - d_A \zmap - \zmap d_A$.
On $A/A_0$ we have $\one_{>0} = \one$ and,
considering \eqref{eq:zr3} on this quotient,
a standard calculation using $g^2=0$ implies $\pi = 0$ on $A/A_0$.
So $\image \pi \subset A_0$.
But $g|_{A_0}=0$, hence $\zmap|_{A_0}=0$, hence $\pi|_{A_0} = \one$.
Therefore, $\pi^2 = \pi$ with $\image \pi = A_0$.
Continuing, $g^2=0$ and
$\one_{>0} g = g \one_{>0}$
imply $\zmap g = 0$, hence $\zmap^2 = 0$.
This implies $\pi \zmap = \zmap \pi$.
Since $\zmap|_{A_0} = 0$ we have $\zmap \pi = 0$, hence $\pi \zmap = 0$.
Note $A_0 \simeq V$, and $\image \pi = A_0 = \image \imap$.
Let $\pmap$ be the composition of $\pi$
with the canonical $A \to V[-2]$,
so $\pmap\imap = \one$ and $\ker \pmap = \ker \pi$.

To prove \eqref{eq:3vert},
using $\image(\apr \circ \imap^{\otimes 2}) \subset A_0 \oplus A_{1,0}$
it suffices to consider
the subcomplex $(A_0 \oplus A_{1,0},d_A)$.
The diagonal blocks $A_0 \to A_0$ and $A_{1,0} \to A_{1,0}$
are described above,
the block $A_{1,0} \to A_0$ only uses $\theta$, a piece of $\delta^1$.
The construction of $\zmap$ implies \eqref{eq:3vert}.

Now recall the assumption
that the only nonzero component of $\delta^{-1}$ is its $V \to V$ component,
$\delta^{-1} = h$.
Then $h_A(A_r) \subset A_r$ hence $\one_{>0}h_A = h_A \one_{>0}$.
Using \eqref{eq:gpr} and \eqref{eq:zred} we have $g h_A + h_A g = 0$.
Using $d_A h_A + h_A d_A = \Qx \one$
and $\Qx g = g \Qx$, by $\Hopf$-equivariance of $g$,
it follows that $h_A$ commutes with $\one_{>0} - d_Ag - gd_A$.
Hence $\zmap h_A + h_A \zmap = 0$.
Now, using $\Qx \zmap = \zmap \Qx$, by $\Hopf$-equivariance of $\zmap$,
one gets $\pi h_A = h_A \pi$.
Hence $\pmap h_A = h \pmap$.
\qed\end{proof}

\section{Acknowledgements}
M.R.~is grateful to have received funding from ERC through grant agreement No.~669655.


\appendix

\section{Addendum to Section \ref{subsec:mod}}\label{app:proofs}

Here are the details for several claims in Section \ref{subsec:mod}.
\step
To see that $\qq$ is well-defined, recall the definition of $GV^\ast$
via terms and relations.
Set $C(a,b) = ab-(-1)^{ab}ba$, $A(a,b,c) = (ab)c-a(bc)$ and $X(a,b) = [a,b]+(-1)^{(a+1)(b+1)}[b,a]$
where here $a,b,c$ are terms before imposing any relations.
By direct calculation,
\begin{align*}
\qq(C(a,b)) & = 
C(\qq(a),b) + (-1)^a C(a,\qq(b)) \\
\qq(A(a,b,c))
& = 
A(\qq(a),b,c)
+ (-1)^a A(a,\qq(b),c)
+ (-1)^{a+b} A(a,b,\qq(c))\\
\qq(X(a,b)) & = X(\qq(a),b) + (-1)^{a-1} X(a,\qq(b))
+ (-1)^a \big(
\Qx C(a,b) - C(\Qx a,b) - C(a,\Qx b)
\big)
\end{align*}
and note that\footnote{%
Using the $\Hopf$-module structure on the space of terms, and cocommutativity.}
$\Qx C(a,b) = C(\Qx_1 a,\Qx_2 b)$
using Sweedler notation $\Delta \Qx = \Qx_1 \otimes \Qx_2$.
Given this and given how $\qq$
acts on products, see \eqref{eq:qq81},
one concludes that $\qq$ preserves the `ideal' generated by $C$, $A$, $X$.
From here on we work modulo these relations.
Set
\begin{align*}
P(a,b,c) & = [ab,c] - a[b,c] - (-1)^{ab} b[a,c]\\
J(a,b,c) & = [a,[b,c]] + (-1)^{(a+1)(b+c)} [b,[c,a]]
+ (-1)^{(c+1)(a+b)} [c,[a,b]]
\end{align*}
One finds
\begin{align*}
& \qq(P(a,b,c)) = P(\qq(a),b,c) + (-1)^a P(a,\qq(b),c) + (-1)^{a+b-1}
P(a,b,\qq(c))\\
& \quad + (-1)^{a+b} \big(
\Qx(abc)
- (\Qx(ab))c - a\Qx(bc) - (-1)^{ab} b\Qx(ac)
+ (\Qx a)bc +a (\Qx b)c + ab(\Qx c)
\big)
\end{align*}
and note that the last line vanishes by the 7-term identity \eqref{eq:7term}.
One concludes that $\qq$ also preserves the `ideal'
generated by $P$. We work modulo $P$ from here on.
Similarly,
\begin{align*}
\qq(J(a,b,c)) 
& = J(\qq(a),b,c) + (-1)^{a-1} J(a,\qq(b),c) + (-1)^{a+b-2} J(a,b,\qq(c))\\
& \qquad
+ (-1)^a (\Qx(a[b,c]) - (\Qx a)[b,c] - a \Qx[b,c])\\
& \qquad
- (-1)^{a+b} [a,\Qx(bc)-(\Qx b)c - b(\Qx c)] \\
& \qquad
+ (-1)^{a(b+c)+c} (\Qx(b[c,a])-(\Qx b)[c,a] - b \Qx[c,a]) \displaybreak[0] \\ 
& \qquad
- (-1)^{a(b+c)} [b,\Qx(ca)-(\Qx c)a - c (\Qx a)] \\ 
& \qquad
+ (-1)^{c(a+b)+a+b+c} (\Qx(c[a,b]) - (\Qx c)[a,b] - c\Qx[a,b]) \\ 
& \qquad
- (-1)^{c(a+b)+b+c} [c,\Qx(ab)-(\Qx a)b - a (\Qx b)] 
\end{align*}
To see that the terms involving $\Qx$ cancel,
replace $[c,\Qx(ab)] = [c,(\Qx_1a)(\Qx_2b)]$ and so forth,
then rewrite using $P=0$
and use \eqref{eq:7term}.
So, $\qq$ is well-defined on $GV^\ast$.
\step
Suppose now $\delta \in \Der(GV^\ast)^{\Hopf}$.
We show $d_{\qq}(\delta) \in \Der(GV^\ast)^{\Hopf}$.
By direct calculation,
\[
d_{\qq}(\delta)(ab)
 = (d_{\qq}(\delta)(a))b + (-1)^{(\delta+1)a} a d_{\qq}(\delta)(b)
\]
for all $a,b \in GV^\ast$
without using $\Hopf$-equivariance of $\delta$.
By direct calculation,
\begin{align*}
d_{\qq}(\delta)([a,b])
& = [d_{\qq}(\delta)(a),b] + (-1)^{(\delta+1)(a-1)} [a,d_{\qq}(\delta)(b)]\\
& \qquad
+ (-1)^{\delta + a}\big(
\Qx(\delta(a)b) - (\Qx\delta(a))b - \delta(a)(\Qx b)\\
& \qquad\qquad\qquad\qquad
+ (-1)^{\delta a}\Qx(a\delta(b))
    - (-1)^{\delta a} (\Qx a)\delta(b) - (-1)^{\delta a} a (\Qx \delta(b))\\
& \qquad\qquad\qquad\qquad
-  \delta(\Qx(ab)) + \delta((\Qx a)b) + \delta(a(\Qx b))
\big)
\end{align*}
To see that the $\Qx$ terms cancel,
replace $\Qx(\delta(a)b) = (\Qx_1 \delta(a)) (\Qx_2 b)
= (\delta(\Qx_1 a)) (\Qx_2 b)$ using $\Hopf$-equivariance of $\delta$;
similarly $\delta(\Qx(ab)) = \delta((\Qx_1 a)(\Qx_2b))$; and so on.
Since $\delta$
is a derivation, the $\Qx$ terms cancel.
Here \eqref{eq:7term} is not used.
\step
To check \eqref{eq:dddd},
set $\naivecurv = \tt\qq + \qq \tt \in \End^0(GV^\ast)$
and check by direct calculation that
\begin{equation}\label{eq:nnonder}
\begin{aligned}
\naivecurv|_{V^\ast} & = 0 \\
\naivecurv(ab) & = \Qx(ab)-(\Qx a)b-a(\Qx b) + \naivecurv(a)b + a\naivecurv(b)\\
\naivecurv([a,b]) & = \Qx[a,b]-[\Qx a,b]-[a,\Qx b] + [\naivecurv(a),b] + [a,\naivecurv(b)]
\end{aligned}
\end{equation}
for all $a,b \in GV^\ast$;
this uses the $\Hopf$-equivariance of $\tt$ and $\qq$
but not \eqref{eq:7term}. Hence
\begin{equation}\label{eq:cnq}
\naivecurv + \curv = \Qx \one
\end{equation}
Clearly $d_{\tt}d_{\qq}+d_{\qq}d_{\tt} = [\naivecurv,-]$.
Hence $[\naivecurv,-] = -[\curv,-]$ when restricted to $\Der(GV^\ast)^{\Hopf}$
since $[\Qx \one,-]$ annihilates $\Hopf$-equivariant derivations.
This gives \eqref{eq:dddd}.
To see $d_{\tt}(\curv) = 0$,
note that 
$d_{\tt}(\Qx \one) = 0$ since $\tt$ is $\Hopf$-equivariant,
and $d_{\tt}(\naivecurv) = \tt(\tt \qq + \qq \tt) - (\tt \qq + \qq \tt) \tt
= 0$ because $\tt^2 = 0$.
The proof that $d_{\qq}(\curv)=0$ goes the same way.

\section{Some remarks about the BV formalism}\label{app:bv}

The $\bv{\Box}$ algebra in Theorem \ref{theorem:mainh}
is not directly related, to the author's knowledge,
to the algebraic structures in the BV formalism
for the quantization of gauge theories.
Nevertheless, for comparison,
here are some remarks about the BV formalism \cite{ksm,schw}.
Costello \cite{costello}
uses the BV formalism to introduce the
dgca denoted $\ax$ in this paper.
\begin{enumerate}
\item Let $M$ be a finite-dimensional manifold.
The gca of polyvector fields $A = \Gamma(\wedge(TM))$
with the Schouten bracket $\{-,-\}$ is a Gerstenhaber algebra\footnote{%
The differential on $A$ is zero:
By the Hochschild-Kostant-Rosenberg theorem,
the graded vector space
$A$ is the homology of a subcomplex of the
complex of Hochschild cochains for the algebra $C^{\infty}(M)$.}.
\item Suppose $M$ has a volume form, $\mu$.
It induces an isomorphism between $A$ and the de Rham complex of $M$.
Define the degree $-1$ map
$\Delta_\mu = \mu^{-1} \circ \ddr \circ \mu : A \to A$.
Then $\Delta_\mu$ is second order on $A$; evidently $\Delta_\mu^2=0$;
and the associated Gerstenhaber bracket is the Schouten bracket
\cite[Theorem 2.8]{ksm}.
This turns $A$ into a BV algebra.
\item The supergeometry perspective is that $A$
is the `algebra of functions' on the graded supermanifold $T^{\ast}[-1]M$,
the odd cotangent bundle of $M$;
and the Schouten bracket is an odd Poisson,
in fact odd symplectic, structure \cite[Section 2.3]{ksm}.
\item For every $S \in A^\text{even}$,
the quantum master equation
$\Delta_\mu S + \tfrac12 \{S,S\} = 0$
is formally equivalent to $\Delta_\mu e^S = 0$,
and formally implies that the integration
of $e^S$ over a Lagrangian
submanifold $L \subset T^\ast[-1]M$,
a small perturbation of the
canonical Lagrangian submanifold $M \subset T^\ast[-1]M$,
is independent of $L$, see \cite{schw}.
\item 
If $M$ is a finite-dimensional vector space,
then $\mu$ is canonical up to normalization.
There is an analogous story where $A$ is replaced
by formal power series on the graded vector space $T^\ast[-1]M = M \oplus M^\ast[-1]$;
the supergeometry perspective is that this is
the `algebra of functions' on a formal neighborhood of the origin $0 \in T^\ast[-1]M$.
\item
Applying this to gauge theory
involves generalizations:
$M$ is an infinite-dimensional vector space
of `fields',
say the space of sections of a vector bundle on a $4$-manifold;
$M$ is itself graded,
say due to the addition of BRST ghosts.
Passing from $M$ to $T^{\ast}[-1]M$
is referred to as adding antifields.
In \cite[Chapter 6.2]{costello} the BV formalism
is used to introduce fields for first order Yang-Mills
and, based on this, to define the dgca $\ax$.
\end{enumerate}


\section{Code}\label{app:code}

This Wolfram Mathematica code is not part of the logic of this paper.
It can be used to construct $\theta_n$ for small arity $n$
and check all axioms up to that arity (cf.~Section \ref{sec:proofmain}).
\step
The code uses ordered bases, so
\begin{vv}
e01i23
\end{vv}
{} corresponds to $e^0 e^1 + i e^2 e^3 \in \Lambda^2_+$ in row $0$ in \eqref{eq:cpx},
\begin{vv}
s01i23
\end{vv}
{} corresponds to $e^0 e^1 + ie^2 e^3 \in \Lambda^2_+$ in row $1$.
The degree 
\begin{vv}
deg
\end{vv}
{} is that in $V^{-2} \oplus \ldots \oplus V^1$;
\begin{vv}
rdeg
\end{vv}
{} is the row degree.
The gca product and differential are given by pre-calculated arrays.
For instance $e^0e^1 + ie^2e^3$ in row 0 times $e^0e^1+ie^2e^3$
in row 1 yields $2i e^0e^1e^2e^3$ in row 1,
corresponding to the entry
\begin{vv}
{16,6,9}->2*I
\end{vv}
{} below. Further,
\begin{vv}
k[0],k[1],k[2],k[3]
\end{vv}
{} are symbols for the four partial derivatives,
or equivalently, for the four components of the momentum.
\begin{verbatim}
(* ordered bases and degrees *)
bas={one, e0,e1,e2,e3, e01i23,e02i31,e03i12,
     s01i23,s02i31,s03i12, s123,s023,s031,s012, s0123};
dual=Map[ToExpression["dual"<>ToString[#]]&,bas];
MapThread[(deg[#1]=deg[#5]=#3;deg[#2]=-#3;rdeg[#1]=rdeg[#5]=#4;pos[#2]=#5;ndeg[#2]=1;)&,
  {bas,dual,{-2,-1,-1,-1,-1,0,0,0,-1,-1,-1,0,0,0,0,1},
            {0,0,0,0,0,0,0,0,1,1,1,1,1,1,1,1},Range[1,16]}];

(* pre-calculated gca product and differential, and h *)
prod=SparseArray[{{1,1,1}->1,{2,1,2}->1,{2,2,1}->1,{3,1,3}->1,{3,3,1}->1,{4,1,4}->1,{4,4,1}->1,{5,1,5}->1,
  {5,5,1}->1,{6,1,6}->1,{6,2,3}->1/2,{6,3,2}->-1/2,{6,6,1}->1,{6,4,5}->-I/2,{6,5,4}->I/2,{7,1,7}->1,
  {7,2,4}->1/2,{7,4,2}->-1/2,{7,7,1}->1,{7,3,5}->I/2,{7,5,3}->-I/2,{8,1,8}->1,{8,2,5}->1/2,{8,5,2}->-1/2,
  {8,8,1}->1,{8,3,4}->-I/2,{8,4,3}->I/2,{9,1,9}->1,{9,9,1}->1,{10,1,10}->1,{10,10,1}->1,{11,1,11}->1,
  {11,11,1}->1,{12,1,12}->1,{12,3,9}->I,{12,4,10}->I,{12,5,11}->I,{12,9,3}->-I,{12,10,4}->-I,{12,11,5}->-I,
  {12,12,1}->1,{13,1,13}->1,{13,2,9}->I,{13,4,11}->-1,{13,5,10}->1,{13,9,2}->-I,{13,10,5}->-1,{13,11,4}->1,
  {13,13,1}->1,{14,1,14}->1,{14,2,10}->I,{14,3,11}->1,{14,5,9}->-1,{14,9,5}->1,{14,10,2}->-I,{14,11,3}->-1,
  {14,14,1}->1,{15,1,15}->1,{15,2,11}->I,{15,3,10}->-1,{15,4,9}->1,{15,9,4}->-1,{15,10,3}->1,{15,11,2}->-I,
  {15,15,1}->1,{16,1,16}->1,{16,2,12}->1,{16,3,13}->-1,{16,4,14}->-1,{16,5,15}->-1,{16,6,9}->2*I,
  {16,7,10}->2*I,{16,8,11}->2*I,{16,9,6}->2*I,{16,10,7}->2*I,{16,11,8}->2*I,{16,12,2}->1,{16,13,3}->-1,
  {16,14,4}->-1,{16,15,5}->-1,{16,16,1}->1},{16,16,16}];
dmat=SparseArray[{{2,1}->k[0],{3,1}->k[1],{4,1}->k[2],{5,1}->k[3],{6,9}->1,{6,4}->I*k[3]/2,
  {6,5}->-I*k[2]/2,{6,2}->-k[1]/2,{6,3}->k[0]/2,{7,10}->1,{7,3}->-I*k[3]/2,{7,2}->-k[2]/2,{7,5}->I*k[1]/2,
  {7,4}->k[0]/2,{8,11}->1,{8,2}->-k[3]/2,{8,3}->I*k[2]/2,{8,4}->-I*k[1]/2,{8,5}->k[0]/2,{12,11}->I*k[3],
  {12,10}->I*k[2],{12,9}->I*k[1],{13,10}->k[3],{13,11}->-k[2],{13,9}->I*k[0],{14,9}->-k[3],{14,11}->k[1],
  {14,10}->I*k[0],{15,9}->k[2],{15,10}->-k[1],{15,11}->I*k[0],{16,15}->-k[3],{16,14}->-k[2],{16,13}->-k[1],
  {16,12}->k[0]},{16,16}];
nvar:=With[{x=Unique[]},x/:HoldPattern[NumericQ[x]]=True;x];
hmat=Outer[If[deg[#1]==deg[#2]-1,Switch[rdeg[#1]-rdeg[#2],0,
           Sum[nvar*k[i],{i,0,3}],-1,nvar,_,0],0]&,bas,bas];
\end{verbatim}
\indent\indent
The ad-hoc implementation
of a free Gerstenhaber algebra with product
\begin{vv}
mul
\end{vv}
, bracket
\begin{vv}
lie
\end{vv}
{} below
is not guaranteed to reduce all expressions to a normal form\footnote{%
Unlike, say, Gr\"obner basis reduction for
the quotient of a polynomial algebra by an ideal.},
but suffices here.
\begin{verbatim}
(* degrees and linearity *)
setdeg[s_,offset_:0]:=( s[_?NumericQ]=0;
                        s[mul[a_,b__]]:=s[a]+s[mul[b]];
                        s[lie[a_,b_]]:=s[a]+s[b]+offset;
                        s[(_?NumericQ|k[_]|Power[k[_],_])*a_]:=s[a];);
lin[s_]:=( s[a___,(b_:1)*c_Plus,d___]:=Map[s[a,b*#,d]&,c];
           s[a___,n_?NumericQ*b_,c___]:=n*s[a,b,c];
           s[___,0,___]=0;);
setdeg[deg,-1]; setdeg[ndeg]; setdeg[rdeg]; lin[mul]; lin[lie];
set[pairs_]:=Scan[(Evaluate[First[#]]=Last[#])&,pairs];

(* ad-hoc implementation of free Gerstenhaber algebra *)
mul[a_]:=a; mul[a___,n_?NumericQ,b___]:=n*mul[a,b]; lie[_,_?NumericQ]=lie[_?NumericQ,_]=0;
HoldPattern[mul[a___,mul[b___],c___]]:=mul[a,b,c];
mul[a___,b_,c_,d___]/;Not[OrderedQ[{b,c}]]:=(-1)^(deg[b]*deg[c])*mul[a,c,b,d];
lie[a_,b_]/;Not[OrderedQ[{a,b}]]:=-(-1)^((deg[a]-1)*(deg[b]-1))*lie[b,a];
lie[mul[a_,b__],c_]:=mul[a,lie[mul[b],c]]+(-1)^(deg[a]*deg[mul[b]])*mul[b,lie[a,c]];
lie[a_,mul[b_,c__]]:=mul[lie[a,b],c]+(-1)^((deg[a]-1)*deg[b])*mul[b,lie[a,mul[c]]];
mul/:Power[k[i_],j_:1]*mul[a_,b__]:=k[i]^(j-1)*(mul[k[i]*a,b]+mul[a,k[i]*mul[b]])//Expand;
lie/:Power[k[i_],j_:1]*lie[a_,b_]:=k[i]^(j-1)*(lie[k[i]*a,b]+lie[a,k[i]*b])//Expand;
Scan[(lie[#,#]=0;mul[___,#,#,___]=0;)&,
  Flatten[Map[{1,k[0],k[1],k[2],k[3]}*#&,Select[dual,OddQ[deg[#]]&]]]];
J[a_,b_,c_]:=(lie[lie[a,b],c]+(-1)^((deg[a]-1)*(deg[b]+deg[c]))*lie[lie[b,c],a]
                             +(-1)^((deg[c]-1)*(deg[a]+deg[b]))*lie[lie[c,a],b]);
Outer[set[First[Solve[Thread[J[##]==0]]]]&,dual,dual,dual]; (* Jacobi identities *)
\end{verbatim}
\indent\indent
Below is an implementation of Gerstenhaber derivations;
the operations $\tt$, $\qq$, $\hh$ and $d_{\tt}$, $d_{\qq}$, $\Gamma$ and the curvature $\curv$;
the operations $\nu_k$, $\mu_k$ in terms of $\theta_k$ where at the moment
$\theta_k$ is given by an ansatz with undetermined coefficients; and the axioms $A_k$, $B_k$, $C_k$.
\begin{verbatim}
(* Gerstenhaber derivation associated to f in the default case g=h=0 *)
op[ddeg_,dndeg_,f_,g_:(0&),h_:(0&)]:=Module[{s},
  lin[s]; s[_?NumericQ]=0; Scan[(s[#]:=s[#]=Expand[f[#]];)&,dual];
  s[mul[a_,b__]]:=(-1)^deg[a]*g[a,mul[b]]+mul[s[a],b]+(-1)^(ddeg*deg[a])*mul[a,s[mul[b]]];
  s[lie[a_,b_]]:=(-1)^deg[a]*h[a,b]+lie[s[a],b]+(-1)^(ddeg*(deg[a]-1))*lie[a,s[b]];
  s/:deg[s]=ddeg; s/:ndeg[s]=dndeg;
  s[Power[k[i_],j_:1]*a_]:=k[i]^j*s[a];
  s/:(n_?NumericQ)*s:=op[ddeg,dndeg,n*f[#]&,n*g[#]&,n*h[#]&];
  s];
comm[x_,y_]:=op[deg[x]+deg[y],ndeg[x]+ndeg[y],x[y[#]]-(-1)^(deg[x]*deg[y])*y[x[#]]&];

(* alpha, beta, gamma, curvature, etc *)
Q=-k[0]^2+k[1]^2+k[2]^2+k[3]^2;
alpha=op[-1,0,0&,lie,0&];
beta=op[1,0,0&,0&,Q*mul[#1,#2]-mul[Q*#1,#2]-mul[#1,Q*#2]&];
gamma=op[1,0,0&,0&,ndeg[#1]*ndeg[#2]*mul[#1,#2]&];
da=comm[alpha,#]&; db=comm[beta,#]&; K=op[0,0,Q*#&];
GAMMA[x_]:=With[{n=ndeg[x]},op[deg[x]+1,n,If[n==0,0,2/((n+1)*n)]*gamma[x[#]]&]];

(* d, h and gca product *)
d=op[1,0,dmat[[pos[#]]].dual&]; h=op[-1,0,hmat[[pos[#]]].dual&];
theta[2]=op[2,1,Flatten[prod[[pos[#]]]].Flatten[Outer[mul,dual,dual]]&];

(* ansatz for theta_n, n>=3: graded symmetric, deg=2 and rdeg=0 *)
incr[n_]:=incr[n]=Select[Select[Tuples[Range[1,16],n],LessEqual@@#&],
                         DuplicateFreeQ[Cases[#,Alternatives@@Select[Range[1,16],OddQ[deg[#]]&]]]&];
incr2[n_]:=incr2[n]=Table[Select[incr[n],And[Total[Map[deg,#]]+2==deg[i],
                                             Total[Map[rdeg,#]]==rdeg[i]]&],{i,1,16}];
theta[n_]:=theta[n]=op[2,n-1,Total[Map[nvar*(mul@@dual[[#]])&,incr2[n][[pos[#]]]]]&];
nu[k_]:=nu[k]=da[theta[k]]; mu[k_]:=mu[k]=-GAMMA[comm[h,nu[k]]];

(* auxiliary derivations, and axioms A,B,C decomposed by degree *)
dd=comm[d,d]; dh=comm[d,h]; hh=comm[h,h];
dn[k_]:=dn[k]=comm[d,nu[k]]; dm[k_]:=dm[k]=comm[d,mu[k]];
hn[k_]:=hn[k]=comm[h,nu[k]]; hm[k_]:=hm[k]=comm[h,mu[k]];
mm[m_,n_]:=mm[m,n]=comm[mu[m],mu[n]]; mn[m_,n_]:=mn[m,n]=comm[mu[m],nu[n]];
nn[m_,n_]:=nn[m,n]=comm[nu[m],nu[n]];
dam[k_]:=dam[k]=da[mu[k]]; dan[k_]:=dan[k]=da[nu[k]]; dbn[k_]:=dbn[k]=db[nu[k]];
Unprotect[C];
A[1,2]=dd/2; A[1,0]=op[0,0,-K[#]+dh[#]&]; A[1,-2]=hh/2;
A[k_,2]:=A[k,2]=op[2,k-1,dbn[k][#]+dm[k][#]+1/2*Sum[mm[k+1-n,n][#],{n,2,k-1}]&];
A[k_,0]:=A[k,0]=op[0,k-1,hm[k][#]&];
B[k_,2]:=B[k,2]=op[2,k-1,dn[k][#]+Sum[mn[k+1-n,n][#],{n,2,k-1}]&];
B[k_,0]:=B[k,0]=op[0,k-1,dam[k][#]+hn[k][#]&];
C[k_,2]:=C[k,2]=op[2,k-1,1/2*Sum[nn[k+1-n,n][#],{n,2,k-1}]&];
C[k_,0]:=C[k,0]=op[0,k-1,dan[k][#]&];
q[k_]:=q[k]=GAMMA[op[2,k-1,Sum[mn[k+1-n,n][#],{n,2,k-1}]&]];
dt[k_]:=dt[k]=comm[d,theta[k]];
b[k_]:=b[k]=op[3,k-1,-dt[k][#]+q[k][#]&];
\end{verbatim}
\indent\indent
Here an $h$ is fixed such that $A_1=0$, note that this $h$ is not unique;
then the coefficients of $\theta_n$ are computed for small $n$;
finally all the axioms up to that arity are checked.
\begin{verbatim}
(* fix h, not unique *)
coeffs[expr_]:=Map[Last,Flatten[CoefficientRules[expr,Join[Table[k[i],{i,0,3}],dual]]]];
set[First[Solve[Thread[coeffs[Map[A[1,0],dual]]]==0]]];
set[With[{X=coeffs[Map[A[1,-2],dual]]},First[FindInstance[Thread[X==0],Variables[X]]]]];

(* compute theta_n for n=3...nmax *)
nmax=3;
sol[expr_]:=Module[{x},CoefficientList[expr/.MapIndexed[(#1->x^First[#2])&,
                  DeleteDuplicates[Extract[expr,Position[expr,_mul]]]],x]
                    //DeleteDuplicates//First[Solve[Thread[#==0]]]&]; 
Do[Scan[set[sol[b[n][#]]]&,dual],{n,3,nmax}];

(* check all axioms up to arity nmax *)
Table[Map[A[1,i],dual],{i,{-2,0,2}}]//Flatten//DeleteDuplicates
Table[Map[A[n,i],dual],{n,2,nmax},{i,{0,2}}]//Flatten//DeleteDuplicates
Table[Map[B[n,i],dual],{n,2,nmax},{i,{0,2}}]//Flatten//DeleteDuplicates
Table[Map[C[n,i],dual],{n,3,nmax},{i,{0,2}}]//Flatten//DeleteDuplicates
\end{verbatim}

\end{document}